\documentclass{lmcs}
\usepackage{amssymb}
\usepackage{amsmath}
\usepackage[mathscr]{eucal}
\usepackage{mathrsfs}
\usepackage{xspace}
\usepackage[latin1]{inputenc}
\usepackage[english]{babel}
\usepackage[all,2cell]{xy}
\usepackage{xcolor}
\usepackage{graphicx}
\usepackage{adjustbox}
\usepackage{wasysym}
\usepackage{stmaryrd}
\usepackage{marvosym}
\UseAllTwocells

\newcommand{\concat}{\ensuremath{\text{\Large $\curlywedge$}}}
\newcommand{\bb}[1]{\ensuremath{\lvert #1 \rvert}}

\newcommand{\sbs}[2]{\left(\!\begin{smallmatrix}#1\\ #2\end{smallmatrix}\!\right)}

\numberwithin{equation}{section}

\begin{document}

\title{A Kleene theorem for free many-sorted algebras}

\author[L.~Gong]{L\"{u} Gong
\lmcsorcid{0000-0002-2819-2820}}[a]
\address{Nantong University, School of Mathematics and Statistics.}
\email{lieningzai1917@126.com}

\author[R.~Ruiz Mora]{Ra\'{u}l Ruiz Mora
\lmcsorcid{0009-0004-3739-7811}}[b]
\address{Universitat de Val\`{e}ncia, Departament de Matem\`{a}tiques.}
\email{raul.ruiz-mora@uv.es}
\email{enric.cosme@uv.es\,\Letter}

\author[N.~Sanmart\'{i}n Vich]{Nofre Sanmart\'{i}n Vich
\lmcsorcid{0000-0003-2606-9295}}[c]
\address{Analog Devices Inc.\ and Universitat Jaume I, Visual Engineering Group.}
\email{nofre.sanmartinvich@analog.com}

\author[E.~Cosme Ll\'{o}pez]{Enric Cosme Ll\'{o}pez
\lmcsorcid{0000-0001-8618-7328}}[b,a]

\subjclass{Theory of computation~\textemdash~Formal languages and automata theory~\textemdash~Algebraic language theory; Theory of computation~\textemdash~Formal languages and automata theory~\textemdash~Tree languages; Theory of computation~\textemdash~Formal languages and automata theory~\textemdash~Regular languages}
\keywords{Free many-sorted algebra, languages, recognizable, regular, Kleene theorem}

\begin{abstract}
In this work, we generalize Kleene's theorem from free single-sorted algebras to free many-sorted algebras. Our main result establishes that, under appropriate finitary assumptions, a language of a given sort in a free many-sorted algebra is recognizable if and only if it is regular.
\end{abstract}

\maketitle

\section{Introduction}
Kleene's theorem~\cite{Kle56, MY60} is among the foundational results of formal language theory. It establishes, for the regular languages, the coincidence of two a priori distinct descriptions: those denoted by the regular expressions, built from the alphabet symbols by means of union, concatenation, and the Kleene star, and those recognized by the finite automata. Equivalent algebraic characterizations were subsequently obtained by Myhill~\cite{Myh57}, in terms of congruences of finite index, and by Nerode~\cite{Ner58}, in terms of right-invariant equivalences of finite index, and were recast in automata-theoretic terms by Rabin and Scott~\cite{RS59}.

A first step towards abstraction was taken by Mezei and Wright~\cite{MW67}, who detached the notion of recognizability from that of string: a subset of an arbitrary single-sorted abstract algebra is said to be recognizable if it is the inverse image, under a homomorphism, of a subset of a finite algebra or, equivalently, if it is saturated by a congruence of finite index. When strings are replaced by terms, and hence by finite ordered trees labelled by function symbols, one is led to the theory of recognizable tree languages and finite tree automata, developed by G\'{e}cseg and Steinby~\cite{GS84, GS97}. In that setting the Kleene--Myhill--Nerode correspondence reappears: the recognizable tree languages are precisely those denoted by the regular tree expressions, built from the finite languages by means of union, tree concatenation, and an iteration operation that plays the role of the Kleene star.

On the algebraic side, the correspondence has been extended in several directions. Sch\"{u}tzenberger~\cite{Sch61} replaced languages by formal power series over a semiring, thereby identifying the series recognized by the weighted automata with the rational ones---the Kleene--Sch\"{u}tzenberger theorem---a quantitative theory whose systematic treatment is due to Berstel and Reutenauer~\cite{BR88}. A purely equational account was given by Conway~\cite{Con71}, whose regular algebras axiomatize the star operation and thus place iteration at the centre of the theory. In all of these settings the underlying structure is, as in the classical case, free or nearly so; and this is no accident, since already for the recognizable subsets of an arbitrary monoid regularity and recognizability diverge, as the counterexamples in the non-free case show~\cite{Eil74}. It is precisely this circumstance that singles out the free many-sorted algebras as the natural setting in which to seek an exact Kleene correspondence.

Beyond the algebraic setting, Kleene-type theorems have been pursued in a variety of computational frameworks, each one seeking the operations that exactly capture the finite-state behaviour. Cruchten~\cite{Cru25} obtained such a theorem for the lasso languages, and hence for the $\omega$-languages, where the ultimately periodic words are presented as pairs and an algebraic theory of rational lasso expressions is developed to match the saturated lasso automata; Bonsangue, Rutten, and Silva~\cite{BRS09} obtained one for the polynomial coalgebras, by means of a uniform language of expressions that subsumes several automata models at once; and Fahrenberg, Johansen, Struth, and Ziemia\'{n}ski~\cite{FJSZ24} obtained one for the higher-dimensional automata, where concurrency and choice give rise to regular expressions over pomset-like structures.

The origin of the present work lies at the intersection of these two traditions, the algebraic and the automata-theoretic, in the field of many-sorted universal algebra. Many of the structures that arise in computation are inherently many-sorted: from the typed lambda calculi and the structured programming languages to the models of concurrent processes and the XML schemas, the data and the operations upon them are organized into distinct types, or sorts. The framework is that of the many-sorted algebras~\cite{Mat76, Mat79}, in which a signature assigns to each function symbol the sorts of its arguments and that of its value, the variables themselves carry sorts, and a language is no longer a single set of terms but an $S$-sorted family of such sets. Here recognizability is witnessed by the homomorphisms onto the finite many-sorted algebras, and the operations out of which the regular expressions are built, in particular substitution and iteration, are performed sortwise.

We build on~\cite{CVCL20}, in which the preservation, and where applicable the reflection, of recognizability was established for several operators on the terms over a free many-sorted algebra---substitution, iteration, quotient, inverse image under a tree homomorphism, and direct image under a linear one---by way of the characterization in terms of finite-index congruences. The question of the regular languages, and of whether they exhaust the recognizable ones, was there left open. It is this question that the present work answers: we introduce the many-sorted regular expressions, we organize the regular languages into a regular algebra, and we prove that the two notions coincide.

The inclusion of the regular languages among the recognizable ones follows from the closure properties established in~\cite{CVCL20}. The converse inclusion is the harder one, and it is here that the substance of the work lies. Rather than appealing to congruences, we give for it a constructive and syntactic argument: from a homomorphism recognizing a language we build, by induction on the structure of the terms, a regular expression that denotes it. This generalizes Lemma~2.5.7 of G\'{e}cseg and Steinby~\cite{GS84}, which is itself an adaptation, to trees, of the proof given by McNaughton and Yamada~\cite{MY60} for the free monoid. The passage to the many-sorted setting is where the difficulty resides: whereas in the single-sorted case the elimination of states proceeds along a single dimension, here each step is indexed by a sort, and the inductive invariant---the family of languages introduced in the proof of the main result, in Section~4---must record, for every sort, how many iteration steps remain admissible.

Our main theorem states that, for every finite set of sorts $S$, every finite $S$-sorted signature $\Sigma$, every finite $S$-sorted set $X$, and every sort $s\in S$, the set $\mathrm{Rec}_{s}(\mathbf{T}_{\Sigma}(X))$ of $s$-recognizable languages in the free many-sorted algebra $\mathbf{T}_{\Sigma}(X)$ coincides with the set $\mathrm{Reg}_{s}(\mathbf{T}_{\Sigma}(X))$ of $s$-regular languages, that is, those denoted by a regular expression over $\Sigma$ built from the operations of $\Sigma$ together with the many-sorted operations of empty language, union, substitution, and iteration.

The paper is organized as follows. Section~2 provides the background on many-sorted sets and algebras and introduces the notion of recognizability, including recognizability at a given sort. Section~3 introduces the free many-sorted algebra together with the operations of substitution and iteration at a sort, which are shown to be compatible with recognizability. Section~4 defines the regular languages at a sort and establishes the main result, the many-sorted Kleene theorem.

Our underlying set theory is $\mathbf{ZFSk}$, Zermelo-Fraenkel-Skolem set theory (also known as $\mathbf{ZFC}$, i.e., Zermelo-Fraenkel set theory with the axiom of choice) plus the existence of a Grothendieck universe $\boldsymbol{\mathcal{U}}$, fixed once and for all (see~\cite{Mac98}, pp.~21--24). We recall that the elements of $\boldsymbol{\mathcal{U}}$ are called $\boldsymbol{\mathcal{U}}$-small sets and the subsets of $\boldsymbol{\mathcal{U}}$ are called $\boldsymbol{\mathcal{U}}$-large sets or classes. Moreover, from now on $\mathsf{Set}$ stands for the category of sets, i.e., the category whose set of objects is $\boldsymbol{\mathcal{U}}$ and whose set of morphisms is the set of all mappings between $\boldsymbol{\mathcal{U}}$-small sets.

In all that follows we use standard concepts and constructions from category theory, see e.g., \cite{HS73,Mac98,Rie16}, universal algebra, see e.g., \cite{Ber15,BS81,GTW78,Gra08,Mat76,Wec92}, and set theory, see e.g., \cite{Bou70,End77}. Nevertheless, regarding set theory, we have adopted the following conventions. An \emph{ordinal} $\alpha$ is a transitive set that is well-ordered by $\in$, thus $\alpha = \{\,\beta\mid \beta\in \alpha\,\}$. The first transfinite ordinal $\omega_{0}$ will be denoted by $\mathbb{N}$, which is the set of all \emph{natural numbers}, and, from what we have just said about the ordinals, for every $n\in \mathbb{N}$, $n = \{0, \ldots,n-1\}$. 

We will denote by $\mathrm{Fnc}(A,B)$ the set of all functions from $A$ to $B$. We recall that a function from $A$ to $B$ is a subset $F$ of $A\times B$ that satisfies the functional condition, i.e., such that for every $x\in A$, there exists a unique $y\in B$ such that $(x,y)\in F$. A function $F$ from $A$ to $B$ is usually denoted by $(F_{x})_{x\in A}$. We will denote by $\mathrm{Hom}(A,B)$ (and, sometimes, also by $B^{A}$) the set of all mappings from $A$ to $B$. We recall that a mapping from $A$ to $B$ is an ordered triple $f = (A,F,B)$, denoted by $f\colon A\longrightarrow B$, in which $F$ is a function from $A$ to $B$. Therefore $\mathrm{Hom}(A,B) = \{A\}\times \mathrm{Fnc}(A,B)\times\{B\}$. 

\section{Preliminaries}

In this section we collect the basic facts, mostly without proofs, about many-sorted sets, many-sorted algebras, and recognizability for arbitrary many-sorted algebras, that we will need in this work.

\subsection{Many-sorted sets}
\begin{asm}
From now on $S$ stands for a set of sorts in $\boldsymbol{\mathcal{U}}$, fixed once and for all.
\end{asm}

\begin{defi}
An $S$-\emph{sorted set} is a mapping $A = (A_{s})_{s\in S}$ from $S$ to $\boldsymbol{\mathcal{U}}$. If $A$ and $B$ are $S$-sorted sets, an $S$-\emph{sorted mapping from} $A$ \emph{to} $B$ is an $S$-indexed family $f = (f_{s})_{s\in S}$, where, for every $s$ in $S$, $f_{s}$ is a mapping from $A_{s}$ to $B_{s}$. Thus, an $S$-sorted mapping from $A$ to $B$ is an element of $\prod_{s\in S}\mathrm{Hom}(A_{s}, B_{s})$. We will denote by $\mathrm{Hom}(A,B)$ the set of all $S$-sorted mappings from $A$ to $B$. An $S$-sorted mapping from $A$ to $B$ will simply be denoted by $f\colon A \longrightarrow B$.
An $S$-sorted mapping $f\colon A \longrightarrow B$ will be called \emph{injective} if, for every $s\in S$, $f_{s}\colon A_{s}\longrightarrow B_{s}$ is injective. Similarly, it will be called \emph{surjective} if, for every $s\in S$, $f_{s}\colon A_{s}\longrightarrow B_{s}$ is surjective.

$S$-sorted sets and $S$-sorted mappings form a category which we will denote henceforth by $\mathsf{Set}^{S}$. 
\end{defi}

\begin{defi}
Let $I$ be a set in $\boldsymbol{\mathcal{U}}$ and $(A^{i})_{i\in I}$ an $I$-indexed family of $S$-sorted sets. Then the \emph{product} of $(A^{i})_{i\in I}$, denoted by $\prod_{i\in I}A^{i}$, is the $S$-sorted set defined, for every $s\in S$, as $\left(\prod\nolimits_{i\in I}A^{i}\right)_{s} = \prod\nolimits_{i\in I}A^{i}_{s}$, where
$$
\textstyle
\prod_{i\in I}A^{i}_{s} = \left\{(a_{i})_{i\in I}\in\mathrm{Fnc}(I,\bigcup_{i\in I}A^{i}_{s})\mid \forall\,i\in I\,(a_{i}\in A^{i}_{s})\right\}.
$$
For every $i\in I$, the \emph{$i$-th canonical projection}, $\mathrm{pr}^{i} = (\mathrm{pr}^{i}_{s})_{s\in S}$, is the $S$-sorted mapping from $\prod_{i\in I}A^{i}$ to $A^{i}$ that, for every $s\in S$, sends $(a_{i})_{i\in I}$ in $\prod_{i\in I}A^{i}_{s}$ to $a_{i}$ in $A^{i}_{s}$. 

The \emph{coproduct} of $(A^{i})_{i\in I}$, denoted by $\coprod_{i\in I}A^{i}$, is the $S$-sorted set defined, for every $s\in S$, as $\left(\coprod\nolimits_{i\in I}A^{i}\right)_{s} = \coprod\nolimits_{i\in I}A^{i}_{s}$, where
$$
\textstyle
\coprod_{i\in I}A^{i}_{s} = \bigcup_{i\in I}(A^{i}_{s}\times\{i\}).
$$
For every $i\in I$, the \emph{$i$-th canonical injection}, $\mathrm{in}^{i}$, is the $S$-sorted mapping from $A^{i}$ to $\coprod_{i\in I}A^{i}$ that, for every $s\in S$, sends $a$ in $A^{i}_{s}$ to $(a,i)$ in $\coprod_{i\in I}A^{i}_{s}$. 

The remaining set-theoretic operations on $S$-sorted sets: $\times$ (binary product), $\amalg$ (binary coproduct), $\bigcup$ (union), $\cup$ (binary union), $\bigcap$ (intersection), $\cap$ (binary intersection), $-$ (difference), and $\complement_{A}$ (complement of an $S$-sorted set in a fixed $S$-sorted $A$), are defined in a similar way, i.e., componentwise.
\end{defi}

\begin{defi}
We will denote by $1^{S}$ the (standard) final $S$-sorted set of $\mathsf{Set}^{S}$, which is $1^{S} = (1)_{s\in S}$, and by $\varnothing^{S}$ the initial $S$-sorted set, which is $\varnothing^{S} = (\varnothing)_{s\in S}$. We shall abbreviate $1^{S}$ to $1$ and $\varnothing^{S}$ to $\varnothing$ when this is unlikely to cause confusion.
\end{defi}

\begin{defi}
If $A$ and $X$ are $S$-sorted sets, then we will say that $X$ is a \emph{subset} of $A$, denoted by $X\subseteq A$, if, for every $s\in S$, $X_{s}\subseteq A_{s}$. We will denote by $\mathrm{Sub}(A)$ the set of all $S$-sorted sets $X$ such that $X\subseteq A$.
\end{defi}

\begin{defi}
Let $\delta$ be the mapping from $S\times \boldsymbol{\mathcal{U}}$ to $\boldsymbol{\mathcal{U}}^{S}$ that sends $(t,X)$ in $S\times \boldsymbol{\mathcal{U}}$ to the $S$-sorted set $\delta^{t,X} = (\delta^{t,X}_{s})_{s\in S}$ defined, for every $s\in S$, as follows: 
\begin{equation*}
\delta^{t,X}_{s} =
\begin{cases}
X, &\text{if $s = t$;}\\
\varnothing, & \text{if $s \neq t$.}
\end{cases}
\end{equation*}
We will call the value of $\delta$ at $(t,X)$ the \emph{delta of Kronecker associated to} $(t,X)$. If $X = \{x\}$, then, for simplicity of notation, we will write $\delta^{t,x}$ instead of $\delta^{t,\{x\}}$. Moreover, for a sort $t$ in $S$, $\delta^{t,1}$, the delta of Kronecker associated to $(t,1)$, will be denoted by $\delta^{t}$ and called \emph{delta of Kronecker}.
\end{defi}

We next define for an $S$-sorted mapping the associated mappings of direct and inverse image formation.

\begin{defi}
Let $f\colon A\longrightarrow B$ be an $S$-sorted mapping. Then the mapping 
$f[\cdot]\colon\mathrm{Sub}(A)\to\mathrm{Sub}(B)$,
of $f$-\emph{direct image formation}, sends $X \in \mathrm{Sub}(A)$ to $f[X] = (f_{s}[X_{s}])_{s\in S} \in \mathrm{Sub}(B)$, and the mapping 
$f^{-1}[\cdot]\colon\mathrm{Sub}(B)\longrightarrow\mathrm{Sub}(A)$,
of $f$-\emph{inverse image formation}, sends $Y \in \mathrm{Sub}(B)$ to $f^{-1}[Y] = (f_{s}^{-1}[Y_{s}] )_{s\in S} \in \mathrm{Sub}(A)$.
\end{defi}

We next define, for an $S$-sorted set $A$, the power object of $A$ in $\mathsf{Set}^{S}$. 

\begin{defi}
For an $S$-sorted set $A$, the $S$-sorted set $(\mathrm{Sub}(A_{s}))_{s\in S}$, usually denoted by $A^{\wp}$, is the \emph{power} set of $A$ in $\mathsf{Set}^{S}$. With the notation $A^{\wp}$, we have, for $s\in S$, that $A^{\wp}_{s} = \mathrm{Sub}(A_{s})$. 
\end{defi}

\begin{rem}
One should take great care not to confuse, for an $S$-sorted set $A$, $\mathrm{Sub}(A)$, which is a \emph{set}, i.e., an object of $\mathsf{Set}$---naturally isomorphic to the set $\prod_{s\in S}2^{A_{s}}$---, and $A^{\wp}$, which is an $S$-\emph{sorted set}, i.e., an object of $\mathsf{Set}^{S}$---naturally isomorphic to the $S$-sorted set $(2^{A_{s}})_{s\in S}$. However, it is clear that there exists a natural isomorphism between $\mathrm{Sub}(A)$ and $\prod A^{\wp}$.
\end{rem}

We next state that every $S$-sorted mapping $f\colon X\longrightarrow Y^{\wp}$ has a distinguished extension up to an $S$-sorted mapping $f^{\mathsf{p}}\colon X^{\wp}\longrightarrow Y^{\wp}$.

\begin{prop}\label{PSubFunct}
Let $X$ and $Y$ be $S$-sorted sets and $f$ an $S$-sorted mapping from $X$ to $Y^{\wp}$. Then there exists a unique $S$-sorted mapping $f^{\mathsf{p}}$ from $X^{\wp}$ to $Y^{\wp}$ such that 
\begin{enumerate}
\item $f^{\mathsf{p}}$ is completely additive, i.e., for every $s\in S$ and every $\mathcal{L}\subseteq \mathrm{Sub}(X_{s})$, 
$$
\textstyle
f^{\mathsf{p}}_{s}(\bigcup\mathcal{L}) = \bigcup_{L\in \mathcal{L}}f^{\mathsf{p}}_{s}(L), \text{ and}
$$ 
\item $f^{\mathsf{p}}\circ \{\cdot\}^{X} = f$, where $\{\cdot\}^{X} = (\{\cdot\}^{X}_{s})_{s\in S}$ is the $S$-sorted mapping from $X$ to $X^{\wp}$ defined, for every $s\in S$, as:
$$
\{\cdot\}^{X}_{s} 
\left\lbrace
\begin{array}{ccc}
X_{s} & \longrightarrow &
X^{\wp}_{s}\\
x& \longmapsto &
\{x\}
\end{array}
\right.
$$
\end{enumerate}
\end{prop}

\begin{proof}
The $S$-sorted mapping $f^{\mathsf{p}}$ from $X^{\wp}$ to $Y^{\wp}$ that, for every $s\in S$, assigns to $L\subseteq X_{s}$ the set $\bigcup_{x\in L}f_{s}(x)\subseteq Y_{s}$ is completely additive and $f^{\mathsf{p}}\circ \{\cdot\}^{X} = f$. Moreover, $f^{\mathsf{p}}$ is clearly the unique $S$-sorted mapping from $X^{\wp}$ to $Y^{\wp}$ satisfying the aforementioned conditions.
\end{proof}

\begin{rem}
The mappings of type $f^{\mathsf{p}}$, which are the solution of a universal mapping problem, will be extensively used in Section~3.
\end{rem}

We next state, after defining a suitable subcategory $\mathsf{Set}^{S}_{\wp,\mathrm{ca}}$ of $\mathsf{Set}^{S}$, that the canonical embedding $\mathrm{In}_{\mathsf{Set}^{S}_{\wp,\mathrm{ca}},\mathsf{Set}^{S}}$ of $\mathsf{Set}^{S}_{\wp,\mathrm{ca}}$ into $\mathsf{Set}^{S}$ has a left adjoint.

\begin{defi}
We denote by $\mathsf{Set}^{S}_{\wp,\mathrm{ca}}$ the subcategory of $\mathsf{Set}^{S}$ defined as follows:
\begin{enumerate}
\item Objects: the $S$-sorted sets $X^{\wp}$, where $X\in\boldsymbol{\mathcal{U}}^{S}$. 
\item Morphisms from the object $X^{\wp}$ to the object $Y^{\wp}$: the completely additive $S$-sorted mappings from $X^{\wp}$ to $Y^{\wp}$.
\end{enumerate}
\end{defi}

From Proposition~\ref{PSubFunct} we immediately obtain the following corollary.

\begin{cor}\label{CSubFunct}
The functor $(\cdot)^{\wp}$ from $\mathsf{Set}^{S}$ to $\mathsf{Set}^{S}_{\wp,\mathrm{ca}}$
that sends
\begin{enumerate}
\item an $S$-sorted set $X$ to $X^{\wp}$ and
\item an $S$-sorted mapping $f$ from $X$ to $Y$ to the $S$-sorted mapping $(\{\cdot\}^{Y}\circ f)^{\mathsf{p}}$ from $X^{\wp}$ to $Y^{\wp}$, denoted by $f^{\wp}$ for short, 
\end{enumerate} 
is left adjoint for the functor $\mathrm{In}_{\mathsf{Set}^{S}_{\wp,\mathrm{ca}},\mathsf{Set}^{S}}$ from $\mathsf{Set}^{S}_{\wp,\mathrm{ca}}$ to $\mathsf{Set}^{S}$, that is, $$(\cdot)^{\wp}\dashv \mathrm{In}_{\mathsf{Set}^{S}_{\wp,\mathrm{ca}},\mathsf{Set}^{S}}.$$

Moreover, the $S$-sorted mappings $\{\cdot\}^{X}$, for $X$ in $\boldsymbol{\mathcal{U}}^{S}$, are the components of the unit of the adjunction: 
$$
\mathrm{Hom}_{\mathsf{Set}^{S}_{\wp,\mathrm{ca}}}(X^{\wp},Y^{\wp})\cong
\mathrm{Hom}_{\mathsf{Set}^{S}}(X,Y^{\wp}).
$$ 
\end{cor}

\begin{rem}
For an $S$-sorted mapping $f$ from $X$ to $Y$ the $S$-sorted mapping $f^{\wp}$ from $X^{\wp}$ to $Y^{\wp}$ sends, for every $s\in S$, $L\subseteq X_{s}$ to $f_{s}[L]\subseteq Y_{s}$, i.e., $f^{\wp} = (f_{s}[\cdot])_{s\in S}$. One should be careful not to confuse $f[\cdot]$, which is a \emph{mapping} from the \emph{set} $\mathrm{Sub}(X)$ to the \emph{set} $\mathrm{Sub}(Y)$, and $f^{\wp}$, which is an $S$-\emph{sorted mapping} from the $S$-\emph{sorted set} $X^{\wp} = (\mathrm{Sub}(X_{s}))_{s\in S}$ to the $S$-\emph{sorted set} $Y^{\wp} = (\mathrm{Sub}(Y_{s}))_{s\in S}$. However, it is evident that $f[\cdot]$ and $\prod f^{\wp}$ are \emph{essentially} the same mapping.
\end{rem}

We next introduce the notion of finite many-sorted set.

\begin{defi}
Let $A$ be an $S$-sorted set. Then the \emph{cardinal} of $A$, denoted by $\mathrm{card}(A)$, is $\mathrm{card}(\coprod A)$, i.e., the cardinal of the set $\coprod A = \bigcup_{s\in S}(A_{s}\times \{s\})$. An $S$-sorted set $A$ is \emph{finite} if $\mathrm{card}(A)<\aleph_{0}$. We will say that an $S$-sorted set $X$ is a \emph{finite} subset of $A$ if $X$ is finite and $X\subseteq A$. We will denote by $\mathrm{Sub}_{\mathrm{f}}(A)$ the set of all $S$-sorted sets $X$ in $\mathrm{Sub}(A)$ which are finite.
\end{defi}

\begin{defi}
Let $A$ be an $S$-sorted set. Then the \emph{support of} $A$, denoted by $\mathrm{supp}_{S}(A)$, is the set $\{s\in S\mid A_{s}\neq \varnothing\}$.
\end{defi}

\begin{rem}
An $S$-sorted set $A$ is finite if and only if $\mathrm{supp}_{S}(A)$ is finite and, for every $s\in \mathrm{supp}_{S}(A)$, $A_{s}$ is finite.
\end{rem}

\subsection{Free monoid}

We next define the concept of free monoid on a set and several notions associated with it that will be used afterwards to construct the free algebra on an $S$-sorted set and to define various substitution operators.

\begin{defi}
Let $A$ be a set. The \emph{free monoid on} $A$, denoted by $\mathbf{A}^{\star}$, is $(A^{\star},\curlywedge,\lambda)$, where $A^{\star}$, the set of all \emph{words on} $A$, is $\bigcup_{n\in\mathbb{N}}\mathrm{Hom}(n,A)$, the set of all mappings $\mathbf{a}\colon n\longrightarrow A$ from some $n\in \mathbb{N}$ to $A$. A word $\mathbf{a}\in A^{\star}$ is usually denoted as a sequence $(a_{i})_{i\in\lvert{\mathbf{a}}\rvert}$, where, for $i\in\lvert{\mathbf{a}}\rvert$, $a_{i}$ is the letter in $A$ satisfying $\mathbf{a}(i)=a_{i}$. Furthermore, $\curlywedge$, the \emph{concatenation} of words on $A$, is the binary operation on $A^{\star}$ which sends a pair of words $(\mathbf{a},\mathbf{b})$ on $A$ to the mapping $\mathbf{a}\curlywedge \mathbf{b}$ from $\lvert{\mathbf{a}}\rvert+\lvert{\mathbf{b}}\rvert$ to $A$, where $\lvert{\mathbf{a}}\rvert$ and $\lvert{\mathbf{b}}\rvert$ are the lengths ($\equiv$ domains) of the mappings $\mathbf{a}$ and $\mathbf{b}$, respectively, defined as follows:
$$
\mathbf{a}\curlywedge \mathbf{b}
\left\lbrace
\begin{array}{ccl}
{\lvert{\mathbf{a}}\rvert+\lvert{\mathbf{b}}\rvert}&\longrightarrow&{A}\\
{i}&\longmapsto&{
\begin{cases}
 a_{i}, & \text{if $0\leq i < \lvert{\mathbf{a}}\rvert$;}\\
 b_{i-\lvert{\mathbf{a}}\rvert}, & \text{if $\lvert{\mathbf{a}}\rvert\leq i < \lvert{\mathbf{a}}\rvert+\lvert{\mathbf{b}}\rvert$,}
\end{cases}
 }
\end{array}
\right.
$$
and $\lambda$, the \emph{empty word on} $A$, is the unique mapping from $\varnothing$ to $A$. We will denote by $\eta^{A}$ the mapping from $A$ to $A^{\star}$ that sends $a\in A$ to $(a)\in A^{\star}$, i.e., to the mapping $(a)\colon 1\longrightarrow A$ that sends $0$ to $a$. The ordered pair $(\mathbf{A}^{\star},\eta^{A})$ is a universal morphism from $A$ to the forgetful functor from the category $\mathsf{Mon}$, of monoids, to $\mathsf{Set}$.
\end{defi}

\begin{rem}
For a word $\mathbf{a}\in A^{\star}$, $\lvert{\mathbf{a}}\rvert$, the length of $\mathbf{a}$, is the value at $\mathbf{a}$ of the unique homomorphism $\lvert{\,\cdot\,}\rvert$ from $\mathbf{A}^{\star}$ to $(\mathbb{N},+,0)$, the additive monoid of the natural numbers, such that $\lvert{\,\cdot\,}\rvert\circ \eta^{A} = \kappa_{1}$, where $\kappa_{1}$ is the mapping from $A$ to $\mathbb{N}$ constantly $1$. Note that, for every $n\in \mathbb{N}$, $\lvert{\,\cdot\,}\rvert$ sends $\mathbf{a}\in \mathrm{Hom}(n,A)$ to $n$. Thus, for the family of mappings $(\kappa_{n})_{n\in \mathbb{N}}$, where, for every $n\in \mathbb{N}$, $\kappa_{n}$ is the mapping from $\mathrm{Hom}(n,A)$ to $\mathbb{N}$ constantly $n$, and by applying the universal property of the coproduct, we have $\lvert{\,\cdot\,}\rvert = [\kappa_{n}]_{n\in \mathbb{N}}$.
\end{rem}

\begin{defi}\label{DSubw}
Let $\mathbf{a}$ and $\mathbf{a}'$ be words in $A^{\star}$ and $k$, $l\in\lvert{\mathbf{a}}\rvert$ such that $k\leq l$. We will say that $\mathbf{a}'$ is a \emph{subword} of $\mathbf{a}$ \emph{beginning at position} $k$ and \emph{ending at position} $l$ if there are words $\mathbf{b}$ and $\mathbf{c}$ such that $\mathbf{a} = \mathbf{b}\curlywedge \mathbf{a}'\curlywedge \mathbf{c}$, $\lvert{\mathbf{b}}\rvert = k$, and $\lvert{\mathbf{c}}\rvert = (\lvert{\mathbf{a}}\rvert-1)-l$. A word $\mathbf{a}$ may have several subwords equal to $\mathbf{a}'$. In that case, the equation $\mathbf{a} = \mathbf{b}\curlywedge \mathbf{a}'\curlywedge \mathbf{c}$ has several solutions $(\mathbf{b},\mathbf{c})$.
If the pairs $(\mathbf{b}_{i},\mathbf{c}_{i})$ $(i\in n)$ are all solutions of $\mathbf{a} = \mathbf{b}\curlywedge \mathbf{a}'\curlywedge \mathbf{c}$ and if $\lvert{\mathbf{b}_{0}}\rvert< \lvert{\mathbf{b}_{1}}\rvert<\cdots<\lvert{\mathbf{b}_{n-1}}\rvert$, then each pair $(\mathbf{b}_{i},\mathbf{c}_{i})$ determines the $i$-th occurrence of $\mathbf{a}'$ in $\mathbf{a}$. The solution $(\mathbf{b},\mathbf{c})$ in which either $\mathbf{b}$ or $\mathbf{c}$ is $\lambda$ is not excluded.
Let $\mathbf{a}$ be a word in $A^{\star}$ and $a\in A$. We will say that \emph{$a$ occurs in} $\mathbf{a}$ if there are words $\mathbf{b}$, $\mathbf{c}$ in $A^{\star}$ such that $\mathbf{a} = \mathbf{b}\curlywedge(a)\curlywedge \mathbf{c}$. Note that $a$ occurs in $\mathbf{a}$ if and only if there exists an $i\in \lvert{\mathbf{a}}\rvert$ such that $\mathbf{a}(i) = a$. We will denote by $\lvert{\mathbf{a}}\rvert_{a}$ the natural number $$\mathrm{card}(\{i\in \lvert{\mathbf{a}}\rvert\mid \mathbf{a}(i) = a\}) = \mathrm{card}(\mathbf{a}^{-1}[\{a\}]),$$
that is, the number of occurrences of $a$ in $\mathbf{a}$. Moreover, we let $(i_{j})_{j\in\lvert{\mathbf{a}}\rvert_{a}}$ stand for the enumeration in ascending order of the occurrences of $a$ in $\mathbf{a}$. Thus, $(i_{j})_{j\in\lvert{\mathbf{a}}\rvert_{a}}$ is the order embedding of $(\lvert{\mathbf{a}}\rvert_{a},<)$ into $(\lvert{\mathbf{a}}\rvert,<)$ defined recursively as follows:
\[i_{0} = \mathrm{min}\{i\in \lvert{\mathbf{a}}\rvert\mid \mathbf{a}(i) = a\},\]
and, for $j\in [1,\lvert{\mathbf{a}}\rvert_{a}-1]$, 
$$
i_{j} = \mathrm{min}\{i\in \lvert{\mathbf{a}}\rvert-\{i_{0},\ldots,i_{j-1}\}\mid \mathbf{a}(i) = a\}.
$$
If, for every $j\in \lvert{\mathbf{a}}\rvert_{a}$, the pairs $(\mathbf{b}_{i_{j}},\mathbf{c}_{i_{j}})$ are all solutions of $\mathbf{a} = \mathbf{b}_{i_{j}}\curlywedge (a)\curlywedge \mathbf{c}_{i_{j}}$ and if $\lvert{\mathbf{b}_{i_{0}}}\rvert< \lvert{\mathbf{b}_{i_{1}}}\rvert<\cdots<\lvert{\mathbf{b}_{i_{\lvert{\mathbf{a}}\rvert_{a}-1}}}\rvert$, then each pair $(\mathbf{b}_{i_{j}},\mathbf{c}_{i_{j}})$ determines the $j$-th occurrence of $(a)$ in $\mathbf{a}$ and we will say that \emph{$a$ occurs at the $i_{j}$-th place of $\mathbf{a}$}.
\end{defi}

\begin{rem}
Let $a$ be an element of $A$. Then, for $\mathbf{a}\in A^{\star}$, $\lvert{\mathbf{a}}\rvert_{a}$, the number of occurrences of $a$ in $\mathbf{a}$, is the value at $\mathbf{a}$ of the unique homomorphism $\lvert{\,\cdot\,}\rvert_{a}$ from $\mathbf{A}^{\star}$ to $(\mathbb{N},+,0)$ such that $\lvert{\,\cdot\,}\rvert_{a}\circ\eta^{A} = \delta_{a}$, where $\delta_{a}$ is the mapping from $A$ to $\mathbb{N}$ that sends $a\in A$ to $1\in\mathbb{N}$ and $b\in A-\{a\}$ to $0\in \mathbb{N}$. Therefore, since, for every $\mathbf{a}\in A^{\star}$, the $A$-indexed family $(\lvert{\mathbf{a}}\rvert_{a})_{a\in A}$ in $\mathbb{N}$ is such that $\lvert{\mathbf{a}}\rvert_{a} = 0$ for all but a finite number of elements $a$ in $A$, i.e., is such that $\mathrm{card}(\{a\in A\mid \lvert{\mathbf{a}}\rvert_{a}\neq 0\})<\aleph_{0}$, we have that $\sum_{a\in A}\lvert{\mathbf{a}}\rvert_{a}\in \mathbb{N}$ and, obviously, for every $\mathbf{a}\in A^{\star}$, $\lvert{\mathbf{a}}\rvert = \sum_{a\in A}\lvert{\mathbf{a}}\rvert_{a}$, i.e., $\lvert{\,\cdot\,}\rvert = \sum_{a\in A}\lvert{\,\cdot\,}\rvert_{a}$.
\end{rem}

\begin{defi} Let $\mathbf{a}$ be a word in $A^{\star}$. Then we will denote by $\sbs{a}{\cdot}_{a\in A}(\mathbf{a})$ the mapping from $\prod_{a\in A}(A^{\star})^{\lvert{\mathbf{a}}\rvert_{a}}$ to $A^{\star}$ defined as follows:
$$
\sbs{a}{\cdot}_{a\in A}(\mathbf{a})
\left\lbrace
\begin{array}{ccl}
{\prod_{a\in A}(A^{\star})^{\lvert{\mathbf{a}}\rvert_{a}}}&\longrightarrow&
{A^{\star}}\\
{((q^{a}_{\alpha})_{\alpha\in \lvert{\mathbf{a}}\rvert_{a}})_{a\in A}}
&\longmapsto&
{\sbs{a}{(q^{a}_{\alpha})_{\alpha\in \lvert{\mathbf{a}}\rvert_{a}}}_{a\in A}(\mathbf{a}),}
\end{array}
\right.
$$
where the word $\sbs{a}{(q^{a}_{\alpha})_{\alpha\in \lvert{\mathbf{a}}\rvert_{a}}}_{a\in A}(\mathbf{a})$ in $A^{\star}$ is obtained by substituting in $\mathbf{a}$, for every $a\in A$ and every $\alpha\in \lvert{\mathbf{a}}\rvert_a$, $q^{a}_{\alpha}$ for the $i_{\alpha}$-th occurrence of $a$ in $\mathbf{a}$. We call $\sbs{a}{(q^{a}_{\alpha})_{\alpha\in \lvert{\mathbf{a}}\rvert_{a}}}_{a\in A}(\mathbf{a})$ \emph{the substitution of $(q^{a}_{\alpha})_{\alpha\in\lvert{\mathbf{a}}\rvert_{a}}$ for $a$ in $\mathbf{a}$ for every $a$ in $A$}, and $\sbs{a}{\cdot}_{a\in A}(\mathbf{a})$ the \emph{substitution operator for $\mathbf{a}$}. 

Let $a$ be an element of $A$. Then we will denote by $\sbs{a}{\cdot}(\mathbf{a})$ the mapping from $(A^{\star})^{\lvert{\mathbf{a}}\rvert_{a}}$ to $A^{\star}$ defined as follows:
$$
\sbs{a}{\cdot}(\mathbf{a})
\left\lbrace
\begin{array}{rcl}
{(A^{\star})^{\lvert{\mathbf{a}}\rvert_{a}}}&\longrightarrow&
{A^{\star}}\\
{(q^{a}_{\alpha})_{\alpha\in \lvert{\mathbf{a}}\rvert_{a}}}&\longmapsto&
{\sbs{a}{(q^{a}_{\alpha})_{\alpha\in \lvert{\mathbf{a}}\rvert_{a}}}(\mathbf{a}),}
\end{array}
\right.
$$
where the word $\sbs{a}{(q^{a}_{\alpha})_{\alpha\in \lvert{\mathbf{a}}\rvert_{a}}}(\mathbf{a})$ in $A^{\star}$ is obtained by substituting in $\mathbf{a}$, for every $\alpha\in \lvert{\mathbf{a}}\rvert_{a}$, $q^{a}_{\alpha}$ for the $i_{\alpha}$-th occurrence of $a$ in $\mathbf{a}$. We call $\sbs{a}{(q^{a}_{\alpha})_{\alpha\in \lvert{\mathbf{a}}\rvert_{a}}}(\mathbf{a})$ \emph{the substitution of $(q^{a}_{\alpha})_{\alpha\in\lvert{\mathbf{a}}\rvert_{a}}$ for $a$ in $\mathbf{a}$}.
\end{defi}

\subsection{Many-sorted algebras} 
Our next aim is to provide the notions from the field of many-sorted universal algebra that will be used later.

\begin{conv} In what follows, for the set of sorts $S$, an arbitrary word on $S^{\star}$ will be denoted by $\mathbf{s}$, i.e., a lower case bold type $s$. The letter $s$ will be used to
represent an arbitrary letter in $S$.
\end{conv}

\begin{defi}\label{DSig}
An $S$-\emph{sorted signature} is a mapping $\Sigma$ from $S^{\star}\times S$ to $\boldsymbol{\mathcal{U}}$ which sends a pair $(\mathbf{s},s)\in S^{\star}\times S$ to the set $\Sigma_{\mathbf{s},s}$ of the \emph{formal operations} of \emph{arity} $\mathbf{s}$, \emph{sort} (or \emph{coarity}) $s$, and \emph{rank} (or \emph{biarity}) $(\mathbf{s},s)$.
\end{defi}

\begin{asm}
From now on $\Sigma$ stands for an $S$-sorted signature, fixed once and for all.
\end{asm}

We shall now give precise definitions of the concepts of many-sorted algebra and of homomorphism between many-sorted algebras.

\begin{defi}\label{DAlg}
The $S^{\star}\times S$-sorted set of the \emph{finitary operations on} an $S$-sorted set $A$ is $(\mathrm{Hom}(A_{\mathbf{s}},A_{s}))_{(\mathbf{s},s)\in S^{\star}\times S}$, where, for every $\mathbf{s}\in S^{\star}$, $A_{\mathbf{s}} = \prod_{j\in \bb{\mathbf{s}}}A_{s_{j}}$, with $\bb{\mathbf{s}}$ denoting the length of the word $\mathbf{s}$ (if $\mathbf{s} = \lambda$, then $A_{\lambda}$ is a final set). Sometimes we let $\mathrm{O}_{\mathrm{H}_{S}}(A)$ stand for 
$(\mathrm{Hom}(A_{\mathbf{s}},A_{s}))_{(\mathbf{s},s)\in S^{\star}\times S}$. A \emph{structure of} $\Sigma$-\emph{algebra on} an $S$-sorted set $A$ is a family $(F_{\mathbf{s},s})_{(\mathbf{s},s)\in S^{\star}\times S}$, denoted by $F$, where, for $(\mathbf{s},s)\in S^{\star}\times S$, $F_{\mathbf{s},s}$ is a mapping from $\Sigma_{\mathbf{s},s}$ to $\mathrm{Hom}(A_{\mathbf{s}},A_{s})$ (if $(\mathbf{s},s) = (\lambda,s)$ and $\sigma\in \Sigma_{\lambda,s}$, then $F_{\lambda,s}(\sigma)$ picks out an element of $A_{s}$). A \emph{many-sorted} $\Sigma$-\emph{algebra} (or, to abbreviate, $\Sigma$-\emph{algebra}) is a pair $(A,F)$, denoted by $\mathbf{A}$, where $A$ is an $S$-sorted set and $F$ a structure of $\Sigma$-algebra on $A$. For a pair $(\mathbf{s},s)\in S^{\star}\times S$ and a formal operation $\sigma\in \Sigma_{\mathbf{s},s}$, in order to simplify the notation, the operation $F_{\mathbf{s},s}(\sigma)$ from $A_{\mathbf{s}}$ to $A_{s}$ will be written as $F_{\sigma}$. In some cases, to avoid mistakes, we will denote by $F^{\mathbf{A}}$ the structure of $\Sigma$-algebra on $A$, and, for $(\mathbf{s},s)\in S^{\star}\times S$ and $\sigma\in \Sigma_{\mathbf{s},s}$, by $F^{\mathbf{A}}_{\sigma}$, or simply by $\sigma^{\mathbf{A}}$, the corresponding operation. Moreover, for $s\in S$ and $\sigma\in\Sigma_{\lambda,s}$, we will, usually, denote by $\sigma^{\mathbf{A}}$ the value of the mapping $F^{\mathbf{A}}_{\sigma}\colon A_{\lambda}\longrightarrow A_{s}$ at the unique element in $A_{\lambda}$.

A $\Sigma$-\emph{homomorphism} (or, to abbreviate, \emph{homomorphism}) from $\mathbf{A}$ to $\mathbf{B}$, where $\mathbf{B} = (B,G)$, is a triple $(\mathbf{A},f,\mathbf{B})$, denoted by $f\colon \mathbf{A}\longrightarrow \mathbf{B}$, where $f$ is an $S$-sorted mapping from $A$ to $B$ such that, for every $(\mathbf{s},s)\in S^{\star}\times S$, every $\sigma\in \Sigma_{\mathbf{s},s}$, and every $(a_{j})_{j\in \bb{\mathbf{s}}}\in A_{\mathbf{s}}$, we have 
$$
f_{s}(\sigma^{\mathbf{A}}((a_{j})_{j\in \bb{\mathbf{s}}})) = \sigma^{\mathbf{B}}(f_{\mathbf{s}}((a_{j})_{j\in \bb{\mathbf{s}}})),
$$
where $f_{\mathbf{s}}$ is the mapping $\prod_{j\in \bb{\mathbf{s}}}f_{s_{j}}$ from $A_{\mathbf{s}}$ to $B_{\mathbf{s}}$ that sends $(a_{j})_{j\in \bb{\mathbf{s}}}$ in $A_{\mathbf{s}}$ to $(f_{s_{j}}(a_{j}))_{j\in \bb{\mathbf{s}}}$ in $B_{\mathbf{s}}$. For a $\Sigma$-algebra $\mathbf{A}$ the triple $(\mathbf{A},\mathrm{id}^{A},\mathbf{A})$, where $\mathrm{id}^{A}$ is the identity at $A$, which is a homomorphism from $\mathbf{A}$ to $\mathbf{A}$, will be denoted by $\mathrm{id}^{\mathbf{A}}$ and will be called the \emph{identity} homomorphism at $\mathbf{A}$. We will denote by $\mathsf{Alg}(\Sigma)$ the category of $\Sigma$-algebras and homomorphisms and by $\mathrm{Alg}(\Sigma)$ the set of objects of $\mathsf{Alg}(\Sigma)$.
\end{defi}

\begin{defi}
Let $\mathbf{A}$ be a $\Sigma$-algebra. Then the \emph{support of} $\mathbf{A}$, denoted by $\mathrm{supp}_{S}(\mathbf{A})$, is $\mathrm{supp}_{S}(A)$, i.e., the support of the underlying $S$-sorted set $A$ of $\mathbf{A}$.
\end{defi}

\begin{defi}
Let $\mathbf{A}$ be a $\Sigma$-algebra. We will say that $\mathbf{A}$ is \emph{finite} if $A$, the underlying $S$-sorted set of $\mathbf{A}$, is finite.
\end{defi}

For every $\Sigma$-algebra it is possible to define a preorder on the coproduct of its underlying $S$-sorted set. We note that, for single-sorted algebras, in~\cite{Die66}, Diener introduced an order relation on generalized Peano algebras and studied the properties of this relation.

\begin{defi}\label{DOrdAlg}
Let $\mathbf{A}=(A,F)$ be a $\Sigma$-algebra. Then $<_{\mathbf{A}}$ denotes the binary relation on $\coprod A$ consisting of the ordered pairs $((a,t),(b,s))\in(\coprod A)^{2}$ for which there exists a $\mathbf{s}\in S^{\star}-\{\lambda\}$, a
$\sigma\in\Sigma_{\mathbf{s},s}$, and an $(a_{j})_{j\in\bb{\mathbf{s}}}\in A_{\mathbf{s}}$ such that $\sigma^{\mathbf{A}}((a_{j})_{j\in\bb{\mathbf{s}}}) = b$ and, for some $j\in \lvert{\mathbf{s}}\rvert$, $s_{j}=t$ and $a_{j}=a$. We will denote by $\leq_{\mathbf{A}}$ the reflexive and transitive closure of $<_{\mathbf{A}}$, i.e., the preorder on $\coprod A$ generated by $<_{\mathbf{A}}$.
\end{defi}

\begin{rem}
The preorder $\leq_{\mathbf{A}}$ on $\coprod A$ is defined by letting $((a,t),(b,s))\in \leq_{\mathbf{A}}$ mean that $t = s$ and $a = b$ or there exists an $n\in \mathbb{N}-\{0\}$, a $\mathbf{u}\in S^{\star}$, and a family $(c_{i})_{i\in \lvert{\mathbf{u}}\rvert}\in A_{\mathbf{u}}$ such that $\lvert{\mathbf{u}}\rvert = n+1$, $u_{0} = t$, $u_{n} = s$, $c_{0} = a$, $c_{n} = b$, and, for every $i\in n$, $((c_{i},u_{i}),(c_{i+1},u_{i+1}))\in <_{\mathbf{A}}$.
\end{rem}

We next state, among other things, that there exists an endofunctor of $\mathsf{Alg}(\Sigma)$ that assigns to a $\Sigma$-algebra $\mathbf{A}$ the subset $\Sigma$-algebra constructed on $A^{\wp}$. 
We note that, in Section~3, the subset $\Sigma$-algebra associated to a free many-sorted algebra will appear as the codomain of a substitution operator which will be used on some recognizability results.

\begin{prop}\label{PSubsetAlg}
Let $(\cdot)^{\wp}$ be the mapping that sends
\begin{enumerate}
\item a $\Sigma$-algebra $\mathbf{A}$ to the $\Sigma$-algebra $\mathbf{A}^{\wp} = (A^{\wp},F^{\wp})$ where $A^{\wp}$ is $(\mathrm{Sub}(A_{s}))_{s\in S}$ and where, for every $(\mathbf{s},s)\in S^{\star}\times S$ and every $\sigma\in \Sigma_{\mathbf{s},s}$, $F_{\sigma}^{\wp}$, simply denoted by $\sigma^{\mathbf{A}^{\wp}}$, is the mapping from $A^{\wp}_{\mathbf{s}} = \prod_{i\in \lvert{\mathbf{s}}\rvert}\mathrm{Sub}(A_{s_{i}})$ to $A^{\wp}_{s} = \mathrm{Sub}(A_{s})$ defined as follows:
$$
\sigma^{\mathbf{A}^{\wp}}
\left\lbrace
\begin{array}{ccl}
{A^{\wp}_{\mathbf{s}}}&\longrightarrow&
{A^{\wp}_{s}}\\
{(L_{i})_{i\in \lvert{\mathbf{s}}\rvert}}&\longmapsto&
{\{\sigma^{\mathbf{A}}((x_{i})_{i\in \lvert{\mathbf{s}}\rvert})\mid (x_{i})_{i\in \lvert{\mathbf{s}}\rvert}\in \prod_{i\in \lvert{\mathbf{s}}\rvert}L_{i}\},} 
\end{array}
\right.
$$
\item and a homomorphism $f$ from $\mathbf{A}$ to $\mathbf{B}$ to the $S$-sorted mapping $f^{\wp} = (f_{s}[\cdot])_{s\in S}$ from $A^{\wp}$ to $B^{\wp}$ (see Corollary~\ref{CSubFunct}).
\end{enumerate}
Then $(\cdot)^{\wp}$ is an endofunctor of $\mathsf{Alg}(\Sigma)$. The $\Sigma$-algebra $\mathbf{A}^{\wp}$ is called the \emph{subset} or, better still, \emph{power $\Sigma$-algebra} associated to $\mathbf{A}$ (this notion, but for single-sorted algebras, is due to Mezei and Wright, cf.~\cite{MW67}, Definition 2.2). Moreover, there exists a pointwise monomorphic natural transformation $\{\cdot\}^{\Sigma} = (\{\cdot\}^{\Sigma}_{\mathbf{A}})_{\mathbf{A}\in \mathrm{Alg}(\Sigma)}$ from $\mathrm{Id}_{\mathsf{Alg}(\Sigma)}$ to $(\cdot)^{\wp}$ 
and a natural transformation $\mu^{\Sigma} = (\mu^{\Sigma}_{\mathbf{A}})_{\mathbf{A}\in \mathrm{Alg}(\Sigma)}$
from $(\cdot)^{\wp\circ\wp}$ to $(\cdot)^{\wp}$, 
where $(\cdot)^{\wp\circ\wp} = (\cdot)^{\wp}\circ (\cdot)^{\wp}$, 
such that $((\cdot)^{\wp},\mu^{\Sigma},\{\cdot\}^{\Sigma})$ is a monad on $\mathsf{Alg}(\Sigma)$. 
\end{prop}

\begin{rem}
For $(\mathbf{s},s)\in S^{\star}\times S$ and $\sigma\in \Sigma_{\mathbf{s},s}$, the mapping $\sigma^{\mathbf{A}^{\wp}}$ from $A^{\wp}_{\mathbf{s}} = \prod_{i\in \lvert{\mathbf{s}}\rvert}\mathrm{Sub}(A_{s_{i}})$ to $\mathrm{Sub}(A_{s})$ sends $(L_{i})_{i\in \lvert{\mathbf{s}}\rvert}$ in $A^{\wp}_{\mathbf{s}}$ to $\sigma^{\mathbf{A}}[\prod_{i\in\lvert{\mathbf{s}}\rvert}L_{i}]$ in $\mathrm{Sub}(A_{s})$. Therefore, if, for some $i\in \lvert{\mathbf{s}}\rvert$, $L_{i} = \varnothing$, then $\sigma^{\mathbf{A}^{\wp}}((L_{i})_{i\in \lvert{\mathbf{s}}\rvert}) = \varnothing$.
\end{rem}

\subsection{Subalgebras}

We shall now go on to define the notion of subalgebra of a $\Sigma$-algebra $\mathbf{A}$ and the subalgebra generating operator for $\mathbf{A}$.

\begin{defi}\label{DAlgSub}
Let $\mathbf{A}$ be a $\Sigma$-algebra and $X\subseteq A$. Given $(\mathbf{s},s)\in S^{\star}\times S$ and $\sigma\in\Sigma_{\mathbf{s},s}$, we will say that $X$ is \emph{closed under the operation} $\sigma^{\mathbf{A}}\colon A_{\mathbf{s}}\longrightarrow A_{s}$ if, for every $(a_{j})_{j\in\bb{\mathbf{s}}}\in X_{\mathbf{s}}$, $\sigma^{\mathbf{A}}((a_{j})_{j\in\bb{\mathbf{s}}})\in X_{s}$. We will say that $X$ is a \emph{closed subset} of $\mathbf{A}$ if $X$ is closed under the operations of $\mathbf{A}$. We will denote by $\mathrm{Cl}(\mathbf{A})$ the set of all closed subsets of $\mathbf{A}$ (which is an algebraic closure system on $A$) and by $\mathbf{Cl}(\mathbf{A})$ the algebraic lattice $(\mathrm{Cl}(\mathbf{A}),\subseteq)$. We will say that a $\Sigma$-algebra $\mathbf{B}$ is a \emph{subalgebra} of $\mathbf{A}$ if $B\subseteq A$ and the canonical embedding of $B$ into $A$ determines an embedding of $\mathbf{B}$ into $\mathbf{A}$. We will denote by $\mathrm{Sub}(\mathbf{A})$ the set of all subalgebras of $\mathbf{A}$. Since $\mathrm{Cl}(\mathbf{A})$ and $\mathrm{Sub}(\mathbf{A})$ are isomorphic, we shall feel free to deal with a closed subset of $\mathbf{A}$ or with the correlated subalgebra of $\mathbf{A}$, whichever is more convenient for the work at hand.
\end{defi}

\begin{defi}
Let $\mathbf{A}$ be a $\Sigma$-algebra. Then we denote by $\mathrm{Sg}_{\mathbf{A}}$ the many-sorted closure operator on $A$ defined as follows:
$$\textstyle
 \mathrm{Sg}_{\mathbf{A}}
 \left\lbrace
 \begin{array}{ccl}
 {\mathrm{Sub}(A)}&\longrightarrow&
 {\mathrm{Sub}(A)}\\
 {X}&\longmapsto&
 {\bigcap \{C\in\mathrm{Sub}(\mathbf{A})\mid X\subseteq C\}}
 \end{array}
 \right.
$$

We call $\mathrm{Sg}_{\mathbf{A}}$ the \emph{subalgebra generating many-sorted operator on} $A$ \emph{determined by} $\mathbf{A}$. For every $X\subseteq A$, we call $\mathrm{Sg}_{\mathbf{A}}(X)$ the \emph{subalgebra of} $\mathbf{A}$ \emph{generated by} $X$. Moreover, if $X\subseteq A$ is such that $\mathrm{Sg}_{\mathbf{A}}(X) = A$, then we say that $X$ is an $S$-sorted set of \emph{generators} of $\mathbf{A}$, or that $X$ \emph{generates} $\mathbf{A}$. 
\end{defi}

\subsection{Recognizability}

We finish this section by reviewing a few aspects of recognizability for subsets of the underlying many-sorted set of an arbitrary many-sorted algebra. 

\begin{defi}
Let $\mathbf{A}$ be a $\Sigma$-algebra, and $L\subseteq A$. We will say that $L$ is \emph{recognizable} if there exists a finite $\Sigma$-algebra $\mathbf{B}$, a homomorphism $f\colon\mathbf{A}\longrightarrow\mathbf{B}$, and a subset $M$ of $B$ such that
 $f^{-1}[M] = L$.
We will denote by $\mathrm{Rec}(\mathbf{A})$ the set of all subsets of $A$ which are recognizable.
\end{defi}

We introduce some standard results on recognizability for arbitrary $\Sigma$-algebras.

\begin{propC}
[{\cite[Prop.~2.103]{CVCL20}}]
\label{PRecCl}
Let $\mathbf{A}$ and $\mathbf{B}$ be $\Sigma$-algebras. Then
\begin{enumerate}
\item $\varnothing^{S}, A\in \mathrm{Rec}(\mathbf{A})$.
\item If $K,L\in \mathrm{Rec}(\mathbf{A})$, then $K\cup L, K\cap L, K-L\in \mathrm{Rec}(\mathbf{A})$.
\end{enumerate}
\end{propC}

The following result establishes the relationship between the recognizability in an algebra and the recognizability in (a homomorphic image of) a subalgebra.
\begin{prop}\label{PRecSub} 
Let $\mathbf{A}$ and $\mathbf{B}$ be $\Sigma$-algebras. Let $i\colon \mathbf{A}\longrightarrow\mathbf{B}$ be an injective $\Sigma$-homomorphism (that is, $\mathbf{A}$ is a homomorphic image of a subalgebra of $\mathbf{B}$). If $L\subseteq A$ is such that $i[L]\in\mathrm{Rec}(\mathbf{B})$, then $L\in\mathrm{Rec}(\mathbf{A})$. In particular, if $\mathbf{A}$ is a subalgebra of $\mathbf{B}$ and $L\subseteq A$ is such that $L\in \mathrm{Rec}(\mathbf{B})$, then $L\in \mathrm{Rec}(\mathbf{A})$.
\end{prop}
\begin{proof}
Let $L$ be a subset of $A$ satisfying $i[L]\in \mathrm{Rec}(\mathbf{B})$. Then there exists a finite $\Sigma$-algebra $\mathbf{C}$, a $\Sigma$-homomorphism $f\colon \mathbf{B}\longrightarrow\mathbf{C}$, and a subset $M\subseteq C$ such that $f^{-1}[M]=L$. Consider the $\Sigma$-homomorphism $f\circ i\colon \mathbf{A} \longrightarrow \mathbf{C}$ and the subset $M\subseteq C$. It suffices to prove that 
$(f\circ i)^{-1}[M]=L$, that is, for every $s\in S$,  $(f\circ i)^{-1}_{s}[M_{s}]=L_{s}$.

Let $s\in S$. Consider an element $a\in (f\circ i)^{-1}_{s}[M_{s}]$, then $f_{s}(i_{s}(a))\in M_{s}$, that is, $i_{s}(a)\in f^{-1}_{s}[M_{s}]=i_{s}[L_{s}]$. Consequently, there exists $l\in L_{s}$ such that $i_{s}(a)=i_{s}(l)$. Since $i_{s}$ is injective, we conclude that $a=l$, that is, $a\in L_{s}$. Thus, $(f\circ i)^{-1}_{s}[M_{s}]\subseteq L_{s}$. Now, let $a\in L_{s}$. Since $L\subseteq A$, we have $a\in A_{s}$ and $i_{s}(a)\in i_{s}[L_{s}]$. Since $f^{-1}[M]=i[L]$, we conclude that $f_{s}(i_{s}(a))\in M_{s}$, that is, $a\in (f\circ i)_{s}^{-1}[M_{s}]$. Thus, $L_{s}\subseteq (f\circ i)^{-1}_{s}[M_{s}]$. Hence, we conclude that $(f\circ i)^{-1}[M]=L$.
\end{proof}

We next introduce the notion of $s$-recognizability, for a given sort $s\in S$.

\begin{defi}
Let $\mathbf{A}$ be a $\Sigma$-algebra, $s\in S$, and $L\subseteq A_{s}$. We will say that $L$ is $s$-\emph{recognizable} if there exists a finite $\Sigma$-algebra $\mathbf{B}$, a homomorphism $f\colon\mathbf{A}\longrightarrow\mathbf{B}$, and a subset $M$ of $B_{s}$ such that
 $f_{s}^{-1}[M] = L$.
We will denote by $\mathrm{Rec}_{s}(\mathbf{A})$ the set of all subsets of $A_{s}$ which are $s$-recognizable.
\end{defi}

The following propositions explain the relationship between recognizability and $s$-recognizability.

\begin{propC}[{\cite[Prop.~2.104]{CVCL20}}]
Let $\mathbf{A}$ be a $\Sigma$-algebra, $s\in S$, and $L\subseteq A_{s}$. Then $L\in \mathrm{Rec}_{s}(\mathbf{A})$ if and only if $\delta^{s,L}\in \mathrm{Rec}(\mathbf{A})$. Thus, $\mathrm{Rec}_{s}(\mathbf{A})$ is isomorphic to a subset of $\mathrm{Rec}(\mathbf{A})$.
\end{propC}

\begin{asm}
The following proposition assumes that $S$ is finite.
\end{asm}

\begin{propC}[{\cite[Prop.~2.105]{CVCL20}}]\label{PRecSubd}
Let $\mathbf{A}$ be a $\Sigma$-algebra and $L\subseteq A$. Then $L\in \mathrm{Rec}(\mathbf{A})$ if and only if, for every $s\in S$, $L_{s}\in \mathrm{Rec}_{s}(\mathbf{A})$. Thus, there exists an embedding from $\mathrm{Rec}(\mathbf{A})$ into $\prod_{s\in S}\mathrm{Rec}_{s}(\mathbf{A})$ and $\mathrm{Rec}(\mathbf{A})$ is a subdirect product of $(\mathrm{Rec}_{s}(\mathbf{A}))_{s\in S}$.
\end{propC}

\begin{rem}
 From the proposition just stated (which, we recall, has been obtained under the assumption that $S$ is finite) and since, for a $\Sigma$-algebra $\mathbf{A}$ and an $L\subseteq A$, $L = \bigcup_{s\in S}\delta^{s,L_{s}}$, it follows that in order to investigate the languages in $\mathrm{Rec}(\mathbf{A})$ it suffices to investigate, for every $s\in S$, the languages in $\mathrm{Rec}_{s}(\mathbf{A})$.
\end{rem}

We conclude this subsection with some standard results on $s$-recognizability for arbitrary $\Sigma$-algebras. They follow immediately from Propositions~\ref{PRecCl} and~\ref{PRecSubd}.

\begin{propC}
[{\cite[Prop.~2.106]{CVCL20}}]
\label{PsRecCl}
Let $\mathbf{A}$ and $\mathbf{B}$ be $\Sigma$-algebras and $s, t\in S$. Then
\begin{enumerate}
\item $\varnothing, A_{s}\in \mathrm{Rec}_{s}(\mathbf{A})$.
\item If $K,L\in \mathrm{Rec}_{s}(\mathbf{A})$, then $K\cup L, K\cap L, K-L\in \mathrm{Rec}_{s}(\mathbf{A})$.
\end{enumerate}
\end{propC}

The following result is analogous to Proposition~\ref{PRecSub} for the case of $s$-recognizability in a homomorphic image of a subalgebra. 

\begin{prop}\label{PsRecSub}
Let $\mathbf{A}$ and $\mathbf{B}$ be $\Sigma$-algebras. Let $i\colon \mathbf{A}\longrightarrow\mathbf{B}$ be an injective $\Sigma$-homomorphism (that is, $\mathbf{A}$ is a homomorphic image of a subalgebra of $\mathbf{B}$). Let $s\in S$ and $L\subseteq A_{s}$ be such that $i_{s}[L_{s}]\in\mathrm{Rec}_{s}(\mathbf{B})$, then $L\in\mathrm{Rec}_{s}(\mathbf{A})$. In particular, if $\mathbf{A}$ is a subalgebra of $\mathbf{B}$, $s\in S$, and $L\subseteq A_{s}$ is such that $L\in \mathrm{Rec}_{s}(\mathbf{B})$, then $L\in \mathrm{Rec}_{s}(\mathbf{A})$.
\end{prop}

\section{Recognizability on the free many-sorted algebra}

In this section we specialize the previous notions to the free many-sorted algebra and introduce the substitution and iteration operators on which the regular expressions of Section 4 are built. We begin by recalling that the forgetful functor $\mathrm{G}_{\Sigma}$ from $\mathsf{Alg}(\Sigma)$ to
$\mathsf{Set}^{S}$ has a left adjoint $\mathbf{T}_{\Sigma}$ which assigns to an $S$-sorted set $X$ the free $\Sigma$-algebra $\mathbf{T}_{\Sigma}(X)$ on $X$. Let us note that in what follows, to construct the algebra of $\Sigma$-rows in $X$, and the free $\Sigma$-algebra on $X$, since neither the $S$-sorted signature $\Sigma$ nor the $S$-sorted set $X$ are subject to any constraint, coproducts must necessarily be used.

\begin{asm}
From now on $X$ stands for an $S$-sorted set, fixed once and for all.
\end{asm}

\begin{defi}\label{DWAlg}
Let $X$ be an $S$-sorted set. The \emph{algebra of} $\Sigma$-\emph{rows in} $X$, denoted by $\mathbf{W}_{\Sigma}(X)$, is defined as follows:
\begin{enumerate}
\item The underlying $S$-sorted set of $\mathbf{W}_{\Sigma}(X)$, written as $\mathrm{W}_{\Sigma}(X)$, is precisely the $S$-sorted set $((\coprod\Sigma \amalg \coprod X)^{\star})_{s\in S}$, i.e., the mapping from $S$ to $\boldsymbol{\mathcal{U}}$ constantly $(\coprod\Sigma \amalg \coprod X)^{\star}$, where $(\coprod\Sigma \amalg \coprod X)^{\star}$ is the set of all words on the set $\coprod\Sigma \amalg \coprod X$, i.e., on the set
 $$
 \textstyle
 [(\bigcup_{(\mathbf{s},s)\in S^{\star}\times S}(\Sigma_{\mathbf{s},s}\times\{(\mathbf{s},s)\}))\times\{0\}]\cup
 [(\bigcup_{s\in S}(X_{s}\times\{s\}))\times\{1\}].
 $$

\item For every $(\mathbf{s},s)\in S^{\star}\times S$, and every $\sigma\in\Sigma_{\mathbf{s},s}$, the structural operation $F_{\sigma}$ associated to $\sigma$ is the mapping from $\mathrm{W}_{\Sigma}(X)_{\mathbf{s}}$ to ${\mathrm{W}_{\Sigma}(X)}_{s}$ defined as
$$
F_{\sigma}
\left\lbrace
\begin{array}{rcl}
{\mathrm{W}_{\Sigma}(X)_{\mathbf{s}}}&\longrightarrow&
{{\mathrm{W}_{\Sigma}(X)}_{s}}\\
{(P_{j})_{j\in\lvert \mathbf{s} \rvert}}&\longmapsto&
{(\sigma)\curlywedge\concat_{j\in\lvert \mathbf{s} \rvert}P_{j}}
\end{array}
\right.
$$
where, for every $(\mathbf{s},s)\in S^{\star}\times S$, and every $\sigma\in\Sigma_{\mathbf{s},s}$, $(\sigma)$ stands for $(((\sigma,(\mathbf{s},s)),0))$, which is the value at $\sigma$ of the canonical mapping from $\Sigma_{\mathbf{s},s}$ to $(\coprod\Sigma \amalg \coprod X)^{\star}$.
\end{enumerate}
\end{defi}

\begin{defi}\label{DAlgFree}
The \emph{free} $\Sigma$-\emph{algebra on} an $S$-sorted set $X$, denoted by $\mathbf{T}_{\Sigma}(X)$, is the $\Sigma$-algebra determined by $\mathrm{Sg}_{\mathbf{W}_{\Sigma}(X)}((\{(x)\mid x\in X_{s}\})_{s\in S})$, the subalgebra of $\mathbf{W}_{\Sigma}(X)$ generated by $(\{(x)\mid x\in X_{s}\})_{s\in S}$, where, for every $s\in S$ and every $x\in X_{s}$, $(x)$ stands for $(((x,s),1))$, which is the value at $x$ of the canonical mapping from $X_{s}$ to $(\coprod\Sigma \amalg \coprod X)^{\star}$.

We will denote by $\mathrm{T}_{\Sigma}(X)$ the underlying $S$\nobreakdash-sorted set of $\mathbf{T}_{\Sigma}(X)$ and, for $s\in S$, we will call the elements of $\mathrm{T}_{\Sigma}(X)_{s}$ \emph{terms of type} $s$ \emph{with variables in} $X$ or $(X,s)$-\emph{terms}.
\end{defi}

In the many-sorted case we have, as in the single-sorted case, the following characterization of the elements of $\mathrm{T}_{\Sigma}(X)_{s}$, for $s\in S$.

\begin{prop}\label{PTermChar}
For every $s\in S$ and every $P\in
\mathrm{T}_{\Sigma}(X)_{s}$, exactly one of the following conditions holds
\begin{enumerate}
 \item $P = (x)$, for a unique $x\in X_{s}$, or
 \item $P = (\sigma)$, for a unique $\sigma\in\Sigma_{\lambda,s}$, or
 \item $P = (\sigma)\curlywedge\concat_{j\in\bb{\mathbf{s}}}P_{j}$, for a unique $\mathbf{s}\in S^{\star}-\{\lambda\}$, a unique $\sigma\in\Sigma_{\mathbf{s},s}$, and a unique family $(P_{j})_{j\in\bb{\mathbf{s}}}\in\mathrm{T}_{\Sigma}(X)_{\mathbf{s}}$.
\end{enumerate}
Moreover, the three possibilities are mutually exclusive. From now on, for simplicity of notation, we will write $x$, $\sigma^{\mathbf{T}_{\Sigma}(X)}$, and $\sigma^{\mathbf{T}_{\Sigma}(X)}((P_{j})_{j\in \bb{\mathbf{s}}})$ instead of $(x)$, $(\sigma)$, and $(\sigma)\curlywedge\concat_{j\in\bb{\mathbf{s}}} P_{j}$, respectively. Thus, in particular, the structural operation $F_{\sigma}$ (or more accurately $F_{\sigma}^{\mathbf{T}_{\Sigma}(X)}$) associated with $\sigma$ is identified with $\sigma^{\mathbf{T}_{\Sigma}(X)}$.
\end{prop}

From the above proposition, the universal property of the free $\Sigma$-algebra on an $S$-sorted set $X$ follows immediately, as stated in the subsequent proposition.

\begin{prop}\label{PPropUniv}
The pair $(\mathbf{T}_{\Sigma}(X),\eta^{X})$, where $\eta^{X}$, the
\emph{insertion of (the $S$-sorted set of generators)} $X$ \emph{into} $\mathrm{T}_{\Sigma}(X)$, is the co\-restric\-tion to
$\mathrm{T}_{\Sigma}(X)$ of the canonical embedding of $X$ into $\mathrm{W}_{\Sigma}(X)$, has the following universal property: for every $\Sigma$-algebra $\mathbf{A}$ and every $S$-sorted mapping $f\colon X\longrightarrow A$, there exists a unique homomorphism $f^{\sharp}\colon\mathbf{T}_{\Sigma}(X)\longrightarrow\mathbf{A}$ such that $f^{\sharp}\circ \eta^{X} = f$.
\end{prop}
\begin{proof}
For every $s\in S$, the $s$-th coordinate $f^{\sharp}_{s}$ of $f^{\sharp}$ is defined recursively as follows: 
$$
f^{\sharp}_{s}
\left\lbrace
\begin{array}{ccl}
{\mathrm{T}_{\Sigma}(X)_{s}}&\longrightarrow&
{A_{s}}\\
{P}&\longmapsto&
{\begin{cases}
 f_{s}(x),&\text{if }P = x; \\
 \sigma^{\mathbf{A}},&\text{if }P = \sigma^{\mathbf{T}_{\Sigma}(X)};\\
 {\sigma}^{\mathbf{A}}((f^{\sharp}_{s_{j}}(P_{j}))_{j\in\lvert{\mathbf{s}}\rvert}),
 &\text{if }P = \sigma^{\mathbf{T}_{\Sigma}(X)}((P_{j})_{j\in\lvert{\mathbf{s}}\rvert}). 
\end{cases}}
\end{array}
\right.
$$
It is straightforward to prove the remaining properties.
\end{proof}

Moreover, for the case of free algebras, the preorder introduced in Definition~\ref{DOrdAlg} is, in fact, an Artinian order whose minimal elements are, precisely, the variables and the constant operation symbols. 

\begin{prop}\label{PArtOrd}
The order $\leq_{\mathbf{T}_{\Sigma}(X)}$ does not have strictly descending $\omega_{0}$-chains, i.e., is an Artinian order. Moreover, $\mathrm{Min}(\coprod\mathrm{T}_{\Sigma}(X),\leq_{\mathbf{T}_{\Sigma}(X)})$, the set of minimal elements of the ordered set $(\coprod\mathrm{T}_{\Sigma}(X),\leq_{\mathbf{T}_{\Sigma}(X)})$, is given by
$$
\textstyle
\mathrm{Min}(\coprod\mathrm{T}_{\Sigma}(X),\leq_{\mathbf{T}_{\Sigma}(X)})=\{
(x,s)\mid s\in S, x \in X_{s}
\}\cup 
\{
(\sigma^{\mathbf{T}_{\Sigma}(X)},s)\mid s \in S, \sigma \in \Sigma_{\lambda,s}
\}.$$
\end{prop}

\begin{rem}\label{RRestII}
Let $(\mathbf{s},s)\in (S-\{\lambda\})\times S$, $\sigma\in \Sigma_{\mathbf{s},s}$, and $(P_{j})_{j\in\bb{\mathbf{s}}}\in\mathrm{T}_{\Sigma}(X)_{\mathbf{s}}$. Finally, let us consider the term $P=\sigma^{\mathbf{T}_{\Sigma}(X)}((P_{j})_{j\in\bb{\mathbf{s}}})$ in
$\mathrm{T}_{\Sigma}(X)_{s}$. By Proposition~\ref{PTermChar} and
Definition~\ref{DOrdAlg}, for every 
$(Q,t)\in\coprod\mathrm{T}_{\Sigma}(X)$, we have $(Q,t)<(P,s)$ if and only
if $(Q,t)\leq(P_{j},s_{j})$ for some $j\in\bb{\mathbf{s}}$. Indeed, the
relation $<_{\mathbf{T}_{\Sigma}(X)}$ of Definition~\ref{DOrdAlg} relates a
pair to a term exactly through the immediate arguments of one of its
decompositions, and the decomposition of $P$ is unique.
\end{rem}

This order allows us to define the notion of subterm of a given term. 

\begin{defi}
Let $s\in S$, and $P\in\mathrm{T}_{\Sigma}(X)_{s}$. Then the $S$-sorted set of all \emph{subterms of} $P$, denoted by $\mathrm{Subt}(P)$, is defined as follows:
$$
\mathrm{Subt}(P) = (\{Q\in \mathrm{T}_{\Sigma}(X)_{t} \mid (Q,t)
\leq_{\mathbf{T}_{\Sigma}(X)} (P,s)\})_{t \in S }.
$$
\end{defi}

\begin{rem} For a given term $P\in\mathrm{T}_{\Sigma}(X)_{s}$, $\mathrm{Subt}(P)$ is a finite subset of $\mathrm{T}_{\Sigma}(X)$, i.e.,
$\mathrm{Subt}(P)\in \mathrm{Sub}_{\mathrm{f}}(\mathrm{T}_{\Sigma}(X))$. Moreover, $\mathrm{Subt}(P)$ can also be characterized as the smallest subset $\mathcal{L}$ of $\mathrm{T}_{\Sigma}(X)$ which satisfies the following conditions:
\begin{enumerate}
\item $P\in \mathcal{L}_{s}$ and
\item for every $(\mathbf{s},s)\in S^{\star}\times S$, every $\sigma\in \Sigma_{\mathbf{s},s}$, and every $(P_{j})_{j\in\bb{\mathbf{s}}}\in \mathrm{T}_{\Sigma}(X)_{\mathbf{s}}$,
$$
\text{if } \sigma^{\mathbf{T}_{\Sigma}(X)}((P_{j})_{j\in\bb{\mathbf{s}}})\in \mathcal{L}_{s}, \text{ then } P_{j}\in \mathcal{L}_{s_{j}} \text{ for every } j\in\bb{\mathbf{s}}.
$$
\end{enumerate}
Note that the second condition is exactly the converse of the defining condition of the concept of subalgebra of a $\Sigma$-algebra.
\end{rem}

\begin{rem}\label{RRest}
Let $Y$ be an $S$-sorted set with $X\subseteq Y$. Then
$\mathrm{T}_{\Sigma}(X)\subseteq\mathrm{T}_{\Sigma}(Y)$ and
$\mathbf{T}_{\Sigma}(X)$ is a subalgebra of $\mathbf{T}_{\Sigma}(Y)$.
Moreover, by the unique readability of terms (Proposition~\ref{PTermChar}),
for $P\in\mathrm{T}_{\Sigma}(X)_{s}$ the unique decompositions of $P$ in
$\mathbf{T}_{\Sigma}(X)$ and in $\mathbf{T}_{\Sigma}(Y)$ coincide. From
this and Definition~\ref{DOrdAlg} it follows that: 
\begin{enumerate}
 \item the order
$\leq_{\mathbf{T}_{\Sigma}(X)}$ is the restriction of
$\leq_{\mathbf{T}_{\Sigma}(Y)}$ to $\coprod\mathrm{T}_{\Sigma}(X)$;
\item every subterm, computed in $\mathrm{T}_{\Sigma}(Y)$, of a term in
$\mathrm{T}_{\Sigma}(X)$ belongs to $\mathrm{T}_{\Sigma}(X)$; 
\item by
Proposition~\ref{PArtOrd}, a pair in $\coprod\mathrm{T}_{\Sigma}(X)$ is
minimal in $(\coprod\mathrm{T}_{\Sigma}(X),\leq_{\mathbf{T}_{\Sigma}(X)})$
if and only if it is minimal in
$(\coprod\mathrm{T}_{\Sigma}(Y),\leq_{\mathbf{T}_{\Sigma}(Y)})$.
\end{enumerate}
Accordingly, when the ambient $S$-sorted set of
variables is clear from the context, or irrelevant to it, we will write
$\leq$, $<$ and $\mathrm{Min}$ without further specification.
\end{rem}

\subsection{
\texorpdfstring
{Substitutions}
{Substitutions}
}\label{SSI}

 In this subsection, we introduce several substitution operators associated with a free many-sorted algebra. This subsection also includes some technical lemmas that will be of interest for the proof of the main result of this paper.

\begin{defi}
Let $s\in S$, and $P\in \mathrm{T}_{\Sigma}(X)_{s}$. We denote by $\sbs{(x,t)}{\cdot}_{(x,t)\in \coprod X}(P)$ the mapping from $\prod_{(x,t)\in \coprod X}\mathrm{T}_{\Sigma}(X)_{t}^{\lvert{P}\rvert_{x}}$ to $\mathrm{T}_{\Sigma}(X)_{s}$ 
defined as follows:
$$
\sbs{(x,t)}{\cdot}_{(x,t)\in \coprod X}(P) 
\left\lbrace
\begin{array}{rcl}
{\prod_{(x,t)\in \coprod X}\mathrm{T}_{\Sigma}(X)_{t}^{\lvert{P}\rvert_{x}}}&\longrightarrow&
{\mathrm{T}_{\Sigma}(X)_{s}}\\
{\left((Q^{x,t}_{\alpha})_{\alpha\in \lvert{P}\rvert_{x}}\right)_{(x,t)\in \coprod X}}&\longmapsto&
{\sbs{(x,t)}{(Q^{x,t}_{\alpha})_{\alpha\in \lvert{P}\rvert_{x}}}_{(x,t)\in \coprod X}(P)}
\end{array}
\right.
$$
where the word $\sbs{(x,t)}{(Q^{x,t}_{\alpha})_{\alpha\in \lvert{P}\rvert_{x}}}_{(x,t)\in \coprod X}(P)$ in $\mathrm{T}_{\Sigma}(X)_{s}$ is obtained by substituting in $P$, for every $t\in S$, every $x\in X_{t}$, and every $\alpha\in \lvert{P}\rvert_x$, $Q^{x,t}_{\alpha}$ for the $i_{\alpha}$-th occurrence of $x$ in $P$. We will call $\sbs{(x,t)}{(Q^{x,t}_{\alpha})_{\alpha\in \lvert{P}\rvert_{x}}}_{(x,t)\in \coprod X}(P)$ \emph{the substitution of $(Q^{x,t}_{\alpha})_{\alpha\in\lvert{P}\rvert_{x}}$ for $x$ in $P$ for every $t\in S$ and every $x$ in $X_{t}$}, and $\sbs{(x,t)}{\cdot}_{(x,t)\in \coprod X}(P)$ the \emph{global substitution operator for $P$}.

Let $u$ be a sort in $S$ and $z\in X_{u}$. Then we will denote by $\sbs{z}{\cdot}(P)$ the mapping from $\mathrm{T}_{\Sigma}(X)_{u}^{\lvert{P}\rvert_{z}}$ to $\mathrm{T}_{\Sigma}(X)_{s}$ 
defined as follows:
$$
\sbs{z}{\cdot}(P) 
\left\lbrace
\begin{array}{rcl}
{\mathrm{T}_{\Sigma}(X)_{u}^{\lvert{P}\rvert_{z}}}&\longrightarrow&
{\mathrm{T}_{\Sigma}(X)_{s}}\\
{(Q^{z}_{\alpha})_{\alpha\in \lvert{P}\rvert_{z}}}&\longmapsto&
{\sbs{z}{(Q^{z}_{\alpha})_{\alpha\in \lvert{P}\rvert_{z}}}(P)}
\end{array}
\right.
$$
where the word $\sbs{z}{(Q^{z}_{\alpha})_{\alpha\in \lvert{P}\rvert_{z}}}(P)$ in $\mathrm{T}_{\Sigma}(X)_{s}$ is obtained by substituting in $P$, for every $\alpha\in \lvert{P}\rvert_{z}$, $Q^{z}_{\alpha}$ for the $i_{\alpha}$-th occurrence of $z$ in $P$. We will call $\sbs{z}{(Q^{z}_{\alpha})_{\alpha\in \lvert{P}\rvert_{z}}}(P)$ \emph{the substitution of $(Q^{z}_{\alpha})_{\alpha\in\lvert{P}\rvert_{z}}$ for $z$ in $P$}.
\end{defi}

\begin{rem}
Let $s$ be a sort in $S$, and $P$ a term in $\mathrm{T}_{\Sigma}(X)_{s}$. Then we will denote by $\sbs{x}{\cdot}_{x\in X_{s}}(P)$ the mapping from $\prod_{x\in X_{s}}\mathrm{T}_{\Sigma}(X)_{s}^{\lvert{P}\rvert_{x}}$ to $\mathrm{T}_{\Sigma}(X)_{s}$ 
defined as follows:
$$
\sbs{x}{\cdot}_{x\in X_{s}}(P) 
\left\lbrace
\begin{array}{ccl}
{\prod_{x\in X_{s}}\mathrm{T}_{\Sigma}(X)_{s}^{\lvert{P}\rvert_{x}}}&\longrightarrow&
{\mathrm{T}_{\Sigma}(X)_{s}}\\
{\left((Q^{x}_{\alpha})_{\alpha\in \lvert{P}\rvert_{x}}\right)_{x\in X_{s}}}&\longmapsto&
{\sbs{x}{(Q^{x}_{\alpha})_{\alpha\in \lvert{P}\rvert_{x}}}_{x\in X_{s}}(P)}
\end{array}
\right.
$$
where the word $\sbs{x}{(Q^{x}_{\alpha})_{\alpha\in \lvert{P}\rvert_{x}}}_{x\in X_{s}}(P)$ in $\mathrm{T}_{\Sigma}(X)_{s}$ is obtained by substituting in $P$, for every $x\in X_{s}$ and every $\alpha\in \lvert{P}\rvert_x$, $Q^{x}_{\alpha}$ for the $i_{\alpha}$-th occurrence of $x$ in $P$. We will call $\sbs{x}{(Q^{x}_{\alpha})_{\alpha\in \lvert{P}\rvert_{x}}}_{x\in X_{s}}(P)$ the substitution of $(Q^{x}_{\alpha})_{\alpha\in\lvert{P}\rvert_{x}}$ for $x$ in $P$ for every $x$ in $X_{s}$, and $\sbs{x}{\cdot}_{x\in X_{s}}(P)$ the substitution operator for $P$.
\end{rem}

\begin{rem}\label{RDistr}
Let $(\mathbf{s},s)\in (S-\{\lambda\})\times S$, $\sigma\in \Sigma_{\mathbf{s},s}$, and $(P_{j})_{j\in\bb{\mathbf{s}}}\in\mathrm{T}_{\Sigma}(X)_{\mathbf{s}}$. Finally, let us consider the term $P=\sigma^{\mathbf{T}_{\Sigma}(X)}((P_{j})_{j\in\bb{\mathbf{s}}})$ in
$\mathrm{T}_{\Sigma}(X)_{s}$. 
Let $u$ be a sort in $S$, $z\in X_{u}$, and let
${(Q^{z}_{\alpha})_{\alpha\in \lvert{P}\rvert_{z}}}$ be a family of terms in
$\mathrm{T}_{\Sigma}(X)_{u}^{\bb{P}_{z}}$. Since, as a word, $P$ is
$(\sigma)\curlywedge\concat_{j\in\bb{\mathbf{s}}}P_{j}$ and the letter
$(z)$ is distinct from $(\sigma)$, the occurrences of $(z)$ in $P$ are
obtained by concatenating, in ascending order, the occurrences of $(z)$ in
the factors $P_{j}$. Hence
$\bb{P}_{z}=\sum_{j\in\bb{\mathbf{s}}}\bb{P_{j}}_{z}$ and, denoting by
$(Q^{j,z}_{\beta})_{\beta\in\bb{P_{j}}_{z}}$, for $j\in\bb{\mathbf{s}}$, the
corresponding partition of the family $(Q^{z}_{\alpha})_{\alpha\in\bb{P}_{z}}$,
it follows immediately from the definition of the substitution operator
that
\begin{equation}\label{EqDistr}
\sbs{z}{(Q^{z}_{\alpha})_{\alpha\in \lvert{P}\rvert_{z}}}\left(\sigma^{\mathbf{T}_{\Sigma}(X)}\left((P_{j})_{j\in\bb{\mathbf{s}}}\right)
\right)=\sigma^{\mathbf{T}_{\Sigma}(X)}\left(\left(
\sbs{z}{(Q^{j,z}_{\beta})_{\beta\in \lvert{P_{j}}\rvert_{z}}}(P_{j})
\right)_{j\in\bb{\mathbf{s}}}\right).
\tag{D}
\end{equation}
Conversely, every choice of families
$(Q^{j,z}_{\beta})_{\beta\in\bb{P_{j}}_{z}}$, for $j\in\bb{\mathbf{s}}$,
arises in this way from a unique family
$(Q^{z}_{\alpha})_{\alpha\in\bb{P}_{z}}$.
\end{rem}

The following two lemmas will be of interest later on.

\begin{lem}\label{LSubVar}
Let $Y$ be an $S$-sorted set. Let $u,s$ be sorts in $S$ and $z\in X_{u}$.
Let $P\in\mathrm{T}_{\Sigma}(X)_{s}$ and
$(Q^{z}_{\alpha})_{\alpha\in\bb{P}_{z}}\in
\mathrm{T}_{\Sigma}(Y)_{u}^{\bb{P}_{z}}$.
Then
$$
\sbs{z}{(Q^{z}_{\alpha})_{\alpha\in\bb{P}_{z}}}(P)\in
\mathrm{T}_{\Sigma}((X-\delta^{u,z})\cup Y)_{s}.
$$
In the above expression, the substitution is being computed
in $\mathrm{T}_{\Sigma}(X\cup Y)$. 
\end{lem}

\begin{proof}
The proof is done by induction on $(\coprod\mathrm{T}_{\Sigma}(X),\leq_{\mathbf{T}_{\Sigma}(X)})$.

\textsf{Base step of the Artinian induction.}

Let $(P,s)$ be a minimal element in $(\coprod\mathrm{T}_{\Sigma}(X),\leq_{\mathbf{T}_{\Sigma}(X)})$. We have, by Proposition~\ref{PArtOrd}, that $P$ has the form (1) $x$, for a unique variable $x\in X_{s}$, or (2) $\sigma^{\mathbf{T}_{\Sigma}(X)}$, for a unique constant operation symbol $\sigma\in \Sigma_{\lambda,s}$.

In case (1), we consider two subcases. If (1.1) $P=z$, so that $s=u$, the substituted term is
$Q^{z}_{0}\in\mathrm{T}_{\Sigma}(Y)_{u}$. If (1.2) $P=x$, for a variable $x\in
X_{s}$ with $x\neq z$, then the substituted term is $x$ itself and
$x\in(X-\delta^{u,z})_{s}$. 

In case (2), i.e., if $P=\sigma^{\mathbf{T}_{\Sigma}(X)}$, for
$\sigma\in\Sigma_{\lambda,s}$, the substituted term is $P$ itself, which
belongs to $\mathrm{T}_{\Sigma}((X-\delta^{u,z})\cup Y)_{s}$.

\textsf{Inductive step of the Artinian induction.}

Let $(P,s)$ be a non-minimal element in $(\coprod\mathrm{T}_{\Sigma}(X),\leq_{\mathbf{T}_{\Sigma}(X)})$. Then, by Proposition~\ref{PArtOrd}, there exists a unique $(\mathbf{s},s)\in (S^{\star}-\{\lambda\})\times S$, a unique operation symbol $\sigma\in\Sigma_{\mathbf{s},s}$ and a unique family of terms $(P_{j})_{j\in\bb{\mathbf{s}}}\in\mathrm{T}_{\Sigma}(X)_{\mathbf{s}}$ for which $P=\sigma^{\mathbf{T}_{\Sigma}(X)}((P_{j})_{j\in\bb{\mathbf{s}}})$. Assume that the statement holds for every immediate subterm of $P$, that is, for every $j\in\bb{\mathbf{s}}$, we have that the lemma holds for $P_{j}$. The statement follows from Equation~\eqref{EqDistr}, the induction hypothesis
applied to each $P_{j}$, and the fact that
$\mathrm{T}_{\Sigma}((X-\delta^{u,z})\cup Y)$ is closed under the structural operations of $\Sigma$.
\end{proof}

\begin{lem}\label{LSubSubt}
Let $u,s$ be sorts in $S$ and $z\in X_{u}$.
Let $P\in\mathrm{T}_{\Sigma}(X)_{s}$,
$(Q^{z}_{\alpha})_{\alpha\in\bb{P}_{z}}\in
\mathrm{T}_{\Sigma}(X)_{u}^{\bb{P}_{z}}$, and set
$
R=\sbs{z}{(Q^{z}_{\alpha})_{\alpha\in\bb{P}_{z}}}(P).
$
Then, for every
$(M,t)\in\coprod\mathrm{T}_{\Sigma}(X)$ with $(M,t)<(R,s)$, at least one of
the following two alternatives holds:
\begin{enumerate}
\item there exists an $\alpha\in\bb{P}_{z}$ such that
$(M,t)\leq(Q^{z}_{\alpha},u)$;
\item there exist a $(P',t)\in\coprod\mathrm{T}_{\Sigma}(X)$ with
$(P',t)<(P,s)$ and a family
$(Q'^{z}_{\beta})_{\beta\in\bb{P'}_{z}}\in
\{Q^{z}_{\alpha}\mid\alpha\in\bb{P}_{z}\}^{\bb{P'}_{z}}$ such that
$M=\sbs{z}{(Q'^{z}_{\beta})_{\beta\in\bb{P'}_{z}}}(P')$.
\end{enumerate}
\end{lem}

\begin{proof}
The proof is done by induction on $(\coprod\mathrm{T}_{\Sigma}(X),\leq_{\mathbf{T}_{\Sigma}(X)})$.

\textsf{Base step of the Artinian induction.}

Let $(P,s)$ be a minimal element in $(\coprod\mathrm{T}_{\Sigma}(X),\leq_{\mathbf{T}_{\Sigma}(X)})$. Consider the case $P=z$, so that $s=u$, then $R=Q^{z}_{0}$ and every $(M,t)$ with $(M,t)<(R,s)$ satisfies
alternative~(1). If the previous case does not apply, then $R=P$ and $(R,s)$ is minimal, so the statement holds vacuously.

\textsf{Inductive step of the Artinian induction.}

Let $(P,s)$ be a non-minimal element in $(\coprod\mathrm{T}_{\Sigma}(X),\leq_{\mathbf{T}_{\Sigma}(X)})$. Then, by Proposition~\ref{PArtOrd}, there exists a unique $(\mathbf{s},s)\in (S^{\star}-\{\lambda\})\times S$, a unique operation symbol $\sigma\in\Sigma_{\mathbf{s},s}$ and a unique family of terms $(P_{j})_{j\in\bb{\mathbf{s}}}\in\mathrm{T}_{\Sigma}(X)_{\mathbf{s}}$ for which $P=\sigma^{\mathbf{T}_{\Sigma}(X)}((P_{j})_{j\in\bb{\mathbf{s}}})$. Assume that the statement holds for every immediate subterm of $P$, that is, for every $j\in\bb{\mathbf{s}}$, we have that the lemma holds for $P_{j}$. By
Equation~\eqref{EqDistr}, we have that 
$R=\sigma^{\mathbf{T}_{\Sigma}(X)}((R_{j})_{j\in\bb{\mathbf{s}}})$, where,
for every $j\in\bb{\mathbf{s}}$,
$R_{j}=\sbs{z}{(Q^{j,z}_{\beta})_{\beta\in\bb{P_{j}}_{z}}}(P_{j})$ for the
corresponding partition of the family. Let $(M,t)<(R,s)$. By
Remark~\ref{RRestII}, $(M,t)\leq(R_{j},s_{j})$ for some $j\in\bb{\mathbf{s}}$.

If $M=R_{j}$, then alternative~(2) holds with $P'=P_{j}$ and the family
$(Q^{j,z}_{\beta})_{\beta\in\bb{P_{j}}_{z}}$, since
$(P_{j},s_{j})<(P,s)$, and  the members of
$(Q^{j,z}_{\beta})_{\beta\in\bb{P_{j}}_{z}}$ belong to
$\{Q^{z}_{\alpha}\mid\alpha\in\bb{P}_{z}\}$. If $(M,t)<(R_{j},s_{j})$, the
induction hypothesis applied to $P_{j}$ yields either alternative~(1)
directly, or alternative~(2) for some $(P'',t)<(P_{j},s_{j})$; in this
latter case $(P'',t)<(P,s)$ holds by transitivity, and the members of the
corresponding family belong to
$\{Q^{j,z}_{\beta}\mid\beta\in\bb{P_{j}}_{z}\}\subseteq
\{Q^{z}_{\alpha}\mid\alpha\in\bb{P}_{z}\}$.
\end{proof}

We next state that every $\Sigma$-homomorphism from $\mathbf{T}_{\Sigma}(X)$ to any $\Sigma$-algebra $\mathbf{A}$ is compatible with the substitution operator.

\begin{propC}[{\cite[Lemma~3.10]{CVCL20}}]
\label{LGlobSubstOp}
Let $s\in S$, $P\in\mathrm{T}_{\Sigma}(X)_{s}$, $\mathbf{A}$ a $\Sigma$-algebra, and $g$ a homomorphism from $\mathbf{T}_{\Sigma}(X)$ to $\mathbf{A}$. Let 
$\left((Q^{x,t}_{\alpha})_{\alpha\in \bb{P}_{x}}\right)_{(x,t)\in \coprod X}$, $\left((R^{x,t}_{\alpha})_{\alpha\in \bb{P}_{x}}\right)_{(x,t)\in \coprod X}$ in 
$ \prod_{(x,t)\in \coprod X}\mathrm{T}_{\Sigma}(X)_{t}^{\bb{P}_{x}}$. 
If, for every $(x,t)\in\coprod X$ and every $\alpha\in\bb{P}_{x}$, it holds that $g_{t}(Q^{x,t}_{\alpha})=g_{t}(R^{x,t}_{\alpha})$, then 
$$
g_{s}\left(\sbs{(x,t)}{(Q^{x,t}_{\alpha})_{\alpha\in \bb{P}_{x}}}_{(x,t)\in \coprod X}(P)\right)
=
g_{s}\left(\sbs{(x,t)}{(R^{x,t}_{\alpha})_{\alpha\in \bb{P}_{x}}}_{(x,t)\in \coprod X}(P)\right).
$$
\end{propC}

The following result is a direct corollary of Proposition~\ref{LGlobSubstOp}.

\begin{cor}\label{CSubInv}
Let $\mathbf{A}$ be a $\Sigma$-algebra, and $g$ a homomorphism from $\mathbf{T}_{\Sigma}(X)$ to $\mathbf{A}$. Let $u,s$ be sorts in $S$, $z\in X_{u}$, 
$P\in\mathrm{T}_{\Sigma}(X)_{s}$,
and $(Q^{z}_{\alpha})_{\alpha\in\bb{P}_{z}}\in
\mathrm{T}_{\Sigma}(X)_{u}^{\bb{P}_{z}}$ a family such that, for every
$\alpha\in\bb{P}_{z}$, $g_{u}(Q^{z}_{\alpha})=g_{u}(z)$. Then
$$
g_{s}\left(\sbs{z}{(Q^{z}_{\alpha})_{\alpha\in\bb{P}_{z}}}(P)\right)=g_{s}(P).
$$
\end{cor}
\begin{proof}
First note that substituting, at every occurrence of $z$ in $P$, the
variable $z$ itself leaves $P$ unchanged, i.e.,
$\sbs{z}{(z)_{\alpha\in\bb{P}_{z}}}(P)=P$. Moreover, the substitution of a
family for the single variable $z$ is the particular case of the global
substitution operator for $P$ in which, at every $(x,t)\in\coprod X$ with
$(x,t)\neq(z,u)$, the constant family $(x)_{\alpha\in\bb{P}_{x}}$ is
substituted. Since, for every $\alpha\in\bb{P}_{z}$,
$g_{u}(Q^{z}_{\alpha})=g_{u}(z)$, Proposition~\ref{LGlobSubstOp}, applied to the two
families just described, yields
\[g_{s}(\sbs{z}{(Q^{z}_{\alpha})_{\alpha\in\bb{P}_{z}}}(P))
=g_{s}(\sbs{z}{(z)_{\alpha\in\bb{P}_{z}}}(P))=g_{s}(P).
\qedhere
\]
\end{proof}

The next lemma isolates the decomposition of a term along
its subterms of a given sort and a given value under a homomorphism.

\begin{lem}\label{LCollapse}
Let $\mathbf{A}$ be a $\Sigma$-algebra and $g$ a homomorphism from
$\mathbf{T}_{\Sigma}(X)$ to $\mathbf{A}$. Let $u\in S$ and $z\in X_{u}$.
Then, for every $s\in S$ and
every $R\in\mathrm{T}_{\Sigma}(X)_{s}$ with $(R,s)\notin\mathrm{Min}$,
there exist a term $P\in\mathrm{T}_{\Sigma}(X)_{s}$ and a family
$(Q^{z}_{\alpha})_{\alpha\in\bb{P}_{z}}\in
\mathrm{T}_{\Sigma}(X)_{u}^{\bb{P}_{z}}$ such that
\begin{enumerate}
\item $R=\sbs{z}{(Q^{z}_{\alpha})_{\alpha\in\bb{P}_{z}}}(P)$;
\item $g_{s}(P)=g_{s}(R)$;
\item for every $\alpha\in\bb{P}_{z}$, $(Q^{z}_{\alpha},u)<(R,s)$ and
$g_{u}(Q^{z}_{\alpha})=g_{u}(z)$; and
\item for every $(M,t)\in\coprod\mathrm{T}_{\Sigma}(X)$ with
$(M,t)<(P,s)$ and $(M,t)\notin\mathrm{Min}$, the following two
properties hold:
\begin{enumerate}
\item it is not the case that both $t=u$ and $g_{u}(M)=g_{u}(z)$; and
\item there exists a $(N,t)\in\coprod\mathrm{T}_{\Sigma}(X)$ with
$(N,t)<(R,s)$, $(N,t)\notin\mathrm{Min}$ and
$g_{t}(N)=g_{t}(M)$.
\end{enumerate}
\end{enumerate}
\end{lem}

\begin{proof}
We first define an $S$-sorted mapping $c\colon \mathrm{T}_{\Sigma}(X)\longrightarrow \mathrm{T}_{\Sigma}(X)$. For every $s\in S$, the $s$-th coordinate $c_{s}$ of $c$ is defined recursively as follows:
$$
c_{s}
\left\lbrace
\begin{array}{ccl}
{\mathrm{T}_{\Sigma}(X)_{s}}&\longrightarrow&
{\mathrm{T}_{\Sigma}(X)_{s}}\\
{P}&\longmapsto&
{\begin{cases}
 z,&\text{if }P\in\mathrm{Min}, s=u \text{ and } g_{u}(P)=g_{u}(z); \\
 P, &\text{if }P\in\mathrm{Min}\text{ and }(s\neq u \text{ or } g_{u}(P)\neq g_{u}(z));\\
 {\sigma}^{\mathbf{T}_{\Sigma}(X)}((c_{s_{j}}(P_{j}))_{j\in\lvert{\mathbf{s}}\rvert}),
 &\text{if }P = \sigma^{\mathbf{T}_{\Sigma}(X)}((P_{j})_{j\in\lvert{\mathbf{s}}\rvert}). 
\end{cases}}
\end{array}
\right.
$$

This completes the definition of $c$.

We first prove the following claim.

\begin{clm}\label{CCollapse}
Let $s\in S$ and $P\in \mathrm{T}_{\Sigma}(X)_{s}$. The following three properties hold:
\begin{enumerate}
\item[(i)] $g_{s}(c_{s}(P))=g_{s}(P)$.
\item[(ii)] There exists a family
$(Q^{z}_{\beta})_{\beta\in\bb{c_{s}(P)}_{z}}$ such that
$P=\sbs{z}{(Q^{z}_{\beta})_{\beta\in\bb{c_{s}(P)}_{z}}}(c_{s}(P))$ and, for
every $\beta\in\bb{c_{s}(P)}_{z}$, $(Q^{z}_{\beta},u)\leq(P,s)$ and
$g_{u}(Q^{z}_{\beta})=g_{u}(z)$.
\item[(iii)] For every $(M,t)\leq(c_{s}(P),s)$ with
$(M,t)\notin\mathrm{Min}$, the following two properties hold:
\begin{enumerate}
\item[(a)] it is not the case that both $t=u$ and
$g_{u}(M)=g_{u}(z)$; and
\item[(b)] there exists a $(N,t)\leq(P,s)$ with
$(N,t)\notin\mathrm{Min}$ and $g_{t}(N)=g_{t}(M)$.
\end{enumerate}
\end{enumerate}
\end{clm}

\begin{proof}
The proof is done by Artinian induction on $(\coprod\mathrm{T}_{\Sigma}(X),\leq_{\mathbf{T}_{\Sigma}(X)})$.

\textsf{Base step of the Artinian induction.}

Let $(P,s)$ be a minimal element in $(\coprod\mathrm{T}_{\Sigma}(X),\leq_{\mathbf{T}_{\Sigma}(X)})$. Consider the case $s=u$ and $g_{u}(P)=g_{u}(z)$, then $c_{s}(P)=z$. Property~(i) holds
since $g_{u}(z)=g_{u}(P)$. For property~(ii), $\bb{z}_{z}=1$ and the family
$(P)$ satisfies $\sbs{z}{(P)}(z)=P$, $(P,u)\leq(P,u)$ and
$g_{u}(P)=g_{u}(z)$. Property~(iii) holds vacuously, since the only
subterm of $z$ is $z$ itself, which is minimal. If the previous case does not apply, then
$c_{s}(P)=P$ and properties~(i) and~(iii) are immediate. For
property~(ii), note that $P\neq z$, for otherwise the first case would
apply; hence $\bb{P}_{z}=0$ and the empty family satisfies the statement.

\textsf{Inductive step of the Artinian induction.}

Let $(P,s)$ be a non-minimal element in $(\coprod\mathrm{T}_{\Sigma}(X),\leq_{\mathbf{T}_{\Sigma}(X)})$. Then, by Proposition~\ref{PArtOrd}, there exists a unique $(\mathbf{s},s)\in (S^{\star}-\{\lambda\})\times S$, a unique operation symbol $\sigma\in\Sigma_{\mathbf{s},s}$ and a unique family of terms $(P_{j})_{j\in\bb{\mathbf{s}}}\in\mathrm{T}_{\Sigma}(X)_{\mathbf{s}}$ for which $P=\sigma^{\mathbf{T}_{\Sigma}(X)}((P_{j})_{j\in\bb{\mathbf{s}}})$. Assume that the statement holds for every immediate subterm of $P$, that is, for every $j\in\bb{\mathbf{s}}$, Claim~\ref{CCollapse} holds for $P_{j}$.

Note that 
$c_{s}(P)=\sigma^{\mathbf{T}_{\Sigma}(X)}
((c_{s_{j}}(P_{j}))_{j\in\bb{\mathbf{s}}})$. Property~(i) follows from the
fact that $g$ is a homomorphism and the induction hypothesis~(i) for the
$P_{j}$. Property~(ii) follows from Remark~\ref{RDistr} and the induction
hypothesis~(ii) for the $P_{j}$: the concatenation of the families
provided for the $P_{j}$ satisfies, by Equation~\eqref{EqDistr}, the
required equality, and its members satisfy
$(Q^{z}_{\beta},u)\leq(P_{j},s_{j})\leq(P,s)$ and $g_{u}(Q^{z}_{\beta})=g_{u}(z)$.
For property~(iii), let $(M,t)\leq(c_{s}(P),s)$ with
$(M,t)\notin\mathrm{Min}$. If $M=c_{s}(P)$, so that $t=s$, then, on the
one hand, $g_{s}(M)=g_{s}(P)$ by property~(i) and, since the first case
does not apply, it is not the case that both $s=u$ and
$g_{u}(P)=g_{u}(z)$; on the other hand, $(P,s)\leq(P,s)$,
$(P,s)\notin\mathrm{Min}$ and $g_{s}(P)=g_{s}(M)$, so $N=P$ works.
If $(M,t)<(c_{s}(P),s)$, then, by Remark~\ref{RRestII},
$(M,t)\leq(c_{s_{j}}(P_{j}),s_{j})$ for some $j\in\bb{\mathbf{s}}$, and
the induction hypothesis~(iii) for $P_{j}$ provides the desired
conclusions, since $(N,t)\leq(P_{j},s_{j})\leq(P,s)$.

This proves Claim~\ref{CCollapse}.
\end{proof}

We can now prove Lemma~\ref{LCollapse}. Let $R\in\mathrm{T}_{\Sigma}(X)_{s}$ with
$(R,s)\notin\mathrm{Min}$. By Proposition~\ref{PTermChar},
$R=\sigma^{\mathbf{T}_{\Sigma}(X)}((R_{j})_{j\in\bb{\mathbf{s}}})$, for a
unique $\mathbf{s}\in S^{\star}-\{\lambda\}$, a unique
$\sigma\in\Sigma_{\mathbf{s},s}$ and a unique family
$(R_{j})_{j\in\bb{\mathbf{s}}}\in\mathrm{T}_{\Sigma}(X)_{\mathbf{s}}$.
Define
$$
P=\sigma^{\mathbf{T}_{\Sigma}(X)}
((c_{s_{j}}(R_{j}))_{j\in\bb{\mathbf{s}}})
$$
and let $(Q^{z}_{\alpha})_{\alpha\in\bb{P}_{z}}$ be the concatenation of the
families provided by property~(ii) for the $R_{j}$. Then condition~(1)
follows from Equation~\eqref{EqDistr} and property~(ii); condition~(2)
follows from property~(i) and the fact that $g$ is a homomorphism;
condition~(3) follows from property~(ii), since the members of the family
satisfy $(Q^{z}_{\alpha},u)\leq(R_{j},s_{j})<(R,s)$. Finally, for
condition~(4), let $(M,t)<(P,s)$ with $(M,t)\notin\mathrm{Min}$. By
Remark~\ref{RRestII}, $(M,t)\leq(c_{s_{j}}(R_{j}),s_{j})$ for some
$j\in\bb{\mathbf{s}}$, and property~(iii) for $R_{j}$ yields both~(4)(a)
and, since $(N,t)\leq(R_{j},s_{j})<(R,s)$, also~(4)(b).
\end{proof}

From now on, following a strongly rooted tradition in the fields of formal languages and automata, we agree to call \emph{languages} the subsets of the underlying many-sorted set of the free many-sorted algebra on a many-sorted set. Since it will be used in the following definition we recall that, as a particular case of Proposition~\ref{PSubsetAlg}, for an $S$-sorted set $X$, we have the power $\Sigma$-algebra $\mathbf{T}_{\Sigma}(X)^{\wp}$ associated with $\mathbf{T}_{\Sigma}(X)$, which has as underlying $S$-sorted set 
$$
\mathrm{T}_{\Sigma}(X)^{\wp} = (\mathrm{Sub}(\mathrm{T}_{\Sigma}(X)_{s}))_{s\in S}
$$ 
and, for every $(\mathbf{s},s)\in S^{\star}\times S$ and every $\sigma\in \Sigma_{\mathbf{s},s}$, as structural operation associated with $\sigma$, the mapping 
$$
\begin{adjustbox}{max width=\linewidth, keepaspectratio}
$
\sigma^{\mathbf{T}_{\Sigma}(X)^{\wp}}
\left\lbrace
\begin{array}{ccl}
{\prod_{j\in \lvert{\mathbf{s}}\rvert}\mathrm{Sub}(\mathrm{T}_{\Sigma}(X)_{s_{j}})}&\longrightarrow&
{\mathrm{Sub}(\mathrm{T}_{\Sigma}(X)_{s})}\\
{(L_{j})_{j\in \lvert{\mathbf{s}}\rvert}}&\longmapsto&
{\{\sigma^{\mathbf{T}_{\Sigma}(X)}((P_{j})_{j\in \lvert{\mathbf{s}}\rvert})\mid (P_{j})_{j\in \lvert{\mathbf{s}}\rvert}\in \prod_{j\in \lvert{\mathbf{s}}\rvert}L_{j}\}}
\end{array}
\right.
$
\end{adjustbox}
$$ 

We remind the reader that, according to what is established in Proposition~\ref{PTermChar}, we have let, for abbreviation, for every $(\mathbf{s},s)\in S^{\star}\times S$ and every $\sigma\in\Sigma_{\mathbf{s},s}$, $\sigma^{\mathbf{T}_{\Sigma}(X)}$ stand for the structural operation of $\mathbf{T}_{\Sigma}(X)$ associated with $\sigma$.

\begin{defi}\label{DSubs}
For every $t\in S$, let $(L_{x})_{x\in X_{t}}$ be a mapping from $X_{t}$ to $\mathrm{T}_{\Sigma}(X)^{\wp}_{t} = \mathrm{Sub}(\mathrm{T}_{\Sigma}(X)_{t})$, also written, in this context, as $\sbs{x}{L_{x}}_{x\in X_{t}}$. Then, for the $S$-sorted mapping $\left(\sbs{x}{L_{x}}_{x\in X_{t}}\right)_{t\in S}$ from $X$ to $\mathrm{T}_{\Sigma}(X)^{\wp}$, we will denote by $\left(\left(\sbs{x}{L_{x}}_{x\in X_{t}}\right)_{t\in S}\right)^{\sharp}$ the unique homomorphism from $\mathbf{T}_{\Sigma}(X)$ to $\mathbf{T}_{\Sigma}(X)^{\wp}$ such that
$$
\left(\left(\sbs{x}{L_{x}}_{x\in X_{t}}\right)_{t\in S}\right)^{\sharp}\circ \eta^{X} = \left(\sbs{x}{L_{x}}_{x\in X_{t}}\right)_{t\in S}.
$$
Let $u$ be a sort in $S$, $z$ an element of $X_{u}$, and $L\in \mathrm{T}_{\Sigma}(X)^{\wp}_{u}$. Then, for the $S$-sorted mapping $\sbs{z}{L} = (\sbs{z}{L}_{t})_{t\in S}$ from $X$ to $\mathrm{T}_{\Sigma}(X)^{\wp}$ defined, for every $t\in S$, as: 
\begin{enumerate}
\item if $t = u$, then 
$$
\sbs{z}{L}_{u}
\left\lbrace
\begin{array}{ccl}
{X_{u}}&\longrightarrow&
{\mathrm{Sub}(\mathrm{T}_{\Sigma}(X)_{u})}\\
{y}&\longmapsto&
{\begin{cases}
 L, &\text{if $y = z$;}\\
 \{y\}, &\text{if $y\in X_{u}-\{z\}$,}
 \end{cases}}
 \end{array}
 \right.
$$
\item while, if $t \neq u$, then
$$
\sbs{z}{L}_{t}
\left\lbrace
\begin{array}{ccl}
{X_{t}}&\longrightarrow&
{\mathrm{Sub}(\mathrm{T}_{\Sigma}(X)_{t})}\\
{x}&\longmapsto&
{\{x\},}
\end{array}
\right.
$$
\end{enumerate}
we will denote by $\sbs{z}{L}^{\sharp}$ the unique homomorphism from $\mathbf{T}_{\Sigma}(X)$ to $\mathbf{T}_{\Sigma}(X)^{\wp}$ such that $\sbs{z}{L}^{\sharp}\circ \eta^ {X} = \sbs{z}{L}$.
\end{defi}

\begin{rem}
We recall that the homomorphism $\left(\left(\sbs{x}{L_{x}}_{x\in X_{t}}\right)_{t\in S}\right)^{\sharp}$ from $\mathbf{T}_{\Sigma}(X)$ to $\mathbf{T}_{\Sigma}(X)^{\wp}$ introduced in Definition~\ref{DSubs} is defined, by algebraic recursion, as follows. Let $s$ be a sort in $S$ and $P\in\mathrm{T}_{\Sigma}(X)_{s}$. Then, according to Proposition~\ref{PTermChar}, we know that $P$ has the form (1) $z$, for a unique $z\in X_{s}$, (2) $\sigma^{\mathbf{T}_{\Sigma}(X)}$, for a unique $\sigma\in\Sigma_{\lambda, s}$, or (3) $\sigma^{\mathbf{T}_{\Sigma}(X)}((P_{j})_{j\in \lvert{\mathbf{s}}\rvert})$, for a unique $\mathbf{s}\in S^{\star}-\{\lambda\}$, a unique $\sigma\in \Sigma_{\mathbf{s},s}$, and a unique family $(P_{j})_{j\in\lvert{\mathbf{s}}\rvert}\in\mathrm{T}_{\Sigma}(X)_{\mathbf{s}}$.

In case (1), we have $\left(\left(\sbs{x}{L_{x}}_{x\in X_{t}}\right)_{t\in S}\right)^{\sharp}_{s}(z) = L_{z}$.

In case (2), we have $\left(\left(\sbs{x}{L_{x}}_{x\in X_{t}}\right)_{t\in S}\right)^{\sharp}_{s}(\sigma^{\mathbf{T}_{\Sigma}(X)}) = \{\sigma^{\mathbf{T}_{\Sigma}(X)}\}$.

Finally, in case (3), and under the hypothesis that, for every $j\in \lvert{\mathbf{s}}\rvert$, the subset $\left(\left(\sbs{x}{L_{x}}_{x\in X_{t}}\right)_{t\in S}\right)^{\sharp}_{s_{j}}(P_{j})$ of $\mathrm{T}_{\Sigma}(X)_{s_{j}}$ has been defined, we have 
\begin{flushleft}
$
\left(\left(\sbs{x}{L_{x}}_{x\in X_{t}}\right)_{t\in S}\right)^{\sharp}_{s}(\sigma^{\mathbf{T}_{\Sigma}(X)}((P_{j})_{j\in \lvert{\mathbf{s}}\rvert}))
$
\allowdisplaybreaks
\begin{align*}
&=
\sigma^{\mathbf{T}_{\Sigma}(X)^{\wp}}\left(\left(\left(\left(\sbs{x}{L_{x}}_{x\in X_{t}}\right)_{t\in S}\right)^{\sharp}_{s_{j}}(P_{j})\right)_{j\in \lvert{\mathbf{s}}\rvert}\right)
\\&=
\textstyle\Big\lbrace\sigma^{\mathbf{T}_{\Sigma}(X)}((Q_{j})_{j\in \lvert{\mathbf{s}}\rvert})\,\,\Big|\,\, (Q_{j})_{j\in \lvert{\mathbf{s}}\rvert}\in \prod_{j\in \lvert{\mathbf{s}}\rvert}\left(\left(\sbs{x}{L_{x}}_{x\in X_{t}}\right)_{t\in S}\right)^{\sharp}_{s_{j}}(P_{j})\Big\rbrace.
\end{align*}
\end{flushleft}
\end{rem}

We finally consider the canonical extension of the homomorphism introduced in Definition~\ref{DSubs} to the corresponding power algebra.

\begin{defi}
Let $\left(\sbs{x}{L_{x}}_{x\in X_{t}}\right)_{t\in S}$ be an $S$-sorted mapping from $X$ to $\mathrm{T}_{\Sigma}(X)^{\wp}$. Then we will denote by $\left(\left(\sbs{x}{L_{x}}_{x\in X_{t}}\right)_{t\in S}\right)^{\sharp\mathsf{p}}$ the $S$-sorted mapping from $\mathrm{T}_{\Sigma}(X)^{\wp}$ to $\mathrm{T}_{\Sigma}(X)^{\wp}$ defined, for every $s\in S$, as: 
$$
\left(\left(\sbs{x}{L_{x}}_{x\in X_{t}}\right)_{t\in S}\right)^{\sharp\mathsf{p}}_{s}\left\lbrace
\begin{array}{ccl}
{\mathrm{T}_{\Sigma}(X)^{\wp}_{s}}&\longrightarrow&
{\mathrm{T}_{\Sigma}(X)^{\wp}_{s}}\\
{K}&\longmapsto&
{\bigcup_{P\in K}\left(\left(\sbs{x}{L_{x}}_{x\in X_{t}}\right)_{t\in S}\right)^{\sharp}_{s}(P)}
\end{array}
\right.
$$

Thus, $\left(\left(\sbs{x}{L_{x}}_{x\in X_{t}}\right)_{t\in S}\right)^{\sharp\mathsf{p}}$ is the canonical extension of the underlying $S$-sorted mapping of the $\Sigma$-homomorphism $\left(\left(\sbs{x}{L_{x}}_{x\in X_{t}}\right)_{t\in S}\right)^{\sharp}$ from $\mathbf{T}_{\Sigma}(X)$ to $\mathbf{T}_{\Sigma}(X)^{\wp}$.

Let $u$ be a sort in $S$, $z\in X_{u}$, and $L\in \mathrm{T}_{\Sigma}(X)^{\wp}_{u}$. Then we will denote by $\sbs{z}{L}^{\sharp\mathsf{p}}$ the canonical extension of the underlying mapping of the $\Sigma$-homomorphism $\sbs{z}{L}^{\sharp}$ from $\mathbf{T}_{\Sigma}(X)$ to $\mathbf{T}_{\Sigma}(X)^{\wp}$.
\end{defi}

The first lemma describes the homomorphisms of Definition~\ref{DSubs} in
terms of the substitution of families.

\begin{lem}\label{LSubFam}
Let $u,s$ be sorts in $S$, $z\in X_{u}$, $L\subseteq \mathrm{T}_{\Sigma}(X)_{u}$, and
$P\in\mathrm{T}_{\Sigma}(X)_{s}$. Then
$$
\sbs{z}{L}^{\sharp}_{s}(P)=\left\{
\sbs{z}{(Q^{z}_{\alpha})_{\alpha\in\bb{P}_{z}}}(P)
\,\middle|\,
(Q^{z}_{\alpha})_{\alpha\in\bb{P}_{z}}\in L^{\bb{P}_{z}}
\right\}.
$$
Consequently, for every $K\subseteq\mathrm{T}_{\Sigma}(X)_{s}$, the set
$\sbs{z}{L}^{\sharp\mathsf{p}}_{s}(K)$ is equal to
$$
\textstyle
\sbs{z}{L}^{\sharp\mathsf{p}}_{s}(K)=\bigcup_{P\in K}\sbs{z}{L}^{\sharp}_{s}(P).
$$
\end{lem}

\begin{proof}
The proof is done by Artinian induction on $(\coprod\mathrm{T}_{\Sigma}(X),\leq_{\mathbf{T}_{\Sigma}(X)})$.

\textsf{Base step of the Artinian induction.}

Let $(P,s)$ be a minimal element in $(\coprod\mathrm{T}_{\Sigma}(X),\leq_{\mathbf{T}_{\Sigma}(X)})$. Consider the case $P=z$, so that $s=u$, then, by Definition~\ref{DSubs},
$\sbs{z}{L}^{\sharp}_{u}(z)=L$; on the other hand, $\bb{z}_{z}=1$ and
$\sbs{z}{(Q^{z}_{0})}(z)=Q^{z}_{0}$, so that, letting $Q^{z}_{0}$ range over $L$, the
right-hand side is also $L$. 

If the previous case does not apply, then $\bb{P}_{z}=0$ and, by
Definition~\ref{DSubs}, $\sbs{z}{L}^{\sharp}_{s}(P)=\{P\}$; on the other
hand, the only family is the empty one and its substitution leaves $P$
unchanged, so the right-hand side is also $\{P\}$.

\textsf{Inductive step of the Artinian induction.}

Let $(P,s)$ be a non-minimal element in $(\coprod\mathrm{T}_{\Sigma}(X),\leq_{\mathbf{T}_{\Sigma}(X)})$. Then, by Proposition~\ref{PArtOrd}, there exists a unique $(\mathbf{s},s)\in (S^{\star}-\{\lambda\})\times S$, a unique operation symbol $\sigma\in\Sigma_{\mathbf{s},s}$ and a unique family of terms $(P_{j})_{j\in\bb{\mathbf{s}}}\in\mathrm{T}_{\Sigma}(X)_{\mathbf{s}}$ for which $P=\sigma^{\mathbf{T}_{\Sigma}(X)}((P_{j})_{j\in\bb{\mathbf{s}}})$. Assume that the statement holds for every immediate subterm of $P$, that is, for every $j\in\bb{\mathbf{s}}$, the lemma holds for $P_{j}$.

Since $\sbs{z}{L}^{\sharp}$ is a $\Sigma$-homomorphism into
$\mathbf{T}_{\Sigma}(X)^{\wp}$, by Proposition~\ref{PSubsetAlg} we have
$$
\sbs{z}{L}^{\sharp}_{s}(P)
=\left\{
\sigma^{\mathbf{T}_{\Sigma}(X)}((R_{j})_{j\in\bb{\mathbf{s}}})
\,\middle|\,
(R_{j})_{j\in\bb{\mathbf{s}}}\in
\textstyle\prod_{j\in\bb{\mathbf{s}}}\sbs{z}{L}^{\sharp}_{s_{j}}(P_{j})
\right\},
$$
which, by the induction hypothesis applied to each $P_{j}$ and
Remark~\ref{RDistr}, is precisely the right-hand side of the statement. 

The
last assertion follows from the definition of
$\sbs{z}{L}^{\sharp\mathsf{p}}$.
\end{proof}

\begin{rem}\label{RMon}
From Lemma~\ref{LSubFam} it is evident that, 
for every $u,s\in S$ and every $z\in X_{u}$, 
the assignment
from 
$\mathrm{T}_{\Sigma}(X)^{\wp}_{u} \times \mathrm{T}_{\Sigma}(X)^{\wp}_{s}$
to 
$\mathrm{T}_{\Sigma}(X)^{\wp}_{s}$, that sends a pair 
$(L,K)$ to $\sbs{z}{L}^{\sharp\mathsf{p}}_{s}(K)$, is monotone in both
of its arguments, i.e., if $L\subseteq L'$ and $K\subseteq K'$, then
$\sbs{z}{L}^{\sharp\mathsf{p}}_{s}(K)\subseteq
\sbs{z}{L'}^{\sharp\mathsf{p}}_{s}(K')$.
\end{rem}

\begin{lem}\label{LSubComp}
Let $u,s$ be sorts in $S$, $z\in X_{u}$, $L,L'\subseteq \mathrm{T}_{\Sigma}(X)_{u}$, and
$K\subseteq\mathrm{T}_{\Sigma}(X)_{s}$. Then, the following inclusion holds
$$
\sbs{z}{\sbs{z}{L}^{\sharp\mathsf{p}}_{u}(L')}^{\sharp\mathsf{p}}_{s}(K)
\subseteq
\sbs{z}{L}^{\sharp\mathsf{p}}_{s}\left(
\sbs{z}{L'}^{\sharp\mathsf{p}}_{s}(K)\right).
$$
\end{lem}

\begin{proof}
Both sides are completely additive in $K$, by the definition of the
operators $(\cdot)^{\sharp\mathsf{p}}$ (see Proposition~\ref{PSubFunct} and
Definition~\ref{DSubs}); it therefore suffices to prove, for every $s\in S$ and every
$P\in\mathrm{T}_{\Sigma}(X)_{s}$, the inclusion
$$
\sbs{z}{\sbs{z}{L}^{\sharp\mathsf{p}}_{u}(L')}^{\sharp}_{s}(P)
\subseteq
\sbs{z}{L}^{\sharp\mathsf{p}}_{s}\left(
\sbs{z}{L'}^{\sharp}_{s}(P)\right).
$$

The proof is done by Artinian induction on $(\coprod\mathrm{T}_{\Sigma}(X),\leq_{\mathbf{T}_{\Sigma}(X)})$.

\textsf{Base step of the Artinian induction.}

Let $(P,s)$ be a minimal element in $(\coprod\mathrm{T}_{\Sigma}(X),\leq_{\mathbf{T}_{\Sigma}(X)})$. Consider the case $P=z$, so that $s=u$, then, by Definition~\ref{DSubs}, the left-hand
side is $\sbs{z}{L}^{\sharp\mathsf{p}}_{u}(L')$ and the right-hand side is
$\sbs{z}{L}^{\sharp\mathsf{p}}_{s}(\sbs{z}{L'}^{\sharp}_{u}(z))
=\sbs{z}{L}^{\sharp\mathsf{p}}_{s}(L')$, so equality holds.

If the previous case does not apply, then
both sides are equal to $\{P\}$.

\textsf{Inductive step of the Artinian induction.}

Let $(P,s)$ be a non-minimal element in $(\coprod\mathrm{T}_{\Sigma}(X),\leq_{\mathbf{T}_{\Sigma}(X)})$. Then, by Proposition~\ref{PArtOrd}, there exists a unique $(\mathbf{s},s)\in (S^{\star}-\{\lambda\})\times S$, a unique operation symbol $\sigma\in\Sigma_{\mathbf{s},s}$ and a unique family of terms $(P_{j})_{j\in\bb{\mathbf{s}}}\in\mathrm{T}_{\Sigma}(X)_{\mathbf{s}}$ for which $P=\sigma^{\mathbf{T}_{\Sigma}(X)}((P_{j})_{j\in\bb{\mathbf{s}}})$. Assume that the statement holds for every immediate subterm of $P$, that is, for every $j\in\bb{\mathbf{s}}$, the lemma holds for $P_{j}$.

We
first note that, by Proposition~\ref{PSubsetAlg}, the structural operation
$\sigma^{\mathbf{T}_{\Sigma}(X)^{\wp}}$ is monotone and completely
distributes over unions in each of its arguments. Since
$\sbs{z}{L'}^{\sharp}$ and $\sbs{z}{L}^{\sharp}$ are
$\Sigma$-homomorphisms into $\mathbf{T}_{\Sigma}(X)^{\wp}$ and
$\sbs{z}{L}^{\sharp\mathsf{p}}$ is the completely additive extension of
$\sbs{z}{L}^{\sharp}$, we obtain
\allowdisplaybreaks
\begin{align*}
\sbs{z}{L}^{\sharp\mathsf{p}}_{s}\left(
\sbs{z}{L'}^{\sharp}_{s}(P)\right)
&=\sbs{z}{L}^{\sharp\mathsf{p}}_{s}\left(
\sigma^{\mathbf{T}_{\Sigma}(X)^{\wp}}\left(\left(
\sbs{z}{L'}^{\sharp}_{s_{j}}(P_{j})
\right)_{j\in\bb{\mathbf{s}}}\right)\right)
\\&=\sigma^{\mathbf{T}_{\Sigma}(X)^{\wp}}\left(\left(
\sbs{z}{L}^{\sharp\mathsf{p}}_{s_{j}}\left(
\sbs{z}{L'}^{\sharp}_{s_{j}}(P_{j})\right)
\right)_{j\in\bb{\mathbf{s}}}\right).
\end{align*}
The statement follows from the induction hypothesis and the monotonicity of
$\sigma^{\mathbf{T}_{\Sigma}(X)^{\wp}}$.
\end{proof}

\subsection{
\texorpdfstring
{Iterations}
{Iterations}
}\label{SSII}

We next introduce the notion of iteration of a language with respect to a variable and prove a technical lemma that will be of interest for the proof of the main result of this paper.

\begin{defi}\label{DIter}
Let $s$ be a sort in $S$, $z\in X_{s}$, and $L\subseteq \mathrm{T}_{\Sigma}(X)_{s}$. The \emph{$z$-iteration of $L$} is the language
$$
\textstyle
L^{\star\, z} = \bigcup_{i\in\mathbb{N}}L^{i\,z},
$$
where $(L^{i\,z})_{i\in \mathbb{N}}$ is the family of subsets of $\mathrm{T}_{\Sigma}(X)_{s}$ defined recursively as follows:
$$
L^{0\,z}=\{z\},\, \text{and, for } i\in \mathbb{N},\, L^{i+1\,z} = L^{i\,z}\cup \sbs{z}{L^{i\,z}}^{\sharp\mathsf{p}}_{s}(L).
$$
\end{defi}

\begin{rem}\label{RIter}
The language $L^{\star\, z}$ is obtained as follows. First include $z$. New members of $L^{\star\, z}$ are obtained by substituting in some term of $L$, for every occurrence of $z$, some term already known to be in $L^{\star\, z}$. Let us note that $L^{1\,z} = L\cup \{z\}$ and that $(L^{i\,z})_{i\in \mathbb{N}}$ is an ascending chain of subsets of $\mathrm{T}_{\Sigma}(X)_{s}$, i.e., that, for every $i\in \mathbb{N}$, $L^{i\,z}\subseteq L^{i+1\,z}$.
\end{rem}

We next prove a technical lemma about $z$-iterations.

\begin{lem}\label{LItAbs}
Let $s$ be a sort in $S$, $z\in X_{s}$, and $L\subseteq \mathrm{T}_{\Sigma}(X)_{s}$. Then, 
\[\sbs{z}{L^{\star\, z}}^{\sharp\mathsf{p}}_{s}(L)\subseteq L^{\star\, z}.\]
\end{lem}

\begin{proof}
Let $R\in\sbs{z}{L^{\star\, z}}^{\sharp\mathsf{p}}_{s}(L)$. By
Lemma~\ref{LSubFam}, there exist a $P\in L$ and a family
$(Q^{z}_{\alpha})_{\alpha\in\bb{P}_{z}}\in(L^{\star\, z})^{\bb{P}_{z}}$ such that
$R=\sbs{z}{(Q^{z}_{\alpha})_{\alpha\in\bb{P}_{z}}}(P)$. Since the family is
finite and, as noted in Remark~\ref{RIter},
$(L^{i\,z})_{i\in\mathbb{N}}$ is an ascending chain whose union is
$L^{\star\, z}$, there exists an $i\in\mathbb{N}$ such that, for every
$\alpha\in\bb{P}_{z}$, $Q^{z}_{\alpha}\in L^{i\,z}$. Hence, by
Lemma~\ref{LSubFam} again and Definition~\ref{DIter},
\[
R\in\sbs{z}{L^{i\,z}}^{\sharp}_{s}(P)
\subseteq\sbs{z}{L^{i\,z}}^{\sharp\mathsf{p}}_{s}(L)
\subseteq L^{i+1\,z}\subseteq L^{\star\, z}.
\qedhere
\]
\end{proof}

\subsection{
\texorpdfstring
{Recognizable subsets of $\mathbf{T}_{\Sigma}(X)$}
{Recognizable subsets of the free many-sorted algebra}
}
In this subsection, we state some  recognizability results relative to a free $\Sigma$-algebra on an $S$-sorted set $X$. Their proofs appear in~\cite{CVCL20}.

\begin{asm}
In this subsection, we assume that $S$ is a finite set of sorts, that $\Sigma$ is a finite $S$-sorted signature, and that $X$ is a finite $S$-sorted set.
\end{asm}

The singletons consisting, respectively, of a variable, a constant, and the action of an operator symbol on a family of variables are $s$-recognizable for a specific sort $s\in S$.

\begin{propC}
[{\cite[Prop.~3.1, 3.2, 3.3]{CVCL20}}]
\label{PRecVar}
Let $\Sigma$ be an $S$-sorted signature, $X$ an $S$-sorted set, and $s\in S$, then the following properties hold:
\begin{enumerate}
 \item For every $x\in X_{s}$, the language $\{x\}\subseteq \mathrm{T}_{\Sigma}(X)_{s}$ is $s$-recognizable.
 \item For every $\sigma\in\Sigma_{\lambda,s}$, the language $\{\sigma^{\mathbf{T}_{\Sigma}(X)}\}\subseteq \mathrm{T}_{\Sigma}(X)_{s}$ is $s$-recognizable.
 \item For every $\mathbf{s}\in S^{\star}-\{\lambda\}$, $\sigma\in\Sigma_{\mathbf{s},s}$, and $(x_{j})_{j\in \lvert{\mathbf{s}}\rvert}\in X_{\mathbf{s}}$, the language $\{\sigma^{\mathbf{T}_{\Sigma}(X)}((x_{j})_{j\in \lvert{\mathbf{s}}\rvert})\}\subseteq\mathrm{T}_{\Sigma}(X)_{s}$ is $s$-recognizable.
\end{enumerate}
\end{propC}

Next proposition states that recognizable languages are closed under substitution.

\begin{propC}[{\cite[Prop.~3.19]{CVCL20}}]
\label{PRecSubs}
Let $s$ be a sort in $S$, $K\in\mathrm{Rec}_{s}(\mathbf{T}_{\Sigma}(X))$, and $\left(\sbs{x}{L_{x}}_{x\in X_{t}}\right)_{t\in S}$ an $S$-sorted mapping from $X$ to $(\mathrm{Rec}_{t}(\mathbf{T}_{\Sigma}(X)))_{t\in S}$. Then
$$
\left(\left(\sbs{x}{L_{x}}_{x\in X_{t}}\right)_{t\in S}\right)^{\sharp\mathsf{p}}_{s}(K)\in\mathrm{Rec}_{s}(\mathbf{T}_{\Sigma}(X)).
$$
\end{propC}

\begin{cor}
\label{CRecSubs}
Let $s,t\in S$, $z\in X_{t}$, $L\in \mathrm{Rec}_{t}(\mathbf{T}_{\Sigma}(X))$, and $K\in\mathrm{Rec}_{s}(\mathbf{T}_{\Sigma}(X))$. Then 
$$\sbs{z}{L}^{\sharp\mathsf{p}}_{s}(K)\in \mathrm{Rec}_{s}(\mathbf{T}_{\Sigma}(X)).$$
\end{cor}

\begin{cor}
\label{CRecOp}
Let $(\mathbf{s},s)$ in $S^{\star}\times S$, $\sigma\in\Sigma_{\mathbf{s},s}$, and $(L_{j})_{j\in \lvert{\mathbf{s}}\rvert}\in \prod_{j\in \lvert{\mathbf{s}}\rvert}\mathrm{Rec}_{s_{j}}(\mathbf{T}_{\Sigma}(X))$. Then 
$$\sigma^{\mathbf{T}_{\Sigma}(X)^{\wp}}((L_{j})_{j\in \lvert{\mathbf{s}}\rvert})\in \mathrm{Rec}_{s}(\mathbf{T}_{\Sigma}(X)).$$
\end{cor}

Next proposition states that, for every sort $s\in S$ and $z\in X_{s}$, if the input language $L\subseteq \mathrm{T}_{\Sigma}(X)_{s}$ is recognizable, then its $z$-iteration is also recognizable.

\begin{propC}
[{\cite[Prop.~3.25]{CVCL20}}]
\label{PRecIt}
Let $s\in S$ and $z\in X_{s}$. If $L\in\mathrm{Rec}_{s}(\mathbf{T}_{\Sigma}(X))$, then $$L^{\star\, z}\in\mathrm{Rec}_{s}(\mathbf{T}_{\Sigma}(X)).$$
\end{propC}

\section{Regular many-sorted languages}

In this section, we prove the main result of the paper: for a finite set of sorts $S$, a finite $S$-sorted signature $\Sigma$, a finite $S$-sorted set $X$, and a sort $s\in S$, the $s$-recognizable and the $s$-regular languages of $\mathbf{T}_{\Sigma}(X)$ coincide. The inclusion of the $s$-regular languages into the $s$-recognizable ones follows from the closure properties established in Section~3. The reverse inclusion, which is the substantial part of the argument, generalizes Lemma~2.5.7 of G\'{e}cseg and Steinby~\cite{GS84}, itself adapted from the proof of McNaughton and Yamada~\cite{MY60} for the free monoid.

\begin{asm}
In this section, we assume that $S$ is a finite set of sorts, that $\Sigma$ is a finite $S$-sorted signature, and that $X$ is a finite $S$-sorted set.
\end{asm}

We first introduce the regular signature determined by the set of sorts $S$, the $S$-sorted signature $\Sigma$, and an $S$-sorted set.

\begin{defi}\label{DRegSig}
Let $Z$ be an $S$-sorted set. The \emph{regular} signature determined by $(S,\Sigma,Z)$ is the $S$-sorted signature $\mathrm{Reg}(S,\Sigma,Z)$ where, for every $(\mathbf{s},s)\in S^{\star}\times S$, the set $\mathrm{Reg}(S,\Sigma,Z)_{\mathbf{s},s}$ is given by
$$
\mathrm{Reg}(S,\Sigma,Z)_{\mathbf{s},s}=
\begin{cases}
\Sigma_{\mathbf{s},s}\amalg \{\varnothing_{s}\} &\text{if } \mathbf{s}=\lambda;\\
\Sigma_{\mathbf{s},s}\amalg \{(\cdot)^{\star\, z}\mid z \in Z_{s}\} &\text{if } \mathbf{s}=(s);\\
\Sigma_{\mathbf{s},s}\amalg (\{+_{s}\}\cup\{\sbs{z}{{\cdot}}^{\sharp\mathsf{p}}_{s}(\cdot)\mid z \in Z_{s}\}) &\text{if } \mathbf{s}=(s,s);\\
\Sigma_{\mathbf{s},s}\amalg \{\sbs{z}{{\cdot}}^{\sharp\mathsf{p}}_{s}(\cdot)\mid z \in Z_{t}\} &\text{if } \mathbf{s}=(t,s), \text{with } t\neq s;\\
\Sigma_{\mathbf{s},s} &\text{otherwise}.
\end{cases}
$$

That is, $\mathrm{Reg}(S,\Sigma,Z)_{\mathbf{s},s}$ is the expansion of $\Sigma$ by adding, for every sort $s\in S$, (1) an \emph{empty} constant of coarity $s$, namely $\varnothing_{s}$, (2) a binary \emph{sum} operation of coarity $s$, namely $+_{s}$, (3) for every $z\in Z_{s}$, the unary operation of coarity $s$ of \emph{$z$-iteration}, namely $(\cdot)^{\star\, z}$, and (4) for every $t\in S$, and every $z\in Z_{t}$, the operation of arity $(t,s)$ and coarity $s$ of \emph{$z$-substitution}, namely $\sbs{z}{{\cdot}}^{\sharp\mathsf{p}}_{s}(\cdot)$.
\end{defi}

\begin{defi}
For an $S$-sorted set $Z$, and a sort $s\in S$, a term in $\mathrm{T}_{\mathrm{Reg}(S,\Sigma,Z)}(Z)_{s}$, the $s$-th component of the underlying set of $\mathbf{T}_{\mathrm{Reg}(S,\Sigma,Z)}(Z)$, the free $\mathrm{Reg}(S,\Sigma,Z)$-algebra over $Z$, will be called a \emph{regular expression over $(S,\Sigma,Z)$} of type $s$.
\end{defi}

We first prove that, for every $S$-sorted set $Z$, the underlying set of the power algebra $\mathbf{T}_{\Sigma}(Z)^{\wp}$ has a natural structure of $\mathrm{Reg}(S,\Sigma,Z)$-algebra.

\begin{prop}\label{PRegTerm} Let $Z$ be an $S$-sorted set. The $S$-sorted set $\mathrm{T}_{\Sigma}(Z)^{\wp}$ is equipped with a structure of $\mathrm{Reg}(S,\Sigma,Z)$-algebra.
\end{prop}
\begin{proof}
 Let us denote by $\mathbf{T}^{\mathrm{Reg}}_{\Sigma}(Z)^{\wp}$ the $\mathrm{Reg}(S,\Sigma,Z)$-algebra defined as follows.
 
\textsf{(1)} The underlying $S$-sorted set of $\mathbf{T}^{\mathrm{Reg}}_{\Sigma}(Z)^{\wp}$ is $\mathrm{T}_{\Sigma}(Z)^{\wp}=(\mathrm{Sub}(\mathrm{T}_{\Sigma}(Z)_{s}))_{s\in S}$.

\textsf{(2)} For every $(\mathbf{s},s)\in S^{\star}\times S$ and every operation symbol $\sigma\in \Sigma_{\mathbf{s},s}$, the operation $\sigma^{\mathbf{T}^{\mathrm{Reg}}_{\Sigma}(Z)^{\wp}}$ is equal to $\sigma^{\mathbf{T}_{\Sigma}(Z)^{\wp}}$, as introduced in Proposition~\ref{PSubsetAlg}.

\textsf{(3)} For every $s\in S$, the constant $\varnothing^{\mathbf{T}^{\mathrm{Reg}}_{\Sigma}(Z)^{\wp}}_{s}$ is equal to $\varnothing \in \mathrm{T}_{\Sigma}(Z)^{\wp}_{s}$.

\textsf{(4)} For every $s\in S$, the operation ${+}^{\mathbf{T}^{\mathrm{Reg}}_{\Sigma}(Z)^{\wp}}_{s}$ is given by the union of languages, that is,
$$
{+}^{\mathbf{T}^{\mathrm{Reg}}_{\Sigma}(Z)^{\wp}}_{s}
\left\lbrace
\begin{array}{ccl}
\mathrm{T}_{\Sigma}(Z)^{\wp}_{s}\times \mathrm{T}_{\Sigma}(Z)^{\wp}_{s} & \longrightarrow &\mathrm{T}_{\Sigma}(Z)^{\wp}_{s} \\
 (L,K) & \longmapsto & L\cup K
\end{array}
 \right.
$$

\textsf{(5)} For every $s\in S$ and every $z\in Z_{s}$, the operation $((\cdot)^{\star\, z})^{\mathbf{T}^{\mathrm{Reg}}_{\Sigma}(Z)^{\wp}}$ is given by $z$-iteration, that is
$$
((\cdot)^{\star\, z})^{\mathbf{T}^{\mathrm{Reg}}_{\Sigma}(Z)^{\wp}}
\left\lbrace
\begin{array}{ccl}
\mathrm{T}_{\Sigma}(Z)^{\wp}_{s} & \longrightarrow &\mathrm{T}_{\Sigma}(Z)^{\wp}_{s} \\
L & \longmapsto & 
 L^{\star\, z} 
\end{array}
 \right.
$$

\textsf{(6)} For every $s,t\in S$ and every $z\in Z_{t}$, the operation $(\sbs{z}{{\cdot}}^{\sharp\mathsf{p}}_{s}(\cdot))^{\mathbf{T}^{\mathrm{Reg}}_{\Sigma}(Z)^{\wp}}$ is given by $z$-substitution, that is,
$$
(\sbs{z}{{\cdot}}^{\sharp\mathsf{p}}_{s}(\cdot))^{\mathbf{T}^{\mathrm{Reg}}_{\Sigma}(Z)^{\wp}}
\left\lbrace
\begin{array}{ccl}
\mathrm{T}_{\Sigma}(Z)^{\wp}_{t} \times \mathrm{T}_{\Sigma}(Z)^{\wp}_{s} & \longrightarrow &\mathrm{T}_{\Sigma}(Z)^{\wp}_{s} \\
(L,K) & \longmapsto & \sbs{z}{L}^{\sharp\mathsf{p}}_{s}(K) 
\end{array}
 \right.
$$
This completes the definition of $\mathbf{T}^{\mathrm{Reg}}_{\Sigma}(Z)^{\wp}$.
\end{proof}

Following the results introduced in Section~3, when the set of sorts $S$ is finite and the $S$-sorted set $Z$ is finite, the $S$-sorted set $(\mathrm{Rec}_{s}(\mathbf{T}_{\Sigma}(Z)))_{s\in S}$ is a closed subset of the $\mathrm{Reg}(S,\Sigma,Z)$-algebra introduced above.

\begin{cor} Let $Z$ be a finite $S$-sorted set.
 The $S$-sorted set $(\mathrm{Rec}_{s}(\mathbf{T}_{\Sigma}(Z)))_{s\in S}$ is a closed subset of the $\mathrm{Reg}(S,\Sigma,Z)$-algebra $\mathbf{T}^{\mathrm{Reg}}_{\Sigma}(Z)^{\wp}$. We denote by $\mathbf{Rec}^{\mathrm{Reg}}_{\boldsymbol{\cdot}}(\mathbf{T}_{\Sigma}(Z))$ the $\mathrm{Reg}(S,\Sigma,Z)$-algebra canonically associated to the closed subset $(\mathrm{Rec}_{s}(\mathbf{T}_{\Sigma}(Z)))_{s\in S}$ of $\mathbf{T}^{\mathrm{Reg}}_{\Sigma}(Z)^{\wp}$.
\end{cor}
\begin{proof}
We need to prove that the $S$-sorted set $(\mathrm{Rec}_{s}(\mathbf{T}_{\Sigma}(Z)))_{s\in S}$ is closed under all the operations coming from $\mathrm{Reg}(S,\Sigma,Z)$.

\textsf{(1)} Let $(\mathbf{s},s)\in S^{\star}\times S$ and let $\sigma\in \Sigma_{\mathbf{s},s}$. That $(\mathrm{Rec}_{s}(\mathbf{T}_{\Sigma}(Z)))_{s\in S}$ is closed under operation $\sigma^{\mathbf{T}^{\mathrm{Reg}}_{\Sigma}(Z)^{\wp}}$, i.e., under operation $\sigma^{\mathbf{T}_{\Sigma}(Z)^{\wp}}$, was proven in Corollary~\ref{CRecOp}.

\textsf{(2)} Let $s\in S$. That $\varnothing^{\mathbf{T}^{\mathrm{Reg}}_{\Sigma}(Z)^{\wp}}_{s}=\varnothing$ is an element in $\mathrm{Rec}_{s}(\mathbf{T}_{\Sigma}(Z))$ was proven in Proposition~\ref{PsRecCl}.

\textsf{(3)} Let $s\in S$. That $(\mathrm{Rec}_{s}(\mathbf{T}_{\Sigma}(Z)))_{s\in S}$ is closed under operation ${+}^{\mathbf{T}^{\mathrm{Reg}}_{\Sigma}(Z)^{\wp}}_{s}$, that is, closed under union, was proven in Proposition~\ref{PsRecCl}.

\textsf{(4)} Let $s\in S$ and $z\in Z_{s}$. The fact that $(\mathrm{Rec}_{s}(\mathbf{T}_{\Sigma}(Z)))_{s\in S}$ is closed under the operation $((\cdot)^{\star\, z})^{\mathbf{T}^{\mathrm{Reg}}_{\Sigma}(Z)^{\wp}}$ follows from Proposition~\ref{PRecIt}. 

\textsf{(5)} Let $s,t\in S$ and $z\in Z_{t}$. That $(\mathrm{Rec}_{s}(\mathbf{T}_{\Sigma}(Z)))_{s\in S}$ is closed under operation $(\sbs{z}{{\cdot}}^{\sharp\mathsf{p}}_{s}(\cdot))^{\mathbf{T}^{\mathrm{Reg}}_{\Sigma}(Z)^{\wp}}$ follows from Corollary~\ref{CRecSubs}. 
\end{proof}

\begin{rem}\label{RInt}
Let $Z$ be a finite $S$-sorted set. 
Consider the following slight modification of the $S$-sorted mapping $\{\cdot\}^{Z}=(\{\cdot\}^{Z}_{s})_{s\in S}$, introduced in Proposition~\ref{PSubFunct}. By abuse of notation, we will keep the same name.
$$
\{\cdot\}^{Z}_{s} 
\left\lbrace
\begin{array}{ccc}
Z_{s} & \longrightarrow &
\mathrm{Rec}_{s}(\mathbf{T}_{\Sigma}(Z))\\
z& \longmapsto &
\{z\}
\end{array}
\right.
$$
This is well-defined according to Proposition~\ref{PRecVar}. Now,  $\{\cdot\}^{Z}$ is an $S$-sorted mapping from $Z$ to $(\mathrm{Rec}_{s}(\mathbf{T}_{\Sigma}(Z)))_{s\in S}$, i.e., to the underlying $S$-sorted set of the $\mathrm{Reg}(S,\Sigma,Z)$-algebra $\mathbf{Rec}^{\mathrm{Reg}}_{\boldsymbol{\cdot}}(\mathbf{T}_{\Sigma}(Z))$. Thus, by the universal property of the free $\mathrm{Reg}(S,\Sigma,Z)$-algebra $\mathbf{T}_{\mathrm{Reg}(S,\Sigma,Z)}(Z)$, there exists a unique $\mathrm{Reg}(S,\Sigma,Z)$-ho\-mo\-mor\-phism $\{\cdot\}^{Z\sharp}\colon \mathbf{T}_{\mathrm{Reg}(S,\Sigma,Z)}(Z)\longrightarrow \mathbf{Rec}^{\mathrm{Reg}}_{\boldsymbol{\cdot}}(\mathbf{T}_{\Sigma}(Z))$ such that $\{\cdot\}^{Z\sharp}\circ \eta^{Z}=\{\cdot \}^{Z}$. 

For a sort $s\in S$, and a regular expression $R\in \mathrm{T}_{\mathrm{Reg}(S,\Sigma,Z)}(Z)_{s}$, we will write $\{R\}^{Z\sharp}_{s}$ instead of $\{\cdot\}^{Z\sharp}_{s}(R)$. Note that $\{R\}^{Z\sharp}_{s}$ is, by construction, an $s$-recognizable language of $\mathbf{T}_{\Sigma}(Z)_{s}$.
\end{rem}

We are now in position to define the notion of $s$-regular languages.
\begin{defi}\label{DReg} Let $s\in S$ and $L\subseteq \mathrm{T}_{\Sigma}(X)_{s}$. We say that $L$ is \emph{$s$-regular} if there exists a finite $S$-sorted set $Z$ with $X\subseteq Z$ and a regular expression $R\in \mathrm{T}_{\mathrm{Reg}(S,\Sigma,Z)}(Z)_{s}$ such that
$L= \{R\}^{Z\sharp}_{s}$. We will denote by $\mathrm{Reg}_{s}(\mathbf{T}_{\Sigma}(X))$ the set of all languages of $\mathbf{T}_{\Sigma}(X)_{s}$ that are $s$-regular.
\end{defi}

\begin{rem} In Definition~\ref{DReg}, the variables in $Z-X$ serve an auxiliary role. Specifically, variables not included in $X$ are used to carry out the defined iteration and substitution operations. However, because the language in question only contains variables from $X$, any variables from $Z-X$ will ultimately vanish when the regular expression is interpreted as a language. 
\end{rem}

As an immediate corollary of the previous definition, every $s$-regular language is $s$-recognizable.

\begin{cor}\label{CRegtoRec} For every $s\in S$, $\mathrm{Reg}_{s}(\mathbf{T}_{\Sigma}(X))\subseteq 
\mathrm{Rec}_{s}(\mathbf{T}_{\Sigma}(X))$.
\end{cor}
\begin{proof}
Let $L\subseteq \mathrm{T}_{\Sigma}(X)_{s}$ be a regular language in $\mathrm{Reg}_{s}(\mathbf{T}_{\Sigma}(X))$, then there exists a finite $S$-sorted set $Z$ with $X\subseteq Z$ and a regular expression $R\in \mathrm{T}_{\mathrm{Reg}(S,\Sigma,Z)}(Z)_{s}$ such that
$L= \{R\}^{Z\sharp}_{s}$. By Remark~\ref{RInt}, $\{R\}^{Z\sharp}_{s}\in \mathrm{Rec}_{s}(\mathbf{T}_{\Sigma}(Z))$. Now, since $X\subseteq Z$, we have that $\mathbf{T}_{\Sigma}(X)$ is a subalgebra of $\mathbf{T}_{\Sigma}(Z)$. By Proposition~\ref{PsRecSub}, we conclude that $L=\{R\}^{Z\sharp}_{s}\in \mathrm{Rec}_{s}(\mathbf{T}_{\Sigma}(X))$.
\end{proof}

Next, we introduce a small lemma that will be used later. It states that for each sort $s \in S$ and every term of sort $s$ in $\mathrm{T}_{\Sigma}(X)_{s}$, the singleton containing this term is $s$-regular. Indeed, the term itself, viewed as a regular expression, suffices for this proof.

\begin{lem}\label{LTermReg} Let $s\in S$ and $P\in \mathrm{T}_{\Sigma}(X)_{s}$, then $P\in\mathrm{T}_{\mathrm{Reg}(S,\Sigma,X)}(X)_{s}$ and $\{P\}^{X \sharp}_{s}=\{P\}$. Consequently, $\{P\}\in\mathrm{Reg}_{s}(\mathbf{T}_{\Sigma}(X))$.
\end{lem}
\begin{proof}
The proof is done by induction on $(\coprod \mathrm{T}_{\Sigma}(X), \leq_{\mathbf{T}_{\Sigma}(X)})$.

\textsf{Base step of the Artinian induction.}

Let $(P,s)$ be a minimal element in $(\coprod\mathrm{T}_{\Sigma}(X),\leq_{\mathbf{T}_{\Sigma}(X)})$. We have, by Proposition~\ref{PArtOrd}, that $P$ has the form (1) $x$, for a unique variable $x\in X_{s}$ or (2) $\sigma^{\mathbf{T}_{\Sigma}(X)}$, for a unique constant operation symbol $\sigma\in\Sigma_{\lambda,s}$.

When $P=x$, for a variable $x\in X_{s}$, we have that 
$$\{P\}^{X \sharp}_{s}=\{x\}^{X \sharp}_{s}=(\{\cdot\}^{X\sharp}\circ \eta^{X})_{s}(x)=\{\cdot\}^{X}_{s}(x)=\{x\}=\{P\}.$$

When $P=\sigma^{\mathbf{T}_{\Sigma}(X)}$, for a constant operation symbol $\sigma\in\Sigma_{\lambda,s}$, we have that 
$$\{P\}^{X \sharp}_{s}=\{\sigma\}^{X \sharp}_{s}=
\sigma^{\mathbf{Rec}^{\mathrm{Reg}}_{\boldsymbol{\cdot}}(\mathbf{T}_{\Sigma}(X))} = \sigma^{\mathbf{T}_{\Sigma}(X)^{\wp}}=\{\sigma^{\mathbf{T}_{\Sigma}(X)}\}=\{P\}.$$

\textsf{Inductive step of the Artinian induction.}

Let $(P,s)$ be a non-minimal element in $(\coprod\mathrm{T}_{\Sigma}(X),\leq_{\mathbf{T}_{\Sigma}(X)})$. Then, by Proposition~\ref{PArtOrd}, there exists a unique $(\mathbf{s},s)\in (S^{\star}-\{\lambda\})\times S$, a unique operation symbol $\sigma\in\Sigma_{\mathbf{s},s}$ and a unique family of terms $(P_{j})_{j\in\bb{\mathbf{s}}}\in \mathrm{T}_{\Sigma}(X)_{\mathbf{s}}$ for which $P=\sigma^{\mathbf{T}_{\Sigma}(X)}((P_{j})_{j\in \bb{\mathbf{s}}})$. Assume that the statement holds for every immediate subterm of $P$, that is, for every $j\in\bb{\mathbf{s}}$, we have that $\{P_{j}\}^{X \sharp}_{s_{j}}=\{P_{j}\}$. 

The following chain of equalities holds
\allowdisplaybreaks
\begin{align*}
\{P\}^{X \sharp}_{s}&=
\{
\sigma^{\mathbf{T}_{\Sigma}(X)}((P_{j})_{j\in \bb{\mathbf{s}}})
\}^{X \sharp}_{s}
\tag{Def. of $P$}
\\
&=
\sigma^{\mathbf{Rec}^{\mathrm{Reg}}_{\boldsymbol{\cdot}}(\mathbf{T}_{\Sigma}(X))}
((
\{P_{j}\}^{X\sharp}_{s_{j}}
)_{j\in\bb{\mathbf{s}}})
\tag{Rem.~\ref{RInt}}
\\
&=
\sigma^{\mathbf{T}_{\Sigma}(X)^{\wp}}
((
\{P_{j}\}^{X\sharp}_{s_{j}}
)_{j\in\bb{\mathbf{s}}})
\tag{Prop.~\ref{PRegTerm}}
\\
&=
\sigma^{\mathbf{T}_{\Sigma}(X)^{\wp}}
((
\{P_{j}\}
)_{j\in\bb{\mathbf{s}}})
\tag{Induction}
\\&=
\{
\sigma^{\mathbf{T}_{\Sigma}(X)}((P_{j})_{j\in \bb{\mathbf{s}}})
\}
\tag{Prop.~\ref{PSubsetAlg}}
\\&=\{P\}.
\tag{Def. of $P$}
\end{align*}
This completes the proof.
\end{proof}

The next result is the converse of Corollary~\ref{CRegtoRec}: every $s$-recognizable language is $s$-regular. Its proof, which occupies the rest of this section, adapts to the many-sorted setting the argument of Lemma~2.5.7 in~\cite{GS84}.

\begin{prop}\label{PRectoReg} For every $s\in S$, $\mathrm{Rec}_{s}(\mathbf{T}_{\Sigma}(X))\subseteq 
\mathrm{Reg}_{s}(\mathbf{T}_{\Sigma}(X))$.
\end{prop}
\begin{proof}
Let $s\in S$ and $L\in \mathrm{Rec}_{s}(\mathbf{T}_{\Sigma}(X))$. Then there exists a finite $\Sigma$-algebra $\mathbf{N}$, a $\Sigma$-homomorphism $f\colon \mathbf{T}_{\Sigma}(X)\longrightarrow\mathbf{N}$, and a subset $M\subseteq N_{s}$ such that $f_{s}^{-1}[M]=L$.

If $M$ is empty, so is $L$. In this case, $L=\{\varnothing^{\mathbf{T}_{\mathrm{Reg}(S,\Sigma,X)}(X)}_{s}\}^{X\sharp}_{s}$. Therefore, we can assume that $M$ is not empty.

Without loss of generality, we will assume that $X\cap N=\varnothing^{S}$ and that $N=(n_{s})_{s\in S}\subseteq \mathbb{N}^{S}$. In other words, for each $s\in S$, the $s$-th component of $N$ is a natural number in $\mathbb{N}$. Since it will be used later, we will say that an $S$-sorted set of natural numbers $K=(k_{s})_{s\in S}$ in $\mathbb{N}^{S}$ is smaller than $N$, written $K\leq N$, if it satisfies that, for every $s\in S$, $k_{s}\leq n_{s}$. Following this, if $K\leq N$, then $K\subseteq N$. Furthermore, we will denote by $\bb{\bb{K}}$ the sum of all the components in $K$, that is, $\bb{\bb{K}}=\sum_{s\in S}k_{s}$. This is well-defined since $S$ is finite. Furthermore, if $t\in S$ and $k_{t}\neq 0$, we denote by $K^{t-1}$ the $S$-sorted set $(K^{t-1}_{s})_{s\in S}$
that has, for $s\in S$, as $s$-th component the set
$$
K^{t-1}_{s}=
\begin{cases}
 K_{s}&\text{if } s\neq t;\\
 K_{t}-1 & \text{if } s=t.
\end{cases}
$$
\begin{rem}\label{ROrd}
 Note that $K^{t-1}\leq K\leq N$ and $\bb{\bb{K^{t-1}}}=\bb{\bb{K}}-1$. Finally, for two $S$-sorted subsets $J,K\leq N$, we denote by $\mathrm{max}(J,K)$ the $S$-sorted set $(\mathrm{max}(J_{s},K_{s}))_{s\in S}$. Since $J,K\leq N$, it follows that $\mathrm{max}(J,K)\leq N$.
Note that, for every $s\in S$, $\mathrm{max}(J,K)_{s}=J_{s}\cup K_{s}$. Hence, for every $m\in\mathbb{N}$, $m\in\mathrm{max}(J,K)_{s}$ if and only if $m\in J_{s}$ or $m\in K_{s}$. Moreover, for $s,t\in S$ and $i\in\mathbb{N}$, an element $m\in\mathbb{N}$ belongs to $\delta^{t,i+1}_{s}$ if and only if $t=s$ and $m\leq i$. Finally, if $K_{t}=k_{t}$ and $m\in k_{t}$ with $m\neq k_{t}-1$, then $m\in k_{t}-1=K^{t-1}_{t}$. 
\end{rem}

For the $S$-sorted set of variables $Z = X \cup N$, we will prove the existence of a regular expression $R \in \mathrm{T}_{\mathrm{Reg}(S,\Sigma,Z)}(Z)_{s}$ such that $L = \{R\}^{Z\sharp}_{s}$, proving that $L$ is $s$-regular.

We define the $S$-sorted mapping $h\colon X\cup N\longrightarrow N$, such that, for every $t\in S$, $h_{t}$ is defined as 
$$
h_{t}
\left\lbrace
\begin{array}{ccl}
X_{t}\cup N_{t} & \longrightarrow &
N_{t}\\
z& \longmapsto &
\begin{cases}
 f_{t}(\eta^{X}_{t}(z))&\text{if } z\in X_{t};\\
 z & \text{if } z\in N_{t}.
\end{cases}
\end{array}
\right.
$$

Let $h^{\sharp}$ be the unique $\Sigma$-homomorphism from $\mathbf{T}_{\Sigma}(X\cup N)$ to $\mathbf{N}$ satisfying that $h^{\sharp}\circ \eta^{X\cup N}=h$. Since $h\!\!\upharpoonright_{X}= f\circ \eta^{X}$, we have that $h^{\sharp}\!\!\upharpoonright_{\mathrm{T}_{\Sigma}(X)}=f$. Note also that $h\!\!\upharpoonright_{N}= \mathrm{id}^{N}$.

Let $u\in S$, $C,K\subseteq N$, with $K\leq N$, and $l\in n_{u}$. We define the set $L_{u}(C,K,l)$ to be the set of terms $P$ in $\mathrm{T}_{\Sigma}(Z)_{u}$ satisfying that
\begin{enumerate}
\item[(a)] $P\in \mathrm{T}_{\Sigma}(X\cup C)_{u}$; 
 \item[(b)] for every $t\in S$, and every $Q\in\mathrm{T}_{\Sigma}(X\cup C)_{t}$, if
 $(Q,t)<(P,u)$ and 
 $(Q,t)\notin\mathrm{Min}$,
 then
 $h^{\sharp}_{t}(Q)\in K_{t}$; and 
 \item[(c)] $h^{\sharp}_{u}(P)=l$.
\end{enumerate}

Note that, by
Remark~\ref{RRest}, condition~(b) in the definition above may be verified equivalently in
$\mathrm{T}_{\Sigma}(X\cup C)$ or in
$\mathrm{T}_{\Sigma}(Z)$.

\begin{rem}
Before proceeding any further, we will describe the idea behind the auxiliary languages $L_{u}(C,K,l)$.
We must produce, from the recognizing homomorphism, a regular expression for the language; we do so by eliminating states one at a time, as in the classical state-elimination proof, but indexed by sort. The auxiliary languages $L_{u}(C,K,l)$ are the units on which this elimination operates. The sort $u$ and the target value $l\in n_{u}$ fix the output: we collect terms of sort $u$ whose image under the recognizing homomorphism is the state $l$. The set $C\subseteq N$ is the set of states currently allowed to appear at the leaves, in addition to the variables and constants; enlarging $C$ corresponds to having already opened up certain states for substitution. The bound $K\leq N$ is the crucial one: it is a sortwise budget, recording, for each sort $t$, the set $K_{t}$ of state values that the internal (proper, non-minimal) subterms are still permitted to take. Thus $L_{u}(C,K,l)$ collects exactly the computations that reach $l$ at sort $u$ while passing only through states within the current budget $K$. 
\end{rem}

Since $h^{\sharp}\!\!\upharpoonright_{\mathrm{T}_{\Sigma}(X)}=f$ and $f_{s}^{-1}[M]=L$, the following equality holds
\begin{equation}
 \textstyle
L=\bigcup_{l\in M} L_{s}(\varnothing^{S}, N, l).
\tag{L}\label{EqL}
\end{equation} 

By Equation~\eqref{EqL} and since $M$ is finite, in order to prove Proposition~\ref{PRectoReg}, it suffices to prove the following Main Claim. 

\begin{clm}[Main Claim]\label{CMain}
 For every $u\in S$, every $C\subseteq N$, every $K\subseteq N$ with $K\leq N$, and every $l\in n_{u}$, there exists a regular expression $R_{u}(C,K,l)$ in $\mathrm{T}_{\mathrm{Reg}(S,\Sigma,Z)}(Z)_{u}$ such that 
 \[
 \{R_{u}(C,K,l)\}^{Z\sharp}_{u}=L_{u}(C,K,l).
 \]
 In other words, the sets $L_{u}(C,K,l)$ are $u$-regular. 
\end{clm}
\begin{proof}
We will prove Main Claim~\ref{CMain} by induction on $\bb{\bb{K}}=\sum_{s\in S}k_{s}$. As the proof of Main Claim~\ref{CMain} shows, lowering the budget at one sort, by removing its top state, is what makes the elimination terminate.

\textsf{Base step of the induction.}

If $\bb{\bb{K}}=0$, then $K=(0)_{s\in S}$, that is, the $S$-sorted set of natural numbers that satisfies that, for every $s\in S$, $K_{s}=0$. Let $u\in S$, $C\subseteq N$, and $l\in n_{u}$, the set $L_{u}(C,(0)_{s\in S},l)$ only contains terms $P\in\mathrm{T}_{\Sigma}(X\cup C)_{u}$ without proper non-minimal subterms. Taking
into account Propositions~\ref{PTermChar} and~\ref{PArtOrd} and
Remark~\ref{RRestII}, a term $P\in L_{u}(C,(0)_{s\in S},l)$ has to be 
\begin{enumerate}
 \item $P=x$, for a unique variable $x\in X_{u}$ satisfying that $f_{u}(\eta^{X}_{u}(x))=l$;
 \item $P=l$ (viewed as a variable of sort $u$), if $l\in C_{u}$; or 
 \item $P=\sigma^{\mathbf{T}_{\Sigma}(X\cup C)}$, for a unique constant operation symbol $\sigma\in \Sigma_{\lambda,u}$ satisfying that $h^{\sharp}_{u}(\sigma^{\mathbf{T}_{\Sigma}(X\cup C)})=f_{u}(\sigma^{\mathbf{T}_{\Sigma}(X)})=\sigma^{\mathbf{N}}=l$.
 \item $P=\sigma^{\mathbf{T}_{\Sigma}(X\cup C)}((z_{j})_{j\in \bb{\mathbf{s}}})$, for a unique $\mathbf{s}\in S^{\star}-\{\lambda\}$, a unique operation symbol $\sigma\in \Sigma_{\mathbf{s},u}$, and a unique $(z_{j})_{j\in\bb{\mathbf{s}}}\in \mathrm{T}_{\Sigma}(X\cup C)_{\mathbf{s}}$
 such that, for every $j\in \bb{\mathbf{s}}$, 
 $(z_{j},s_{j})$ in $\mathrm{Min}$, and, finally,
 satisfying that $$h^{\sharp}_{u}(P)=\sigma^{\mathbf{N}}((h^{\sharp}_{s_{j}}(z_{j}))_{j\in\bb{\mathbf{s}}})=l.$$
\end{enumerate}

In any case, since $X$, $C$ and $\Sigma$ are finite, there are only a finite number of options for $P$ to be an element in $L_{u}(C,(0)_{s\in S},l)$. By Lemma~\ref{LTermReg}, applied to the finite $S$-sorted set $Z$, every $P\in L_{u}(C,(0)_{s\in S},l)\subseteq \mathrm{T}_{\Sigma}(Z)_{u}$ is a regular expression in $\mathrm{T}_{\mathrm{Reg}(S,\Sigma,Z)}(Z)_{u}$ with $\{P\}^{Z\sharp}_{u}=\{P\}$. Hence, the regular expression
$$\textstyle R_{u}(C,(0)_{s\in S},l)
=\sum_{P\in L_{u}(C,(0)_{s\in S},l)} P$$ 
satisfies that $\{R_{u}(C,(0)_{s\in S},l)\}^{Z\sharp }_{u}=L_{u}(C,(0)_{s\in S},l)$.

This completes the base step.

\textsf{Inductive step of the induction.}

Let $k\in\mathbb{N}$.
Assume that Main Claim~\ref{CMain} holds for every bound of size $k$, that is, for every $u\in S$, every $C\subseteq N$, every $K\subseteq N$ with $K\leq N$ and $\bb{\bb{K}}=k$, and every $l\in n_{u}$, there exists a regular expression $R_{u}(C,K,l)$ whose interpretation
$\{R_{u}(C,K,l)\}^{Z\sharp}_{u}$
is $L_{u}(C,K,l)$.

Now, let $u\in S$, $C\subseteq N$, $K\subseteq N$ with $K=(k_{s})_{s\in S}$, $K\leq N$ with $\bb{\bb{K}}=k+1$, and $l\in n_{u}$. Finally, let $t\in S$ be such that $k_{t}\neq 0$. Note that, since $0<k_{t}\leq n_{t}$, we have that $k_{t}-1\in n_{t}=N_{t}$; thus $k_{t}-1$ is a variable of sort $t$ in $Z$ and $h^{\sharp}_{t}(k_{t}-1)=k_{t}-1$.

We define $J^{t-1}_{u}(C,K,l)$ as the language
$$
\begin{adjustbox}{max width=\linewidth, keepaspectratio}
$
 \sbs{k_{t}-1}{{
L_{t}(C,K^{t-1},k_{t}-1)
}}^{\sharp\mathsf{p}}_{u}
\left(\sbs{k_{t}-1}{{
L_{t}(C\cup \delta^{t,\{k_{t}-1\}}, K^{t-1}, k_{t}-1)^{\star (k_{t}-1)}
}}^{\sharp\mathsf{p}}_{u}(
L_{u}(C\cup \delta^{t,\{k_{t}-1\}}, K^{t-1}, l)
)\right).
$
\end{adjustbox}
$$

We will prove the following equality
\begin{equation}
\textstyle
L_{u}(C,K,l)= \bigcup_{t\in S, k_{t}\neq 0} \left(
L_{t}(C,K^{t-1},l)
\cup
J^{t-1}_{u}(C,K,l) \right).
\tag{E}\label{EqE}
\end{equation}

Note that all the languages appearing on the right-hand side of Equation~\eqref{EqE} have, by induction, an associated regular expression. As a consequence, the language $L_{u}(C,K,l)$ will also admit a regular expression, since the equality in Equation~\eqref{EqE} uses the operations of union, substitution and iteration. In addition, there are only a finite number of indices describing the union. Thus, proving Equation~\eqref{EqE} then yields the desired regular expression $R_{u}(C,K,l)$.

\begin{rem}
    The recursion behind the inductive step reads as follows. A term counted by $L_{u}(C,K,l)$ either never reaches, at any internal node, the top state $k_{t}-1$ of some sort $t$ whose budget is nonzero, in which case it already lives in a language with strictly smaller budget, $L_{t}(C,K^{t-1},l)$; or it does reach that top state, and then it factors, through the occurrences of $k_{t}-1$, into an outer part and iterated inner parts, each of which again lives in a strictly smaller budget. The language $J^{t-1}_{u}(C,K,l)$ defined below captures exactly this second, factored case by means of a substitution and an iteration. Equation~\eqref{EqE} expresses $L_{u}(C,K,l)$ as the union of these two alternatives, ranging over the sorts $t$ at which a state can be removed; since every language on its right-hand side has strictly smaller budget, the induction hypothesis applies to all of them.
\end{rem}

In order to prove Equation~\eqref{EqE}, we start by proving the inclusion from right to left. To do so, we present two auxiliary claims.

\begin{clm}\label{C1}
 For every $u,w\in S$, every $l\in n_{u}$, every $i\in n_{w}$, every $J,K\leq N$ and every $C,D\subseteq N$, the following inclusion holds
 $$
 \sbs{i}{{
L_{w}(C,K,i)
}}^{\sharp\mathsf{p}}_{u}
\left(
L_{u}(D\cup \delta^{w,\{i\}}, J,l)
\right) \subseteq 
 L_{u}(C\cup (D-\delta^{w,\{i\}}),\mathrm{max}(K,J,\delta^{w,i+1}),l).
 $$
\end{clm}
\begin{proof}
First note that, since $i\in n_{w}$, we have $i+1\leq n_{w}$, so that $\mathrm{max}(K,J,\delta^{w,i+1})\leq N$ and the right-hand side is well defined. 

Let $P\in \sbs{i}{{
L_{w}(C,K,i)
}}^{\sharp\mathsf{p}}_{u}
\left(
L_{u}(D\cup \delta^{w,\{i\}}, J,l)
\right)$. By Lemma~\ref{LSubFam}, there exists $W$ in $L_{u}(D\cup \delta^{w,\{i\}}, J,l)$ and a family of terms $(U^{i}_{\alpha})_{\alpha\in\bb{W}_{i}}$ in $L_{w}(C,K,i)^{\bb{W}_{i}}$ such that 
$$
P=\sbs{i}{{
(U^{i}_{\alpha})_{\alpha\in\bb{W}_{i}}
}}^{\sharp}_{u}(W).
$$

Recalling the definition of $L_{u}(C\cup (D-\delta^{w,\{i\}}),\mathrm{max}(K,J,\delta^{w,i+1}),l)$, we must verify the membership condition of $P$ in it, namely
\begin{enumerate}
\item[(a)] $P\in \mathrm{T}_{\Sigma}(X\cup C\cup (D-\delta^{w,\{i\}}))_{u}$; 
 \item[(b)] for every $t\in S$, and every $Q\in\mathrm{T}_{\Sigma}(X\cup C\cup (D-\delta^{w,\{i\}}))_{t}$, if
 $(Q,t)<(P,u)$ and 
 $(Q,t)\notin\mathrm{Min}$,
 then
 $h^{\sharp}_{t}(Q)\in \mathrm{max}(K,J,\delta^{w,i+1})_{t}$; and 
 \item[(c)] $h^{\sharp}_{u}(P)=l$.
\end{enumerate}

Before doing so, we record that, since, for every $\alpha\in\bb{W}_{i}$, the term $U^{i}_{\alpha}$ belongs to $L_{w}(C,K,i)$, we have that $h^{\sharp}_{w}(U^{i}_{\alpha})=i=h^{\sharp}_{w}(i)$, the last equality because $i\in n_{w}=N_{w}$.

\begin{enumerate}
 \item[(a)] Since $W\in \mathrm{T}_{\Sigma}(X\cup D\cup \delta^{w,\{i\}})_{u}$ and, for every $\alpha\in \bb{W}_{i}$, the term $U^{i}_{\alpha}$ belongs to $\mathrm{T}_{\Sigma}(X\cup C)_{w}$, Lemma~\ref{LSubVar} yields
$
P\in
\mathrm{T}_{\Sigma}(X\cup C\cup(D-\delta^{w,\{i\}}))_{u}
$.
Here we use that $X\cap N=\varnothing^{S}$, so that
$i\notin X_{w}$.

\item[(b)] Let $t\in S$ and let $Q$ be a term in $\mathrm{T}_{\Sigma}(X\cup C\cup (D-\delta^{w,\{i\}}))_{t}$ satisfying that
 $(Q,t)<(P,u)$ and 
 $(Q,t)\notin\mathrm{Min}$. By Lemma~\ref{LSubSubt}, 
one of the following two alternatives holds.

\textsf{Alternative (b.1).} There exists an $\alpha\in\bb{W}_{i}$ such that
$(Q,t)\leq(U^{i}_{\alpha},w)$. If $Q=U^{i}_{\alpha}$, so that $w=t$, then
$h^{\sharp}_{t}(Q)=i\in i+1=\delta^{w,i+1}_{t}$. If
$(Q,t)<(U^{i}_{\alpha},w)$, then $Q$ is a proper non-minimal subterm of
$U^{i}_{\alpha}$. Since $U^{i}_{\alpha}\in L_{w}(C,K,i)$ we conclude that
$h^{\sharp}_{t}(Q)\in K_{t}$.

\textsf{Alternative (b.2).} There exist a $(W',t)<(W,u)$ and a family
$(V^{i}_{\beta})_{\beta\in\bb{W'}_{i}}$ whose members belong to
$\{U^{i}_{\alpha}\mid\alpha\in\bb{W}_{i}\}$ such that
$Q=\sbs{i}{(V^{i}_{\beta})_{\beta\in\bb{W'}_{i}}}(W')$. Since every
$V^{i}_{\beta}$ satisfies
$h^{\sharp}_{w}(V^{i}_{\beta})=i=h^{\sharp}_{w}(i)$, Corollary~\ref{CSubInv}
yields $h^{\sharp}_{t}(Q)=h^{\sharp}_{t}(W')$. We distinguish two
subcases, whether (b.2.i) $(W',t)\in\mathrm{Min}$ or (b.2.ii) $(W',t)\not\in\mathrm{Min}$.

\textsf{Subcase (b.2.i).} If $(W',t)\in\mathrm{Min}$, then either $W'=i$, in which case
$t=w$ and $h^{\sharp}_{t}(Q)=h^{\sharp}_{w}(i)=i\in\delta^{w,i+1}_{t}$,
or $W'\neq i$, in which case $\bb{W'}_{i}=0$, so that $Q=W'$ and $(Q,t)$
is minimal, contradicting our assumption; hence this latter possibility
does not occur. 

\textsf{Subcase (b.2.ii).} If $(W',t)\notin\mathrm{Min}$, then $W'$ is a proper
non-minimal subterm of $W$. Since $W\in L_{u}(D\cup\delta^{w,\{i\}},J,l)$, it follows that
$h^{\sharp}_{t}(W')\in J_{t}$ and, consequently,
$h^{\sharp}_{t}(Q)\in J_{t}$.

In every case we have shown that
$h^{\sharp}_{t}(Q)\in K_{t}\cup J_{t}\cup\delta^{w,i+1}_{t}$, which, by
Remark~\ref{ROrd}, is precisely $\mathrm{max}(K,J,\delta^{w,i+1})_{t}$.

\item[(c)] Taking into account that $W\in L_{u}(D\cup\delta^{w,\{i\}},J,l)$ and Corollary~\ref{CSubInv}, we directly conclude that
$h^{\sharp}_{u}(P)=h^{\sharp}_{u}(W)=l$.
\end{enumerate}

This proves Claim~\ref{C1}. 
\end{proof}

We now continue with the second auxiliary claim.

\begin{clm}\label{C2}
 For every $u\in S$, every $l\in n_{u}$, every $K\leq N$ and every $C\subseteq N$, the following inclusion holds
 $$
L_{u}(C\cup \delta^{u,\{l\}},K,l)^{\star l} \subseteq 
 L_{u}(C\cup \delta^{u,\{l\}},\mathrm{max}(K,\delta^{u,l+1}),l).
 $$
\end{clm}
\begin{proof}
Note first that, since $l\in n_{u}$, we have that $\mathrm{max}(K,\delta^{u,l+1})\leq N$. Let us recall from Definition~\ref{DIter} that 
$$
\textstyle
L_{u}(C\cup \delta^{u,\{l\}},K,l)^{\star l}=\bigcup_{i\in\mathbb{N}}(L_{u}(C\cup \delta^{u,\{l\}},K,l))^{i,l}.$$

Therefore, in order to prove the inclusion, it suffices to prove that for every $i\in \mathbb{N}$, the following inclusion holds
$$(L_{u}(C\cup \delta^{u,\{l\}},K,l))^{i,l}\subseteq L_{u}(C\cup \delta^{u,\{l\}},\mathrm{max}(K,\delta^{u,l+1}),l).$$

We will prove the last statement by induction on $i$.

\textsf{Base step of the induction}

When $i=0$, we have that $(L_{u}((C\cup \delta^{u,\{l\}}),K,l))^{0,l}=\{l\}$. Note that
\begin{enumerate}
 \item[(a)] $l\in \mathrm{T}_{\Sigma}(X\cup C\cup \delta^{u,\{l\}})_{u}$;
 \item[(b)] $l$ has no proper subterms of any sort;
 \item[(c)] $h^{\sharp}_{u}(l)=l$.
\end{enumerate}

Thus, $l\in L_{u}(C\cup \delta^{u,\{l\}}, \mathrm{max}(K,\delta^{u,l+1}),l)$.

\textsf{Inductive step of the induction}

Assume that the statement holds for $i\in \mathbb{N}$, that is, assume that 
$$(L_{u}(C\cup \delta^{u,\{l\}},K,l))^{i,l}\subseteq L_{u}(C\cup \delta^{u,\{l\}}, \mathrm{max}(K,\delta^{u,l+1}),l).$$ 

Let us prove the statement for $i+1$. Following Definition~\ref{DIter}, we have that $(L_{u}(C\cup \delta^{u,\{l\}},K,l))^{i+1,l}$ is given by
$$
(L_{u}(C\cup \delta^{u,\{l\}},K,l))^{i,l}\cup \sbs{l}{(L_{u}(C\cup \delta^{u,\{l\}},K,l))^{i,l}}^{\sharp\mathsf{p}}_{u}(
L_{u}(C\cup \delta^{u,\{l\}},K,l)
).
$$

By the induction hypothesis, the first member of the above union is contained in $L_{u}(C\cup \delta^{u,\{l\}}, \mathrm{max}(K,\delta^{u,l+1}),l)$. For the second member of the above union, the following sequence of inclusions holds
\begin{flushleft}
 $\sbs{l}{(L_{u}(C\cup \delta^{u,\{l\}},K,l))^{i,l}}^{\sharp\mathsf{p}}_{u}(
L_{u}(C\cup \delta^{u,\{l\}},K,l)
)$
\allowdisplaybreaks
\begin{align*}
 & \subseteq\sbs{l}{L_{u}(C\cup \delta^{u,\{l\}},
\mathrm{max}(K,\delta^{u,l+1})
,l)}^{\sharp\mathsf{p}}_{u}(
L_{u}(C\cup \delta^{u,\{l\}},K,l)
)\tag{Ind. \& Rem.~\ref{RMon}}
\\
 &\subseteq
L_{u}(C\cup \delta^{u,\{l\}}, \mathrm{max}(K,\delta^{u,l+1}),l). 
\tag{Claim~\ref{C1}}
\end{align*}
\end{flushleft}

This completes the proof of Claim~\ref{C2}.
\end{proof}

We are now in position to prove the inclusion from right to left of Equation~\eqref{EqE}.

\begin{clm}[Right to left]\label{C3} The following inclusion holds
 $$\textstyle
 \bigcup_{t\in S, k_{t}\neq 0} \left(
L_{t}(C,K^{t-1},l)
\cup
J^{t-1}_{u}(C,K,l) \right) \subseteq L_{u}(C,K,l).$$ 
\end{clm}

\begin{proof}
Let $t$ be a sort in $S$ for which $k_{t}\neq 0$.
By definition, $L_{u}(C,K^{t-1},l)\subseteq L_{u}(C,K,l)$. Thus, it suffices to prove that $J_{u}^{t-1}(C,K,l)\subseteq L_{u}(C,K,l)$. First, the following chain of inclusions holds
\begin{flushleft}
 $
\sbs{k_{t}-1}{{
L_{t}(C\cup \delta^{t,\{k_{t}-1\}}, K^{t-1}, k_{t}-1)^{\star (k_{t}-1)}
}}^{\sharp\mathsf{p}}_{u}(
L_{u}(C\cup \delta^{t,\{k_{t}-1\}}, K^{t-1}, l)
) $
\allowdisplaybreaks
\begin{align*}
& \subseteq
\sbs{k_{t}-1}{{
L_{t}(C\cup \delta^{t,\{k_{t}-1\}}, K, k_{t}-1)
}}^{\sharp\mathsf{p}}_{u}
 (
L_{u}(C\cup \delta^{t,\{k_{t}-1\}}, K^{t-1}, l)
)
\tag{Cl.~\ref{C2} \& Rem.~\ref{RMon}}\\
&\subseteq 
L_{u}(C\cup \delta^{t,\{k_{t}-1\}}, K, l).
\tag{Claim~\ref{C1}}
\end{align*}
\end{flushleft}

Consequently, the following sequence of inclusions holds
\begin{flushleft}
$J^{t-1}_{u}(C,K,l)$
\allowdisplaybreaks
\begin{align*}
& = 
\begin{adjustbox}{max width=.95\linewidth, keepaspectratio}
 $\sbs{k_{t}-1}{{
L_{t}(C,K^{t-1},k_{t}-1)
}}^{\sharp\mathsf{p}}_{u}
\left(\sbs{k_{t}-1}{{
L_{t}(C\cup \delta^{t,\{k_{t}-1\}}, K^{t-1}, k_{t}-1)^{\star (k_{t}-1)}
}}^{\sharp\mathsf{p}}_{u}(
L_{u}(C\cup \delta^{t,\{k_{t}-1\}}, K^{t-1}, l)
)\right)$
\end{adjustbox}\\
 &
 \subseteq
 \sbs{k_{t}-1}{{
L_{t}(C,K^{t-1},k_{t}-1)
}}^{\sharp\mathsf{p}}_{u}(
L_{u}(C\cup \delta^{t,\{k_{t}-1\}}, K, l)
)\tag{See above \& Rem.~\ref{RMon}}
\\&\subseteq 
L_{u}(C,K,l).\tag{Claim~\ref{C1}}
\end{align*}
\end{flushleft}

This completes the proof of Claim~\ref{C3}. 
\end{proof}

Thus, the inclusion from right to left of Equation~\eqref{EqE} holds. Let us now move to prove the other inclusion, from left to right. Its proof will also require two auxiliary claims.

\begin{clm}[Left to right]\label{C4} The following inclusion holds
 $$\textstyle
 L_{u}(C,K,l) \subseteq 
 \bigcup_{t\in S, k_{t}\neq 0} \left(
L_{t}(C,K^{t-1},l)
\cup
J^{t-1}_{u}(C,K,l) \right).$$ 
\end{clm}

\begin{proof}
Let $P\in L_{u}(C,K,l)$. Since $\bb{\bb{K}}=k+1\geq 1$, there exists a sort $t\in S$ with $k_{t}\neq 0$. We fix one such $t$ and we will prove that 
\[ 
P\in L_{u}(C,K^{t-1},l)\cup J_{u}^{t-1}(C,K,l).
\tag{P}\label{EqP}
\]

If $(P,u)\in\mathrm{Min}$, the defining conditions for $P$ being a term in $L_{u}(C,K^{t-1},l)$ hold:
\begin{enumerate}
 \item[(a)] $P\in\mathrm{T}_{\Sigma}(X\cup C)_{u}$, since $P\in L_{u}(C,K,l)$;
 \item[(b)] holds vacuously, since $P$ has no proper subterms;
\item[(c)] $h^{\sharp}_{u}(P)=l$, since $P\in L_{u}(C,K,l)$.
\end{enumerate}

Thus, if $(P,u)\in\mathrm{Min}$, then $P\in L_{u}(C,K^{t-1},l)$. We may therefore assume that
$(P,u)\notin\mathrm{Min}$. In this case we will prove that $P\in
J_{u}^{t-1}(C,K,l)$. 

To simplify the notation, we write
\begin{align*}
 z&=k_{t}-1,
 &
 L&=L_{t}(C\cup \delta^{t,\{z\}}, K^{t-1},l),
 &
 L'&=L_{t}(C,K^{t-1},z).
\end{align*}
Following these abbreviations, we have that 
\[J_{u}^{t-1}(C,K,l)=\sbs{z}{L'}^{\sharp\mathsf{p}}_{u}
(\sbs{z}{L^{\star\, z}}^{\sharp\mathsf{p}}_{u}
(L_{u}(C\cup\delta^{t,\{z\}},K^{t-1},l))).
\]

In order to prove that $(P,u)\notin\mathrm{Min}$ implies that $P\in
J_{u}^{t-1}(C,K,l)$, we will need two auxiliary claims. For the proof of these claims, we will use the lemmas of Subsection~\ref{SSI} applied to the free $\Sigma$-algebra $\mathbf{T}_{\Sigma}(Y)$ on the $S$-sorted set $Y=X\cup C \cup \delta^{t,\{z\}}$, with the
homomorphism $h^{\sharp}\!\!\upharpoonright_{\mathrm{T}_{\Sigma}(Y)}$ from
$\mathbf{T}_{\Sigma}(Y)$ to $\mathbf{N}$ and the variable $z\in Y_{t}$,
for which $h^{\sharp}_{t}(z)=z$; note that
$\mathrm{T}_{\Sigma}(X\cup C)\subseteq\mathrm{T}_{\Sigma}(Y)$. For
legibility we keep writing $h^{\sharp}$ for this restriction, which is
harmless by Remark~\ref{RRest}.

\begin{clm}\label{C5} 
Let $v\in S$ and let
$T\in\mathrm{T}_{\Sigma}(X\cup C)_{v}$ be such that
$(T,v)\notin\mathrm{Min}$ and every proper non-minimal subterm of $(T,v)$
is a proper non-minimal subterm of $(P,u)$. Let $V$
and $(Q^{z}_{\alpha})_{\alpha\in\bb{V}_{z}}$ be provided by
Lemma~\ref{LCollapse} applied to $T$. Then $$V\in
L_{v}(C\cup\delta^{t,\{z\}},K^{t-1},h^{\sharp}_{v}(T)).$$
\end{clm}
\begin{proof}
The conditions for $V$ being a term in $L_{v}(C\cup\delta^{t,\{z\}},K^{t-1},h^{\sharp}_{v}(T))$ hold:
\begin{enumerate}
 \item[(a)] $V\in\mathrm{T}_{\Sigma}(X\cup C)_{v}$ by
construction. Thus, $V\in\mathrm{T}_{\Sigma}(X\cup C\cup\delta^{t,\{z\}})_{v}$;
\item[(b)] let $(W,w)<(V,v)$ with
$(W,w)\notin\mathrm{Min}$. By condition~(4b) of
Lemma~\ref{LCollapse}, there exists a non-minimal $(W',w)<(T,v)$
with $h^{\sharp}_{w}(W')=h^{\sharp}_{w}(W)$; by the hypothesis on
$T$, $(W',w)$ is a proper non-minimal subterm of $(P,u)$ and,
since $P\in L_{u}(C,K,l)$, we obtain $h^{\sharp}_{w}(W)\in K_{w}$. If
$w\neq t$, then $K^{t-1}_{w}=K_{w}$ and we are done. If $w=t$, then,
by condition~(4a) of Lemma~\ref{LCollapse},
$h^{\sharp}_{t}(W)\neq z=k_{t}-1$, so that
$h^{\sharp}_{t}(W)\in k_{t}-1=K^{t-1}_{t}$; and
\item[(c)] $h^{\sharp}_{v}(V)=h^{\sharp}_{v}(T)$ by condition~(2) of
Lemma~\ref{LCollapse}.
\end{enumerate}

This proves Claim~\ref{C5}.
\end{proof}

Note that Claim~\ref{C5} can be applied to the case $v=u$ and $T=P$, and to every non-minimal proper subterm $T$ of $P$. Now we introduce the second auxiliary claim.

\begin{clm}\label{C6} For every $Q\in\mathrm{T}_{\Sigma}(X\cup C)_{t}$
with $(Q,t)<(P,u)$ and $h^{\sharp}_{t}(Q)=z$, the following statement holds
$$
Q\in\sbs{z}{L'}^{\sharp\mathsf{p}}_{t}(L^{\star\, z}).
$$
\end{clm}
\begin{proof}
The proof is done by Artinian induction on $(\coprod\mathrm{T}_{\Sigma}(Z),\leq_{\mathbf{T}_{\Sigma}(Z)})$, where, we recall, $Z=X\cup N$ and $C\leq N$.

\textsf{Base step of the Artinian induction.}

Let $(Q,t)$ be a minimal element in $(\coprod\mathrm{T}_{\Sigma}(Z),\leq_{\mathbf{T}_{\Sigma}(Z)})$ with $(Q,t)<(P,u)$ and
$h^{\sharp}_{t}(Q)=z$. Then the defining conditions for $Q$ being a term in $L'$ hold:
\begin{enumerate}
 \item[(a)] Since $P\in\mathrm{T}_{\Sigma}(X\cup C)_{u}$ and $(Q,t)<(P,u)$, then $Q\in\mathrm{T}_{\Sigma}(X\cup C)_{t}$;
 \item[(b)] holds vacuously since $Q$ has no proper subterms; and
 \item[(c)] $h^{\sharp}_{t}(Q)=z$ by assumption.
\end{enumerate}

Thus, $Q\in L'$. Moreover, since $z\in L^{0\, z}\subseteq L^{\star\, z}$
and $\sbs{z}{L'}^{\sharp}_{t}(z)=L'$, we
conclude that
$
Q\in L'=\sbs{z}{L'}^{\sharp}_{t}(z)\subseteq
\sbs{z}{L'}^{\sharp\mathsf{p}}_{t}(L^{\star\, z}).
$

\textsf{Inductive step of the Artinian induction.}

Let $(Q,t)$ be a non-minimal element in $(\coprod\mathrm{T}_{\Sigma}(Z),\leq_{\mathbf{T}_{\Sigma}(Z)})$ with $(Q,t)<(P,u)$ and
$h^{\sharp}_{t}(Q)=z$. Assume that the statement holds for every
$Q'\in\mathrm{T}_{\Sigma}(Z)_{t}$ with $(Q',t)<(Q,t)$ and
$h^{\sharp}_{t}(Q')=z$.

Since $(Q,t)\notin\mathrm{Min}$, we can apply Lemma~\ref{LCollapse}, in
$\mathbf{T}_{\Sigma}(Y)$ and with respect to $h^{\sharp}$ and $z$, to the
term $Q$, obtaining a term $V\in\mathrm{T}_{\Sigma}(Y)_{t}$ and a family
$(Q'^{z}_{\alpha})_{\alpha\in\bb{V}_{z}}$ satisfying conditions~(1)--(4) of
that lemma. Since every proper non-minimal subterm of $(Q,t)$ is a proper
non-minimal subterm of $(P,u)$, Claim~\ref{C5} applies to $T=Q$ and yields
$V\in L_{t}(C\cup\delta^{t,\{z\}},K^{t-1},z)=L$.

On the other hand, by condition~(3) of Lemma~\ref{LCollapse}, every
member $Q'^{z}_{\alpha}$ of the family satisfies $(Q'^{z}_{\alpha},t)<(Q,t)$ and
$h^{\sharp}_{t}(Q'^{z}_{\alpha})=z$; hence, by the induction hypothesis,
$Q'^{z}_{\alpha}\in\sbs{z}{L'}^{\sharp\mathsf{p}}_{t}(L^{\star\, z})$, for every
$\alpha\in\bb{V}_{z}$. Therefore, by condition~(1) of
Lemma~\ref{LCollapse} and Lemma~\ref{LSubFam},
$$
Q=\sbs{z}{(Q'^{z}_{\alpha})_{\alpha\in\bb{V}_{z}}}(V)\in
\sbs{z}{\sbs{z}{L'}^{\sharp\mathsf{p}}_{t}(L^{\star\, z})}
^{\sharp}_{t}(V)\subseteq
\sbs{z}{\sbs{z}{L'}^{\sharp\mathsf{p}}_{t}(L^{\star\, z})}
^{\sharp\mathsf{p}}_{t}(L),
$$
the last inclusion because $V\in L$. Applying Lemma~\ref{LSubComp}, then
Lemma~\ref{LItAbs} together with Remark~\ref{RMon}, we obtain
$$
Q\in
\sbs{z}{L'}^{\sharp\mathsf{p}}_{t}\left(
\sbs{z}{L^{\star\, z}}^{\sharp\mathsf{p}}_{t}(L)\right)
\subseteq
\sbs{z}{L'}^{\sharp\mathsf{p}}_{t}(L^{\star\, z}).
$$

This proves Claim~\ref{C6}.
\end{proof}

We are now in position to prove that $(P,u)\notin\mathrm{Min}$ implies that $P\in
J_{u}^{t-1}(C,K,l)$.

Since
$(P,u)\notin\mathrm{Min}$, we may apply Lemma~\ref{LCollapse}, in
$\mathbf{T}_{\Sigma}(Y)$ and with respect to $h^{\sharp}$ and $z$, to the
term $P$, obtaining a $W\in\mathrm{T}_{\Sigma}(Y)_{u}$ and a family
$(Q^{z}_{\alpha})_{\alpha\in\bb{W}_{z}}$ such that
$P=\sbs{z}{(Q^{z}_{\alpha})_{\alpha\in\bb{W}_{z}}}(W)$ and, for every
$\alpha\in\bb{W}_{z}$, $(Q^{z}_{\alpha},t)<(P,u)$ and
$h^{\sharp}_{t}(Q^{z}_{\alpha})=z$. Claim~\ref{C5}, applied to $T=P$, yields
$W\in L_{u}(C\cup\delta^{t,\{z\}},K^{t-1},l)$, taking into account that
$h^{\sharp}_{u}(P)=l$. Moreover, by Claim~\ref{C6}, for every $\alpha\in\bb{W}_{z}$ we have
$Q^{z}_{\alpha}\in\sbs{z}{L'}^{\sharp\mathsf{p}}_{t}(L^{\star\, z})$. Hence, by
Lemma~\ref{LSubFam},
$$
\begin{adjustbox}{max width=\linewidth, keepaspectratio}
$
P=\sbs{z}{(Q^{z}_{\alpha})_{\alpha\in\bb{W}_{z}}}(W)\in
\sbs{z}{\sbs{z}{L'}^{\sharp\mathsf{p}}_{t}(L^{\star\, z})}
^{\sharp}_{u}(W)\subseteq
\sbs{z}{\sbs{z}{L'}^{\sharp\mathsf{p}}_{t}(L^{\star\, z})}
^{\sharp\mathsf{p}}_{u}
\left(L_{u}(C\cup\delta^{t,\{z\}},K^{t-1},l)\right)
$
\end{adjustbox},
$$
and, applying Lemma~\ref{LSubComp}, then Lemma~\ref{LItAbs} together with
Remark~\ref{RMon},
$$
P\in
\sbs{z}{L'}^{\sharp\mathsf{p}}_{u}\left(
\sbs{z}{L^{\star\, z}}^{\sharp\mathsf{p}}_{u}
\left(L_{u}(C\cup\delta^{t,\{z\}},K^{t-1},l)\right)\right)
=J^{t-1}_{u}(C,K,l),
$$
recalling that $z=k_{t}-1$. Therefore,  Statement~\eqref{EqP} holds.

This completes the proof of Claim~\ref{C4}.
\end{proof}

Thus, the inclusion from left to right of Equation~\eqref{EqE} holds. 
All in all, we conclude that 
Equation~\eqref{EqE} holds. This completes the proof of Main Claim~\ref{CMain}.
\end{proof}

We can now finish the proof of Proposition~\ref{PRectoReg}. By Main Claim~\ref{CMain}, let us consider, for every $l\in M$, the regular expression $R_{s}(\varnothing^{S},N,l)$. Since $M$ is finite and nonempty, the regular expression
$
\textstyle
R=\sum_{l\in M}R_{s}(\varnothing^{S},N,l)
\in\mathrm{T}_{\mathrm{Reg}(S,\Sigma,Z)}(Z)_{s}
$ 
satisfies that
$$
\textstyle
\{R\}^{Z\sharp}_{s}
=\bigcup_{l\in M}\{R_{s}(\varnothing^{S},N,l)\}^{Z\sharp}_{s}
=\bigcup_{l\in M}L_{s}(\varnothing^{S},N,l)=L.
$$
Since $Z$ is a finite $S$-sorted set with $X\subseteq Z$, by
Definition~\ref{DReg} we conclude that $L$ is $s$-regular, i.e., that
$L\in\mathrm{Reg}_{s}(\mathbf{T}_{\Sigma}(X))$. 
\end{proof}

\begin{rem}
The correspondence with the proof of Lemma~2.5.7 in~\cite{GS84} is as
follows. The languages $L_{u}(C,K,l)$ are the many-sorted counterparts of
the sets $T(K,h,i)$ of~\cite{GS84}, the size $\bb{\bb{K}}$ of the
bound playing the role of the number $h$, and the Equality~\eqref{EqL} corresponds
to the equality $T(\mathbf{A})=\bigcup(T(\varnothing,k,i)\mid i\in A')$.
Equation~\eqref{EqE} is the many-sorted counterpart of the equation
$$
\begin{adjustbox}{max width=\linewidth, keepaspectratio}
$
T(K,h+1,i)=T(K,h,i)\cup
T(K,h,h+1)\cdot_{h+1}T(K\cup h+1,h,h+1)^{\ast h+1}
\cdot_{h+1}T(K\cup h+1,h,i)
$
\end{adjustbox}
$$
of~\cite{GS84}, with the novelty that the state released at each
inductive step must be chosen sortwise, whence the union over the sorts
$t$ with $k_{t}\neq 0$. The different parts of the proof of Proposition~\ref{PRectoReg} develop in detail the
inclusion that in~\cite{GS84} is declared to be obvious from the
construction, while Claim~\ref{C4} corresponds to the factorization
$$
t\in\{p_{1},\ldots,p_{d}\}\cdot_{h+1}\{q_{11},\ldots,q_{1e_{1}}\}
\cdot_{h+1}\cdots
\cdot_{h+1}\{q_{j1},\ldots,q_{je_{j}}\}
\cdot_{h+1}\{r\}.
$$ 

In the notation of Claim~\ref{C4}, the term
$W$ plays the role of the outermost factor $r$, the membership of the
terms $Q^{z}_{\alpha}$ in
$\sbs{z}{L'}^{\sharp\mathsf{p}}_{t}(L^{\star\, z})$, established in Claim~\ref{C6}
by induction on the Artinian order, accounts for the intermediate factors
$q_{11},\ldots,q_{je_{j}}$, and the innermost factors
$p_{1},\ldots,p_{d}$ correspond to the terms of $L'$ substituted at the
last stage.
\end{rem}

Combining Corollary~\ref{CRegtoRec} with Proposition~\ref{PRectoReg} one
obtains the Kleene theorem for free many-sorted algebras.

\begin{thm}[Free many-sorted Kleene]\label{TKleene} For every $s\in S$, 
$\mathrm{Rec}_{s}(\mathbf{T}_{\Sigma}(X))=\mathrm{Reg}_{s}(\mathbf{T}_{\Sigma}(X))$.
\end{thm}

We illustrate the theorem and the role of the sort discipline with a small example.

\begin{exa}\label{EList}
Take as set of sorts $S=\{n,l\}$, intended to denote numerals and lists of numerals, and the signature $\Sigma$ with a constant symbol denoting the numeral zero
$\mathtt{z}\in\Sigma_{\lambda,n}$, a unary successor operation on numerals
$\mathtt{s}\in\Sigma_{n,n}$, a constant symbol denoting the empty list
$\mathtt{nil}\in\Sigma_{\lambda,l}$, and a constructor operation that takes a numeral and a list of numerals and returns a list of numerals,
$\mathtt{cons}\in\Sigma_{nl,l}$.
Taking $X=\varnothing^{S}$, the free algebra $\mathbf{T}_{\Sigma}(\varnothing^{S})$ has, at sort $n$, the unary numerals of the form $\mathtt{s}^{k}(\mathtt{z})$, for $k\in \mathbb{N}$, and, at sort $l$, the finite lists
$\mathtt{cons}(t_{0},\mathtt{cons}(t_{1},\ldots,\mathtt{cons}(t_{m-1},\mathtt{nil})\ldots))$, with $m\in\mathbb{N}$
and each $t_{i}$ a numeral, for every $i\in m$. Both sorts are inhabited, and $\mathtt{cons}$ binds the two together: a term of sort $l$ different from $\mathtt{nil}$ cannot be built without terms of sort $n$.

Consider the language $L\subseteq \mathrm{T}_{\Sigma}(\varnothing^{S})_{l}$ of those lists all of whose entries are even, that is, of the form $\mathtt{s}^{2k}(\mathtt{z})$, for some $k\in \mathbb{N}$. The defining condition lives at sort $n$, but the language itself lives at sort $l$; recognizing $L$ therefore forces a property at one sort to be tracked across $\mathtt{cons}$ into another.

The language $L$ is $l$-recognizable. Indeed, let $\mathbf{A}$ be the $\Sigma$-algebra with $A_{n}=\{\mathsf{even},\mathsf{odd}\}$ and $A_{l}=\{\top,\bot\}$, and operations
$\mathtt{z}^{\mathbf{A}}=\mathsf{even}$,
$\mathtt{s}^{\mathbf{A}}(\mathsf{even})=\mathsf{odd}$, $\mathtt{s}^{\mathbf{A}}(\mathsf{odd})=\mathsf{even}$,
$\mathtt{nil}^{\mathbf{A}}=\top$, and
$\mathtt{cons}^{\mathbf{A}}(\mathsf{even},\top)=\top$, with $\mathtt{cons}^{\mathbf{A}}$ returning $\bot$ on every other pair. Then $\mathbf{A}$ is finite at both sorts, and the unique homomorphism $f\colon\mathbf{T}_{\Sigma}(\varnothing^{S})\longrightarrow\mathbf{A}$ satisfies $f_{l}^{-1}[\{\top\}]=L$. Hence $L\in\mathrm{Rec}_{l}(\mathbf{T}_{\Sigma}(\varnothing^{S}))$.

The language $L$ is $l$-regular. Indeed, by Theorem~\ref{TKleene} this is automatic, but it is instructive to exhibit the expression and see the sortwise operators at work. Consider the $S$-sorted set $Z$ of variables $y\in Z_{n}$ and $z\in Z_{l}$, so $\varnothing^{S}\subseteq Z$. The even numerals are denoted at sort $n$ by iterating $\mathtt{s}(\mathtt{s}(y))$ from $\mathtt{z}$,
\[
E \;=\; \sbs{y}{\mathtt{z}}^{\sharp\mathsf{p}}_{n}\bigl((\mathtt{s}(\mathtt{s}(y)))^{\star\,y}\bigr)
\;\subseteq\;\mathrm{T}_{\Sigma}(\varnothing^{S})_{n},
\]
and the even lists are obtained at sort $l$ by iterating $\mathtt{cons}(y,z)$ from $\mathtt{nil}$ and then substituting $E$ for $y$,
\[
L \;=\; \sbs{z}{\mathtt{nil}}^{\sharp\mathsf{p}}_{l}\Bigl(\sbs{y}{E}^{\sharp\mathsf{p}}_{l}\bigl((\mathtt{cons}(y,z))^{\star\,z}\bigr)\Bigr).
\]
The expression for $L$ at sort $l$ calls the expression $E$ at sort $n$ through the $y$-substitution that plants $E$ into the first argument of $\mathtt{cons}$. This is exactly the sort-indexed interplay recorded by the  $Z$-sorted set $K$ and by the union over sorts in Equation~\eqref{EqE}: the $l$-component is built only after the $n$-component has been settled, and the two iteration steps are charged to different components of $K$.
\end{exa}

\noindent\textbf{Acknowledgements.} The second and fourth authors were supported by the PID2024\allowbreak-159495NB\allowbreak-I00 grant from the Ministerio de Ciencia e Innovaci\'{o}n, Spain, and the CIAICO/\allowbreak2023/\allowbreak007 grant from the Conselleria d'Educaci\'{o}, Universitats i Ocupaci\'{o}, Generalitat Valenciana. The second author was supported by the grant CIACIF/\allowbreak2022/\allowbreak489 from the Conselleria d'Educaci\'{o}, Universitats i Ocupaci\'{o}, Generalitat Valenciana co-funded by the European Social Fund.  The fourth author has held a Specially Appointed Professor position at Nantong University during the completion of this work. 

We express our gratitude to Juan Climent Vidal for his mentorship and enduring friendship. Special thanks to Asunci\'{o}n de Montesa for her categorical knowledge. 

\bibliographystyle{alphaurl}
\bibliography{MSKleene}

\end{document}